\documentclass[letterpaper,11pt]{article}
\usepackage[margin=1in]{geometry}
\usepackage[bookmarks, colorlinks=true, plainpages = false, colorlinks=true,
            citecolor=red,
            linkcolor=blue,
            anchorcolor=red,
            urlcolor=blue]{hyperref}
\usepackage{url}
\usepackage{amsmath,amsfonts,amsthm,amssymb,bm, verbatim,dsfont,mathtools}
\usepackage{color,graphicx,appendix}
\usepackage{etoolbox}

\usepackage{array}
\usepackage{multirow}

\usepackage{float}
\usepackage{tikz}
\usetikzlibrary{matrix,arrows,calc,shapes,backgrounds}
\usetikzlibrary{shapes.callouts,decorations.text}

\usetikzlibrary{shapes.misc}

\tikzset{cross/.style={cross out, draw=black, minimum size=2*(#1-\pgflinewidth), inner sep=0pt, outer sep=0pt},
cross/.default={1pt}}

\tikzstyle{int}=[draw, fill=blue!20, minimum size=2em]
\tikzstyle{dot}=[circle, draw, fill=blue!20, minimum size=2em]
\tikzstyle{dotred}=[circle, draw, fill=red!20, minimum size=2em]
\tikzstyle{init} = [pin edge={to-,thin,black}]
\tikzstyle{initred} = [pin edge={to-,thin,red}]
\tikzstyle{plan}=[draw, fill=blue!20, minimum size=2em, text width=5em, rounded corners,align=center]
\tikzstyle{planwide}=[draw, fill=blue!20, minimum size=2em, text width=8em, rounded corners,align=center]
\tikzstyle{nodedot}=[circle, draw, fill=white, minimum size=0.3cm,inner sep=0pt]
\tikzstyle{nodedot}=[circle, draw, fill=white, minimum size=3,inner sep=0pt]
\tikzstyle{Medge}=[green!60!black, thick]
\tikzstyle{Bedge}=[red, thick]
\tikzstyle{Cedge}=[blue, thick]
\tikzstyle{Sedge}=[black, thick]
\tikzstyle{Mgiantedge}=[green!60!black, line width=3.0pt]
\tikzstyle{Bgiantedge}=[red,line width=3.0pt]
\tikzstyle{Cgiantedge}=[blue,line width=3.0pt]
\tikzstyle{Sgiantedge}=[black,line width=3.0pt]
\tikzstyle{shadedgiantnode}=[circle, draw, fill=black!10, minimum size=1cm, inner sep=0pt]
\tikzstyle{unshadedgiantnode}=[circle, draw, fill=white, minimum size=1cm, inner sep=0pt]
\makeatletter
\tikzset{my loop/.style =  {to path={
  \pgfextra{}
  [looseness=5,min distance=10mm]
  \tikz@to@curve@path},font=\sffamily\small
  }}  
\makeatletter 
\newcolumntype{C}[1]{>{\centering\arraybackslash}p{#1}}

\tikzstyle{vertexdot}=[circle, draw, fill=white, minimum size=3,inner sep=0pt]
\tikzstyle{root}=[circle, draw, fill=black, minimum size=3,inner sep=0pt]

\tikzstyle{vertexdotsolid}=[circle, draw, fill=black, minimum size=3,inner sep=0pt]

\usepackage{pgfplots}

\pgfplotsset{
    standard/.style={
        axis x line=middle,
        axis y line=middle,
        every axis x label/.style={at={(current axis.right of origin)},anchor=north west},
        every axis y label/.style={at={(current axis.above origin)},anchor=north west}
    }
}
\usepgfplotslibrary{fillbetween}

\usepackage{xr,xspace}
\usepackage{todonotes}
\usepackage{paralist}
\usepackage{enumitem}

\usepackage{caption,subcaption,soul}
\usepackage{algorithm}
\usepackage{algpseudocode}
\makeatletter
\usepackage{romanbar}

\theoremstyle{plain}
\newtheorem{theorem}{Theorem}
\newtheorem{lemma}{Lemma}
\newtheorem{proposition}{Proposition}

\theoremstyle{definition}
\newtheorem{definition}{Definition}

\newtheorem{problem}{Problem}

\newtheorem{remark}{Remark}
\newtheorem{claim}{Claim}
\newtheorem*{remark*}{Remark}
\newtheorem*{theorem*}{Theorem}

\newcommand{\floor}[1]{\left\lfloor #1 \right\rfloor}
\newcommand{\ceil}[1]{\left\lceil #1 \right\rceil}

\newcommand \E[1]{\mathbb{E}[#1]}
\newcommand{\Universal}{\beta}

\usepackage{xspace,prettyref}

\newcommand{\id}{\mathrm{id}}

\newcommand{\diverge}{\to\infty}

\newcommand{\iiddistr}{{\stackrel{\text{\iid}}{\sim}}}

\newcommand{\naturals}{{\mathbb{N}}}

\newcommand{\Expect}{\mathbb{E}}
\newcommand{\expect}[1]{\mathbb{E}\left[ #1 \right]}

\newcommand{\prob}[1]{ \mathbb{P}\left\{ #1 \right\} }

\newcommand{\Cov}{\mathrm{Cov}}

\def\Var{\mathrm{Var}}
\def\E{\mathbb{E}}

\newcommand{\Binom}{{\rm Binom}}

\newcommand{\ie}{i.e.\xspace}
\newcommand{\iid}{i.i.d.\xspace}
\newrefformat{eq}{(\ref{#1})}
\newrefformat{chap}{Chapter~\ref{#1}}
\newrefformat{sec}{Section~\ref{#1}}
\newrefformat{alg}{Algorithm~\ref{#1}}
\newrefformat{fig}{Figure~\ref{#1}}
\newrefformat{tab}{Table~\ref{#1}}
\newrefformat{rmk}{Remark~\ref{#1}}
\newrefformat{clm}{Claim~\ref{#1}}
\newrefformat{claim}{Claim~\ref{#1}}
\newrefformat{def}{Definition~\ref{#1}}
\newrefformat{cor}{Corollary~\ref{#1}}
\newrefformat{lmm}{Lemma~\ref{#1}}
\newrefformat{prop}{Proposition~\ref{#1}}
\newrefformat{prob}{Problem~\ref{#1}}
\newrefformat{app}{Appendix~\ref{#1}}
\newrefformat{hyp}{Hypothesis~\ref{#1}}
\newrefformat{thm}{Theorem~\ref{#1}}

\newcommand{\pth}[1]{\left( #1 \right)}

\newcommand{\indc}[1]{{\mathbf{1}_{\left\{{#1}\right\}}}}

\newcommand{\sfL}{{\mathsf{L}}}
\newcommand{\sfM}{{\mathsf{M}}}
\newcommand{\sfN}{{\mathsf{N}}}

\newcommand{\sfT}{{\mathsf{T}}}

\newcommand{\calB}{{\mathcal{B}}}
\newcommand{\calC}{{\mathcal{C}}}

\newcommand{\calF}{{\mathcal{F}}}
\newcommand{\calG}{{\mathcal{G}}}
\newcommand{\calH}{{\mathcal{H}}}

\newcommand{\calJ}{{\mathcal{J}}}
\newcommand{\calK}{{\mathcal{K}}}

\newcommand{\calM}{{\mathcal{M}}}
\newcommand{\calN}{{\mathcal{N}}}
\newcommand{\calO}{{\mathcal{O}}}

\newcommand{\calR}{{\mathcal{R}}}

\newcommand{\calT}{{\mathcal{T}}}
\newcommand{\calU}{{\mathcal{U}}}

\newcommand{\calW}{{\mathcal{W}}}

\newcommand{\ER}{Erd\H{o}s--R\'enyi\xspace}

\renewcommand{\tilde}{\widetilde}
\renewcommand{\bar}{\overline}
\newcommand{\aut}{\mathsf{aut}}
\newcommand{\sub}{\mathsf{sub}}

\newcommand{\Vprime}{\UTone'(i)}
\newcommand{\Vdoubleprime}{\UTone''(i)}
\newcommand{\Wprime}{\UTone'(j)}
\newcommand{\Wdoubleprime}{\UTone''(j)}

\renewcommand{\SS}{\texttt{S}}
\newcommand{\TT}{\texttt{T}}

\newcommand{\R }{R}
\newcommand{\B}{\calB}
\newcommand{\complete}{\mathbb{K}_n}

\newcommand{\UTone}{U_{\sfL}}
\newcommand{\UTtwo}{U_{\sfM}}
\newcommand{\UN}{U_{\sfN}}
\newcommand{\UT}{U_{\sfT}}
\newcommand{\UNone}{U_{\sfN_1}}
\newcommand{\UNtwo}{U_{\sfN_2}}

\newcommand{\PTone}{P_{\sfL}}
\newcommand{\PTtwo}{P_{\sfM}}
\newcommand{\PN}{P_{\sfN}}

\newcommand{\bTtwo}{b_{\sfM}}

\newcommand{\calUTone}{\calU_{\sfL}}
\newcommand{\calUTtwo}{\calU_{\sfM}}
\newcommand{\calUN}{\calU_{\sfN}}

\newcommand{\calNTone}{\calN_{\sfL}}
\newcommand{\calNTtwo}{\calN_{\sfM}}
\newcommand{\calNT}{\calN_{\sfT}}

\newcommand{\calNN}{\calN_{\sfN}}
\newcommand{\calNNone}{\calN_{\sfN_1}}
\newcommand{\calNNtwo}{\calN_{\sfN_2}}

\newcommand{\stepa}[1]{\overset{\rm (a)}{#1}}
\newcommand{\stepb}[1]{\overset{\rm (b)}{#1}}
\newcommand{\stepc}[1]{\overset{\rm (c)}{#1}}

\begin{document}
\title{Random graph matching at Otter's threshold\\ via counting chandeliers
}

\author{Cheng Mao, Yihong Wu,  Jiaming Xu, and Sophie H.\ Yu\thanks{
C.\ Mao is with the School of Mathematics, Georgia Institute of Technology, Atlanta, Georgia, USA
\texttt{cheng.mao@math.gatech.edu}.
Y.\ Wu is  with the Department of Statistics and Data Science, Yale University, New Haven CT, USA, 
\texttt{yihong.wu@yale.edu}.
J.\ Xu and S.\ H.\ Yu are with The Fuqua School of Business, Duke University, Durham NC, USA, \texttt{\{jx77,haoyang.yu\}@duke.edu}.
C.~Mao is supported in part by the NSF Grants DMS-2053333 and DMS-2210734. Y.~Wu is supported in part by the NSF Grant CCF-1900507, an NSF CAREER award CCF-1651588, and an Alfred Sloan fellowship. J. Xu is supported in part by the NSF Grant CCF-1856424
and an NSF CAREER award CCF-2144593. S. H.~Yu is supported by the NSF Grant CCF-1856424.
}}

\date{\today}
\maketitle

\begin{abstract}

We propose an efficient algorithm for graph matching based on similarity scores constructed from counting a certain family of weighted trees rooted at each vertex.
For two \ER graphs $\calG(n,q)$ whose edges are correlated through a latent vertex correspondence, we show that this algorithm correctly matches all but a vanishing fraction of the vertices with high probability, provided that $nq\to\infty$ and the edge correlation coefficient $\rho$ satisfies  $\rho^2>\alpha \approx 0.338$, where $\alpha$ is Otter's tree-counting constant. Moreover, this almost exact matching can be made exact under an extra condition that is information-theoretically necessary. This is the first polynomial-time graph matching algorithm that succeeds at an explicit constant correlation and applies to both sparse and dense graphs. In comparison, previous methods either require $\rho=1-o(1)$ or are restricted to sparse graphs. 

The crux of the algorithm is a carefully curated family of rooted trees called \textit{chandeliers}, which
allows effective extraction of the graph correlation from the counts of the same tree while suppressing the undesirable correlation between those of different trees.

\end{abstract}

\tableofcontents
\section{Introduction}
\label{sec:intro}

 Graph matching (also known as network alignment) refers to the problem of finding the bijection between the vertex sets of the two graphs that maximizes the total number of common edges. When the two graphs are exactly isomorphic to each other, this reduces to the classical graph isomorphism problem, for which the best known algorithm runs in quasi-polynomial time~\cite{Babai2016}. In general, graph matching is an instance of the \textit{quadratic assignment problem}~\cite{burkard1998quadratic},
which is known to be NP-hard to solve or even approximate~\cite{makarychev2010maximum}.

Motivated by real-world applications (such as social network de-anonymization \cite{narayanan2008robust} and computational biology \cite{singh2008global})
as well as the need to understand the average-case computational complexity, a recent line of work is devoted to the study of theory and algorithms for graph matching
under statistical models, by assuming the two graphs are randomly generated with correlated edges under a hidden vertex correspondence. 
A canonical model is the following
\emph{correlated \ER graph model}~\cite{pedarsani2011privacy}.

\begin{definition}[Correlated \ER graph model] 
\label{def:er-model}
Let $\pi$ denote a latent permutation on $[n] \triangleq \{1,\dots,n\}$. We generate two random graphs
on the common vertex set $[n]$
with adjacency matrices $A$ and $B$ such that  $(A_{ij}, B_{\pi(i)\pi(j)})$ are i.i.d.\ pairs of Bernoulli random variables with
mean $q \in [0,1]$ and correlation coefficient $\rho$ 
for $1 \leq i<j\leq n$.
We write $(A,B) \sim \calG(n,q,\rho)$. 
\end{definition}

Given $(A,B) \sim \calG(n,q,\rho)$, our goal is to recover the latent vertex correspondence $\pi$.
 The information-theoretic thresholds for both exact and partial recovery have been derived \cite{cullina2016improved,cullina2017exact,Hall2020partial,wu2021settling,ganassali2021impossibility,ding2022matching} and
various efficient matching algorithms have been developed with performance
guarantees  \cite{dai2019analysis,ding2021efficient,FMWX19a,FMWX19b,ganassali2020tree,ganassali2021correlation,mao2021random,mao2021exact}. 
Despite these exciting progresses, most existing efficient algorithms require the two graphs to be almost perfectly correlated; as such, the problem of polynomial-time recovery with a constant correlation remains largely unresolved except for sufficiently sparse graphs. 
Specifically, if the correlation $\rho$ is an (unspecified) constant sufficiently close to $1$, exact recovery is achievable in polynomial time for graphs whose average degrees satisfy $(1+\epsilon)\log n \le nq \le n^{\frac{1}{\Theta(\log \log n)}}$ \cite{mao2021exact}, while partial recovery is achievable for sparse graphs with $nq=O(1)$ \cite{ganassali2020tree,ganassali2021correlation}. 
For dense graphs, the best known result for polynomial-time recovery requires $\rho \ge 1 - (\log\log(n))^{-C}$ for some constant $C>0$ \cite{mao2021random}. 
The current paper significantly advances the state of the art by establishing the following results.


\begin{theorem*} 
Assume that $0 < q \le 1/2$ and 
$$
\rho^2 > \alpha \approx 0.338,
$$ 
where 
\[
\alpha=\lim_{K\to\infty} \frac{K}{\log (\text{number of unlabeled trees with $K$ edges})}
\]
is Otter's tree-counting constant \cite{otter1948number}.
Given a pair of correlated \ER graphs $(A,B) \sim \calG(n,q,\rho)$, the following holds:
\begin{itemize}
\item 
\emph{(Exact recovery)} 
If $\rho>0$ and $nq(q+\rho(1-q)) \ge (1+\epsilon) \log n$ for any constant $\epsilon > 0$,\footnote{The condition $nq(q+\rho(1-q)) \ge (1+\epsilon) \log n$ is information-theoretically necessary, for otherwise the intersection graph between $A$ and $B$ (under the vertex correspondence $\pi$) contains isolated vertices with high probability and exact recovery is impossible.}
there is a polynomial-time algorithm that recovers $\pi$ exactly with high probability.

\item
\emph{(Almost exact recovery)}
If $nq = \omega(1)$,  there is a polynomial-time algorithm that outputs a subset $I \subset [n]$ and a map $\hat \pi : I \to [n]$ such that $\hat \pi = \pi|_I$ and $|I| = (1-o(1)) n$ with high probability.

\item
\emph{(Partial recovery)}
For any constant $\delta \in (0,1)$, there is a constant $C(\rho,\delta) > 0$ depending only on $\rho$ and $\delta$ such that if $nq \ge C(\rho,\delta)$,  the above $I$ and $\hat \pi$ satisfy that $\hat \pi = \pi|_I$ with high probability and $\E[|I|] \ge (1-\delta) n$. 
\end{itemize}
\end{theorem*}


The above theorem identifies an explicit threshold $\rho^2 > \alpha$ that allows polynomial-time graph matching for both sparse and dense graphs. 
In certain regimes, the condition for exact recovery in this result is in fact \textit{optimal}, matching the information-theoretic threshold identified in \cite{cullina2017exact,wu2021settling} (see \prettyref{rmk:exactIT} and \prettyref{fig:GM_phase_new} for a detailed discussion).

\subsection{Key challenges and algorithmic innovations}

A principled approach to graph matching is the following three-step procedure:
\begin{enumerate}
    \item 
\textit{Signature embedding}: Associate to each vertex $i$ in $A$ a \emph{signature} $s_i$ and to each vertex $j$ in $B$ a signature $t_j$.

\item \textit{Similarity scoring}: Compute the similarity score $\Phi_{ij}$ based on $s_i$ and $t_j$ using a certain similarity measure on the signature space.

\item \textit{Linear assignment}: 
Solve  max-weight bipartite matching with weights $\Phi_{ij}$ either exactly or approximately (e.g.,~greedy algorithm).
\end{enumerate}
In this way, we reduce the problem from the NP-hard quadratic assignment to the  tractable linear assignment.
Clearly, the key to this approach is the construction of the similarity scores.

Many existing algorithms for graph matching largely follow this paradigm using similarity scores based on  neighborhood statistics
\cite{
bollobas1982distinguishing,czajka2008improved, ding2021efficient,dai2019analysis,mao2021random},
spectral methods
\cite{umeyama1988eigendecomposition,singh2008global,FMWX19b}, or convex relaxations \cite{zaslavskiy2008path,aflalo2015convex,lyzinski2016graph}.
In terms of theoretical guarantees, these methods either require extremely high correlation or are tailored to sparse graphs.
Note that two $\rho$-correlated \ER graphs differ by $\Theta(1-\rho)$ fraction of edges. Thus, to succeed at a constant $\rho$ bounded away from $1$, 
the similarity scores need to be robust to perturbing a constant fraction of edges. 
All existing  algorithms~\cite{mao2021exact,ganassali2020tree,ganassali2021correlation} achieving this goal 
crucially exploit the tree structure of local neighborhoods and are thus restricted to sparse graphs. On the other hand, algorithms that apply to both sparse and dense graphs~\cite{ding2021efficient,FMWX19b,mao2021random} so far can  only tolerate a vanishing fraction of edge perturbation and thus all require $\rho=1-o(1)$.


 The major algorithmic innovation of this work is a new construction based on \textit{subgraph counts}.
 Specifically, the signature assigned to a node $i$ is a vector indexed by a family of non-isomorphic subgraphs, where each entry records the total number of subgraphs rooted at $i$ that appear in the graph weighted by the centered adjacency matrix, known as the signed graph count \cite{bubeck2016testing} (cf.~\prettyref{eq:W_i_H} and \prettyref{eq:def-sig-vec-w-h-i} for the formal definition). 
The similarity score for each pair of vertices is the weighted inner product between their signatures. 
The key to executing this strategy is a carefully curated family of trees called \emph{chandeliers}, which, as we explain next, allows one to extract the graph correlation from the counts of the same tree while suppressing the undesirable correlation between those of different trees.
This leads to a robust construction of signatures
that can withstand perturbing a constant fraction of edges, 
without relying on the locally tree-like property that limits the previous methods to sparse graphs.

Counting subgraphs is a popular method for network analysis in both theory \cite{mossel2015reconstruction,bubeck2016testing} and practice \cite{milo2002network,alon2008biomolecular,ribeiro2021survey}.
We refer to \cite[Sec.~2.4]{mao2021testing} for a comprehensive overview of hypothesis testing and estimation based on subgraph counting for networks with latent structures. Notably, most of these previous works focus on counting cycles. 
%
However, here in order to succeed at a constant $\rho$,
we need to count a sufficiently rich class of subgraphs (whose cardinality grows at least exponentially with the number of edges)\footnote{A high-level explanation is as follows. For a single subgraph $H$ with $N$ edges, the correlation between the subgraph counts of $H$ rooted at vertex $i$
across $A$ and $B$ -- the \emph{signal}, is smaller than their variances by a multiplicative factor of $\rho^{N}$.  Therefore, to pick up the signal, we need
to further average over a  family $\calH$ of such subgraphs so that $|\calH| \rho^{2N} \to \infty$ (cf.~\prettyref{eq:fake_lowerbound} for a more detailed explanation).} and cycles clearly fall short of this basic requirement. A much richer family of \emph{strictly balanced, asymmetric} subgraphs is considered in~\cite{barak2019nearly}, where the edge density of the subgraphs is carefully chosen so that typically they co-occur in both graphs at most once.
Hence, by searching for such rare  subgraphs, dubbed ``black swans'', one can match the corresponding vertices. Although this method succeeds even for vanishing correlation $\rho \ge (\log n)^{-o(1)}$,  it has a quasi-polynomial time complexity $n^{\Theta(\log n)}$ due to the exhaustive search of subgraphs of size $\Theta(\log n)$. Moreover, the construction of this special family of subgraphs requires the average degree 
$nq$ to fall into a
very specific range of $[n^{o(1)},n^{1/153}]\cup[n^{2/3},n^{1-\epsilon}]$ for some arbitrarily small constant $\epsilon>0$ and in particular does not accommodate relatively sparse graphs such as $nq=O(\log n)$.

As opposed to relying on rare subgraphs,
our approach is to count a family of unlabeled rooted trees with size $N=\Theta(\log n)$, which are abundant even in very sparse graphs. 
Moreover, by leveraging the method of \emph{color coding}~\cite{alon1995color,arvind2002approximation,hopkins2017bayesian}, such trees can be counted approximately but sufficiently accurately in polynomial time. While centering the adjacency matrices and counting signed trees are helpful, 
there still remains excessive correlation among  different trees counts which is hard to control -- this is the key difficulty in analyzing signatures based on subgraph counts.
To resolve this challenge, we propose to count a special family $\calT$ of unlabeled rooted trees, which we call chandeliers; see~\prettyref{eq:W_i_H} for the formal definition. 
As discussed in \prettyref{sec:chandelier}, the chandelier structure plays a crucial role in curbing the undesired correlation between different tree counts. 
Moreover, even though chandeliers only occupy a vanishing fraction of all trees, by choosing the parameters appropriately, we can ensure that $|\calT|=(1/\alpha+o(1))^N$, which grows almost at the same rate as the entire family of trees.

A similar idea of counting signed but unrooted trees has been applied
in~\cite{mao2021testing} for the graph correlation detection problem, \ie, testing whether the two graphs are independent \ER graphs or $\rho$-correlated through a latent vertex matching chosen uniformly at random. It is shown that the two hypotheses can be distinguished with high probability in polynomial time at the same threshold of $\rho^2>\alpha$. However, unlike the present paper, 
averaging over the random permutation dramatically simplifies the analysis of correlations between different tree counts. As a result, it suffices to simply count all trees as opposed to a carefully constructed collection of special trees. We refer to the last two paragraphs in~\prettyref{sec:chandelier} for a detailed comparison.

\subsection{Notation}

Given a graph $H$, let $V(H)$ denote its vertex set and $E(H)$ denote its edge set. 
Let $v(H) = |V(H)|$ and 
$e(H) = |E(H)|$.
We call $e(H)-v(H)$ the \textit{excess} of the graph $H$.
We denote by $\complete$ the complete graph with vertex set $[n]$ and edge set $\binom{[n]}{2} \triangleq \{\{u,v\}: u,v\in[n], \, u \ne v\}$. 
An empty graph is denoted as $\emptyset$, if it does not contain any vertex or edge. 
A \emph{rooted} graph is a graph in which one vertex has been distinguished as the root. An \textit{isomorphism} between two rooted graphs $H$ and $G$ is a bijection between the vertex sets that preserves both edges and the root, namely, $f: V(H) \to V(G)$ such that the root of $H$
is mapped to that of $G$ and any two vertices $u$ and $v$ are adjacent in $H$ if and only if $f(u)$ and $f(v)$
are adjacent in $G$.
An \emph{automorphism} of a rooted graph is an isomorphism to itself. Let $\aut(H)$ be the number of automorphisms of $H$. 
For a rooted tree $T$ and a vertex $a\in V(T)$, let $(T)_a$ denote the subtree of $T$ consisting of all descendants of $a$ and we set $(T)_a=\emptyset$ if $a \notin V(T)$. 

For two real numbers $x$ and $y$, we let $x \vee y \triangleq \max\{x, y\}$
and $x\wedge y \triangleq \min\{x, y\}$. 
We use standard asymptotic notation: for two positive sequences $\{x_n\}$ and $\{y_n\}$, we write $x_n = O(y_n)$ or $x_n \lesssim y_n$, if $x_n \le C y_n$ for an absolute constant $C$ and for all $n$; $x_n = \Omega(y_n)$ or $x_n \gtrsim y_n$, if $y_n = O(x_n)$; $x_n = \Theta(y_n)$ or $x_n \asymp y_n$, if $x_n = O(y_n)$ and $x_n = \Omega(y_n)$; 
$x_n = o(y_n)$ or $y_n = \omega(x_n)$, if $x_n / y_n \to 0$ as $n\diverge$. 

\subsection{Organization}
The rest of the paper is organized as follows.
In \prettyref{sec:scores-results}, we first introduce the similarity scores between vertices of the two graphs based on counting signed chandeliers, and then state our main results on the recovery of the latent vertex correspondence for correlated \ER graphs. 
In \prettyref{sec:chandelier}, we explain the rationale for focusing on the class of chandeliers. 
\prettyref{sec:statistical-analysis} provides a statistical analysis of the similarity scores, proving our results on partial and almost exact recovery stated in \prettyref{thm:Phi_ij_almost}. 
In particular, Propositions~\ref{prop:true_pair} and~\ref{prop:fake_pairs_improved} are the key ingredients controlling the variance of the similarity scores, and their proofs are given in \prettyref{sec:second-moment}, which constitutes the bulk of the paper.
In \prettyref{sec:color_coding}, we use the method of color coding to approximate the proposed similarity scores in polynomial time, and show that the same statistical guarantees continue to hold for the approximated scores, thereby proving \prettyref{thm:tilde_Phi_ij_almost}.
Finally, in \prettyref{sec:boost}, we demonstrate how to upgrade an almost exact matching to an exact matching,
establishing \prettyref{thm:exact_reovery}.
\prettyref{app:pre} consists of auxiliary results, and \prettyref{app:unknown_rho} discusses a data-driven way to choose a threshold parameter in our algorithm.

\section{Main results and discussions}
\label{sec:main-results}

\subsection{Similarity scores and statistical guarantees}
\label{sec:scores-results}
We start with some preliminary definitions before specializing to chandeliers.
For any weighted adjacency matrix $M$,  node $i\in [n]$, and rooted graph $H$, define the weighted subgraph count
\begin{align}
    W_{i,H} (M)  \triangleq \sum_{S(i) \cong H} M_S \,, \text{ where } M_S  \triangleq \prod_{e\in E(S)} M_e\, ,
    \label{eq:W_i_H}
\end{align}
and $S(i)$ denotes a subgraph of $\complete$ rooted at $i$.
(Whenever the context is clear, we also abbreviate $S(i)$ as 
$S$.) Note that when $M$ is the adjacency matrix $A$, $W_{i,H}$ reduces to the usual subgraph count, \ie, the number of subgraphs rooted at $i$ in $M$ that are isomorphic to $H$.  
When $M$ is a centered adjacency matrix $\bar A \triangleq A-q$, we call $W_{i,H}$ a \textit{signed} subgraph count  following~\cite{bubeck2016testing}.
For example (with solid vertex as the root),
$W_{i,\,\tikz[scale=0.8,baseline=(zero.base)]{\draw (0,0) node (zero) [vertexdotsolid] {} -- (0.25,0) node[vertexdot] {};}}(\bar A) = d_i-(n-1)q$
and
$W_{i,\,\tikz[scale=0.8,baseline=(zero.base)]{\draw (0,0) node (zero) [vertexdot] {} -- (0.25,0) node[vertexdotsolid] {} -- (0.5,0) node[vertexdot] {};}}(\bar A) = \binom{d_i}{2} - (n-2) d_i q + \binom{n-1}{2}q^2$,
where $d_i$ is the degree of $i$ in $A$.

Next, given a family $\calH$ of non-isomorphic rooted graphs $H$, the \emph{subgraph count signature} 
of a node $i$ is defined as the vector 
\begin{equation}
W_{i}^{\calH}(M) \triangleq \left( W_{i,H}(M) \right)_{H\in\calH}.
\label{eq:def-sig-vec-w-h-i}
\end{equation}
\prettyref{alg:GMCC} below describes our proposed method for graph matching based on subgraph count signatures.

\begin{algorithm}
\normalsize
\caption{Graph Matching by Counting Signed Graphs}\label{alg:GMCC}
\begin{algorithmic}[1]
\State {\bfseries Input:} Adjacency matrices $A$ and $B$ on $n$ vertices, a family $\calH$ of non-isomorphic rooted graphs, and a threshold $\tau>0$.
 \State  {\bfseries Output:} A mapping $\hat{\pi}:I \to [n]$.
\State For each pair of node $i$ in $A$ and node $j$ in $B$, compute their similarity score as the \emph{weighted}\footnotemark~inner product between their subgraph count signatures: 
\begin{align}
    \Phi_{ij}^{\calH}
    & \triangleq ~ 
    \left \langle W_{i}^{\calH} (\bar{A}) \,, W_{j}^{\calH}(\bar{B}) \right \rangle \triangleq ~ 
    \sum_{H \in\calH} \aut(H) \, W_{i,H} (\Bar{A}) \, W_{j,H} (\Bar{B}) \, ,  \label{eq:Phi_ij}
\end{align}
where $\Bar{A} = A - q$ and
$\Bar{B} = B - q$ are the centered adjacency matrices. 
\State For each $i\in [n]$, if there exists a unique  $j\in [n]$ such that $\Phi^{\calH}_{ij} \ge \tau$, 
let $\hat{\pi}(i)=j$ and include $i$ in set $I$.
\end{algorithmic}
\end{algorithm}
\footnotetext{Note that in \prettyref{eq:Phi_ij} the coefficient $\aut(H)$ accounts for the symmetry of $H$ and compensates for the fact that the number of copies of $H$ in the complete graph $\complete$ is inversely proportional to $\aut(H)$.
This simplifies the first moment calculation in \prettyref{prop:mean_Phi_ij}.}

At this point \prettyref{alg:GMCC} is a ``meta algorithm'' and the key to its application is to carefully choose this collection of subgraphs $\calH$. 
Ideally, we would like $\Phi_{ij}^{\calH}$ to be maximized at $j=\pi(i)$, at least on average. To this end, we require $H \in \calH$ to be \emph{uniquely rooted}, under which we have  $\Expect[\Phi_{ij}^{\calH}] \propto \indc{\pi(i)=j}$ (see \prettyref{prop:mean_Phi_ij}).

\begin{definition}[Uniquely rooted graph]
\label{def:unique_root}
We say that a graph $H$ rooted at $i$ is uniquely rooted, if 
$H(i)$ is non-isomorphic to $H(v)$ for any vertex $v \neq i$ in $V(H)$.
\end{definition}

However, the uniquely rooted property is far from enough. In order for the signature $\Phi_{ij}^{\calH}$ to distinguish
whether $j=\pi(i)$ or not, we need to ensure that the fluctuation of $\Phi_{ij}^{\calH}$  does not overwhelm the mean  $\Expect[\Phi_{ii}^{\calH}]$ for all $j \in [n]$. 
In particular, we need 
$\Var[\Phi_{ij}^{\calH}]$ to be much smaller than $(\Expect[\Phi_{ii}^{\calH}])^2$.
This turns out to be  extremely challenging to show and calls for a rather delicate choice of $\calH$. To this end, we construct a special family of trees $\calT$, which we call \textit{chandeliers} (see~\prettyref{fig:chandelier} for an illustration).

\begin{definition}[Chandelier]
\label{def:chandelier}
An \emph{$(L, M, K, R)$-chandelier} is a rooted tree with $L$ branches, each of which consists of a path with $M$ edges (which we call an \emph{$M$-wire}) followed by a rooted tree with $K$ edges (which we call a \emph{$K$-bulb}); the $K$-bulbs are non-isomorphic to each other and each of them has 
at most $R$  automorphisms.

\end{definition}


\begin{figure}
    \centering
    \begin{tikzpicture}[scale=3,transform shape,auto,font=\scriptsize]
	\draw (0,-0.4) node (root) [root] {};
	\draw (-0.5,-0.5) node (a1) [vertexdot] {};
	\draw (-1,-1) node (a2) [vertexdot] {};
	\draw (0,-0.7) node (b1) [vertexdot] {};
	\draw (0,-1) node (b2) [vertexdot] {};
	\draw (0.5,-0.5) node (c1) [vertexdot] {};
	\draw (1,-1) node (c2) [vertexdot] {};
	
	\draw (root) edge[red, bend right=10, thick] (a1);
	\draw (a1) edge[red, bend right=20, thick] (a2);
	
	\draw (root) edge[red, thick] (b1);
	\draw (b1) edge[red, thick] (b2);
	
	\draw (root) edge[red, bend left=10, thick] (c1);
	\draw (c1) edge[red, bend left=20, thick] (c2);
	
	\draw (-1,-1.3) node (a3) [vertexdot] {};
	\draw (-1,-1.6) node (a4) [vertexdot] {};
	\draw (-1,-1.9) node (a5) [vertexdot] {};
	\draw (-1,-2.2) node (a6) [vertexdot] {};
	\draw (a2) edge[blue, thick] (a3);
	\draw (a3) edge[blue, thick] (a4);
	\draw (a4) edge[blue, thick] (a5);
	\draw (a5) edge[blue, thick] (a6);	
	\fill [nearly transparent,blue!50] (a4) ellipse (0.3 and 0.8);
	
	\draw (0.2,-1.3) node (b3) [vertexdot] {};
	\draw (0.2,-1.6) node (b4) [vertexdot] {};
	\draw (0.2,-1.9) node (b5) [vertexdot] {};
	\draw (-0.2,-1.3) node (b6) [vertexdot] {};
	\draw (b2) edge[blue, thick] (b3);
	\draw (b3) edge[blue, thick] (b4);
	\draw (b4) edge[blue, thick] (b5);
	\draw (b2) edge[blue, thick] (b6);
	\fill [nearly transparent,blue!50] (0.05,-1.5) ellipse (0.4 and 0.7);
	
	\draw (1,-1.3) node (c3) [vertexdot] {};
	\draw (1.2,-1.6) node (c4) [vertexdot] {};
	\draw (1.2,-1.9) node (c5) [vertexdot] {};
	\draw (0.8,-1.6) node (c6) [vertexdot] {};
	\draw (c2) edge[blue, thick] (c3);
	\draw (c3) edge[blue, thick] (c4);
	\draw (c4) edge[blue, thick] (c5);
	\draw (c3) edge[blue, thick] (c6);
	\fill [nearly transparent,blue!50] (1.07,-1.45) ellipse (0.4 and 0.6);
\end{tikzpicture}
    \caption{
    A chandelier with $L=3$, $M=2$, $K=4$, rooted at the solid vertex. The wires are shown in red, and the bulbs in blue.
    In this case $R=1$ since each bulb has no non-trivial automorphism (as rooted graphs). 
    }
    \label{fig:chandelier}
\end{figure}
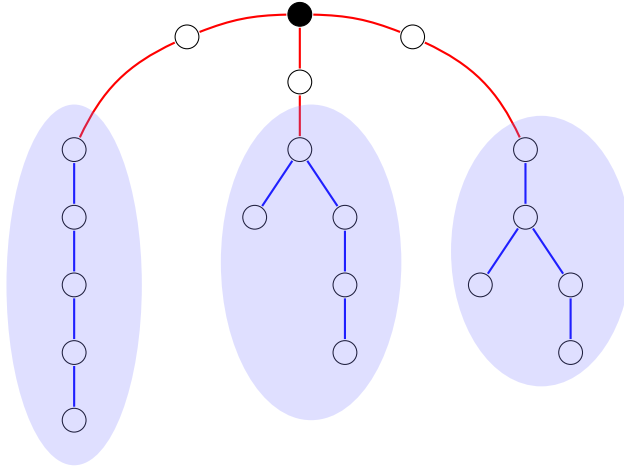

For any chandelier $H$, let
$\calK(H)$ denote its set of bulbs. Since all bulbs are \emph{non-isomorphic} to each other, we have 
\begin{align}
\aut(H)= \prod_{\calB\in \calK(H)}\aut(\calB), \label{eq:aut_H}
\end{align}
which is a special case of the classical recursive formula for the number of automorphisms of rooted trees~\cite{Jordan1869}.
Moreover, when $L \ge 2$, the root of $H$ is the unique vertex incident to $L$ branches each having $M+K$ edges. As a result, each chandelier is uniquely rooted. 




Let $\calT$ denote the family of non-isomorphic $(L,M,K,R)$-chandeliers.
Then 
\begin{equation}
|\calT| = \binom{|\calJ|}{L},
\label{eq:Tsize}
\end{equation}
where $\calJ\equiv \calJ(K,R)$ denotes the collection of unlabeled rooted trees having
$K$ edges and at most $R$ automorphisms. 
Counting unlabeled trees with a prescribed number of automorphisms has been well studied in the literature:
\begin{itemize}
    \item 
    \emph{All trees}: 
    As mentioned earlier in \prettyref{sec:intro}, a classical result in enumerative combinatorics is that
    the total number of unlabeled trees with $K$ edges satisfies as $K\to\infty$ \cite{polya1937kombinatorische,otter1948number}
    \begin{equation}
    |\calJ(K)| \equiv 
    |\calJ(K,\infty)| = (\alpha+o(1)) ^{-K} ,
    \label{eq:otter}
    \end{equation}
    where $\alpha \approx 0.338$ is Otter's constant.
    
  

    \item \emph{Typical trees}: The recent result \cite{olsson2022automorphisms} implies that the majority of the trees have $e^{\Theta(K)}$ automorphisms.\footnote{Indeed, \prettyref{eq:ow22} is an immediate corollary of the following asymptotic normality result in  \cite[Theorem 2]{olsson2022automorphisms}: 
    $
\frac{1}{\sqrt{K}}(\log\aut(H_K)-\mu K) \xrightarrow{K\to\infty} N(0,\sigma^2)$,
 where $H_K$ is a uniform random unlabeled tree with $K$ edges (known as the P\'olya tree of order $K+1$), and
 $\mu \approx 0.137$ and $\sigma^2 \approx 0.197$ are absolute constants.}
    In other words, for some absolute  constant $C$,
     \begin{equation}
    |\calJ(K,\exp(CK))| = (\alpha+o(1))^{-K}.
    \label{eq:ow22}
    \end{equation}
\end{itemize}
It turns out that to bound the fluctuation of the similarity score it is more advantageous if the bulbs do not have too much symmetry.
Thanks to \prettyref{eq:ow22} and in view of \prettyref{eq:aut_H}--\prettyref{eq:Tsize}, by choosing 
$ \R = \exp(CK)$ we can ensure that
$|\calT| = (\alpha+o(1))^{-N}$ has maximal growth while keeping $\aut(H)$ for each $H\in\calT$ relatively small.


In the rest of the paper, we will 
apply the similarity score 
$\Phi_{ij}\equiv \Phi_{ij}^\calT$ in 
\prettyref{eq:Phi_ij}
to the collection $\calT$ of chandeliers with carefully chosen parameters.
Crucially, by exploiting the structure of chandeliers, we show:
\begin{itemize}
    \item For true pairs $j =\pi(i)$, 
    \[
    \Expect[\Phi_{i\pi(i)}] = \mu, \quad 
    \Var(\Phi_{i\pi(i)}) = o(\mu^2) ,
    \]
    where 
    \begin{equation}
        \mu \triangleq |\calT|(\rho\sigma^2)^N \frac{(n-1)!}{(n-N-1)!}, \quad \sigma^2\triangleq q(1-q). \label{eq:tau}
    \end{equation}

    \item For fake pairs $j \neq \pi(i)$, 
  \[
    \Expect[\Phi_{ij}] = 0, \quad 
    \Var(\Phi_{ij}) = o\pth{\frac{\mu^2}{n^2}}.
    \]
\end{itemize}
This immediately implies that by running a greedy matching with weights $\Phi_{ij}$ (or simply thresholding 
$\Phi_{ij}$), we can match all but a vanishing fraction of vertices correctly with high probability. 
This is made precise by the following theorem.

Throughout this paper, we assume without loss of generality\footnote{If $q>1/2$, we can consider the complement graphs of $A$ and $B$, which are correlated \ER graphs with parameter $(n,1-q,\rho)$. In addition, it is not hard to see that the similarity scores $\Phi_{ij}$ remain unchanged.} that $q\leq 1/2$.

\begin{theorem}[Partial and almost exact recovery] \label{thm:Phi_ij_almost}
There exist  absolute constants $C_1, \ldots, C_4>0$ such that the following holds.
 Suppose  
\begin{align}
    \rho^{2} \ge \alpha +\epsilon \, , \label{eq:nq_rho_condition}
\end{align}
where $\epsilon$ is an arbitrarily small constant.
Choose $K,L,M, \R \in \naturals$ such that $N=(K+M)L$ is even\footnote{For simplicity, we assume $N$ is even so that $\mu \ge 0$ even when $\rho<0$. To lighten the notation, we do not explicitly round each parameter in \prettyref{eq:K_L_M_R_simple} to integers as this only changes constant factors; 
see~\prettyref{eq:K_L_M} for a more general condition.}, 
\begin{align}
 L= \frac{C_1}{\epsilon}, \quad K=C_2 \log n, \quad  M= \frac{C_3 K}{\log (nq)}, \quad \R = \exp\left(C_4 K \right) \, .  \label{eq:K_L_M_R_simple}
\end{align}
Fix any constant $0<c<1$ and let $\mu$ be given in \prettyref{eq:tau}.  
Let $\hat\pi:I\to[n]$ denote the output of 
\prettyref{alg:GMCC} applied to the collection $\calT$ of $(L,M,K,R)$-chandeliers and threshold $\tau=c \mu$.
Then $\hat{\pi}=\pi|_I$ with probability $1-o(1)$.
Moreover,
\begin{itemize}
    \item If $nq=\omega(1)$, then $|I|=(1-o(1))n$ with probability $1-o(1)$.  
    \item For any constant $\delta \in (0,1)$, there exists 
    a positive constant $C(\epsilon,\delta)$ depending only on $\epsilon$ and $\delta$, such that 
    if $nq \ge C(\epsilon,\delta)$,
    then $\expect{|I|} \ge (1-\delta) n$. 
\end{itemize}
\end{theorem}


\begin{remark}[Adapting to unknown parameters]
\label{rmk:unknown_rho}
Note that the choice of $M$ and $\tau$ in \prettyref{eq:K_L_M_R_simple} assumes
the knowledge of $q$ and $\rho$. 
The edge probability $q$ can be easily estimated by the empirical graph density of $A$ and $B$. 
Moreover, the threshold $\tau$ can be specified in a data-driven manner (cf.~\prettyref{app:unknown_rho}).

\end{remark}

From a computational perspective, na\"ive evaluation of $W_{i,H}(\bar{A})$ by exhaustive search for each $H$ with $N$ edges   takes $n^{\Theta(N)}$ time which is super-polynomial when $N=\omega(1)$.
To resolve this computational issue, in \prettyref{sec:color_coding}, we give a polynomial-time
algorithm (\prettyref{alg:approximated_Phi_ij}) that computes an approximation $\tilde{\Phi}_{ij}$
for $\Phi_{ij}$ using the strategy of \emph{color coding} as done in~\cite{mao2021testing}. 
 The following result shows that the approximated similarity score $\tilde{\Phi}_{ij} $  enjoys the same statistical guarantee
 under the same condition \prettyref{eq:nq_rho_condition} as~\prettyref{thm:Phi_ij_almost}.

\begin{theorem}\label{thm:tilde_Phi_ij_almost}
\prettyref{thm:Phi_ij_almost} continues to hold with $\tilde{\Phi}_{ij} $ in place of $\Phi_{ij}$.
Moreover, $\{\tilde{\Phi}_{ij}\}_{i,j\in[n]}$ can be computed in $O(n^C)$ for some constant $C\equiv C(\epsilon)$ 
depending only on $\epsilon$.


\end{theorem}

\prettyref{thm:tilde_Phi_ij_almost} shows that our matching algorithm achieves the almost exact recovery in polynomial time when $nq=\omega(1)$ and $\rho^2\ge \alpha+\epsilon$. 
In comparison,  the almost exact recovery is information-theoretically possible if and only if $nq\rho =\omega(1)$, when $\rho>0$ and $q =n^{-1/2-\Omega(1)}$~\cite{cullina2019partial,wu2021settling}.


Moreover, 
under an extra condition that is information-theoretically necessary,
we can upgrade the almost exact recovery to exact recovery in polynomial time.
The main idea is to use the partial matching $\hat\pi|_I$ correctly identified by \prettyref{alg:GMCC} as \textit{seeds} and apply a seeded matching algorithm (which is similar to 
percolation-based matching in \cite{yartseva2013performance,barak2019nearly})
to extend it to a full matching.  
For this purpose we assume $\rho>0$ as the current seeded matching algorithm requires positive correlation.




\begin{theorem}[Exact recovery]\label{thm:exact_reovery} 
Suppose 
\begin{align}
  n q\left( q + \rho (1-q ) \right) \ge (1+\epsilon)\log n, \quad \rho \ge \sqrt{\alpha +\epsilon} \label{eq:exact_recovery}
\end{align}
for some arbitrarily small constant  $\epsilon$.
Then a seeded matching algorithm
(see \prettyref{alg:recovery} in \prettyref{sec:boost})
with input $\hat{\pi}$ outputs $\tilde{\pi} = \pi$
in $O(n^3q^2)$ time with  probability $1-o(1)$. 
\end{theorem}
\begin{remark}[Comparison to the exact recovery threshold]
\label{rmk:exactIT}

It is instructive to compare the performance guarantee~\prettyref{eq:exact_recovery} of our polynomial-time algorithm with the information-theoretic threshold of exact recovery derived in \cite{wu2021settling} for positive correlation,
that is,
\begin{align}
      \rho \geq
      \left(1+\epsilon\right)
      \left(2\sqrt{\frac{\log n}{n}}+  \frac{\log n}{nq }\right)  \,. \label{eq:exact_recovery_threshold}
\end{align}

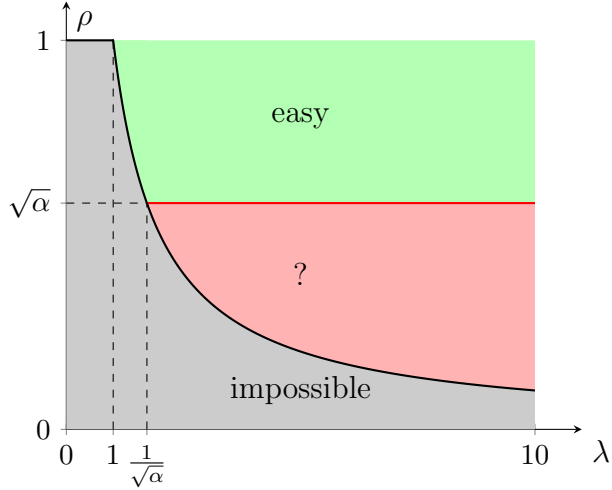
\begin{figure}[ht]
    \centering
\begin{tikzpicture}[transform shape,scale=1]

\pgfmathsetmacro{\a}{0.33832185689}
\pgfmathsetmacro{\n}{200}
\pgfmathsetmacro{\t}{1}
\pgfmathsetmacro{\s}{1/sqrt(\a)}
\pgfmathsetmacro{\b}{10}

\begin{axis}[
    standard,
    enlargelimits=lower,
        enlarge x limits={0.1,upper},
        enlarge y limits={0.1,upper},
         xmin=0,   xmax=10,
	ymin=0,   ymax=1,
extra x ticks={0},
extra y ticks={0},
xlabel={\large $\lambda$},ylabel={\large $\rho$},
scaled ticks=false, tick label style={/pgf/number format/fixed},
xtick={\t,\s, 10}, xticklabels={ $1$, $\frac{1}{\sqrt{\alpha}}$, $10$},
ytick={ sqrt{\a},1}, yticklabels={$\sqrt{\alpha}$,$1$},
every axis plot post/.append style={
  mark=none,samples=1000,smooth,thick} 
]

\addplot[name path =Otter, domain=\s:\b,red]{sqrt(\a)};

\addplot[name path =IT_1, domain=0:\s,black]{min(1,1/x)};
\addplot[name path =IT_2, domain=\s:\b,black]{min(1,1/x)};
\path [name path=xaxis_1] (axis cs:0,0) --(axis cs:\s,0);
\path [name path=xaxis_2] (axis cs:\s,0) --(axis cs:\b,0);
 \path [name path=yaxis] (axis cs:0,0) --(axis cs:0,1);
  \path [name path=xaxis_3] (axis cs:\s,1) --(axis cs:\t,1);
\path [name path=xaxis_4] (axis cs:\t,1) --(axis cs:\b,1);


\addplot[fill=black!20!white] fill between [of=IT_1 and xaxis_1];
\addplot[fill=black!20!white] fill between [of=IT_2 and xaxis_2];

\draw[black,dashed] (axis cs:\t,0) -- (axis cs:\t,1);
\draw[black,dashed] (axis cs:\s,0) -- (axis cs:\s,sqrt{\a});
\draw[black,dashed] (axis cs:0,sqrt{\a}) -- (axis cs:\s,sqrt{\a});

\addplot[fill=red!30!white] fill between [of=IT_2 and Otter];

\addplot[fill=green!30!white] fill between [of=xaxis_3 and IT_1 ];
\addplot[fill=green!30!white] fill between [of= xaxis_4 and Otter];

\node at (axis cs:5,.4) [rotate=0,align=center] {\large ?};

\node at (axis cs:5,.1) [rotate=0,align=center] {\large impossible};
\node at (axis cs:5,.8) [align=center] {\large easy};






\end{axis}
\end{tikzpicture}

   \centering
    \caption{The phase diagram for exact recovery in the logarithmic degree regime, where $nq=\lambda \log n$
    for a fixed constant $\lambda>0$.
    The impossible and easy regime are given by $\rho < \min\{1, 1/\lambda\}$ and $\rho > \max\{\sqrt{\alpha},1/\lambda\}$, respectively. No polynomial-time algorithm is known to achieve exact recovery in the red regime. 
    }
    \label{fig:GM_phase_new}
\end{figure}

Assuming $nq=\lambda \log n$ for a fixed constant $\lambda$, \prettyref{eq:exact_recovery} simplifies to $\rho> \max\{1/\lambda,\sqrt{\alpha}\}$,
while~\prettyref{eq:exact_recovery_threshold} is reduced to $\rho > 1/\lambda$; see \prettyref{fig:GM_phase_new} for an illustration. Observe that when $\lambda< 1/\sqrt{\alpha}$,  the  condition \prettyref{eq:exact_recovery} for exact recovery matches \prettyref{eq:exact_recovery_threshold}
and hence our polynomial-time matching algorithm is information-theoretically optimal.
If $\lambda>1/\sqrt{\alpha}$, there exists a gap,  between~\prettyref{eq:exact_recovery} and~\prettyref{eq:exact_recovery_threshold}, depicted as the red regime in~\prettyref{fig:GM_phase_new}.
It is an open problem whether exact recovery is attainable in polynomial time in the red regime when $\rho < \sqrt{\alpha}$. 
So far the only rigorous evidence for hardness is that  detection (and hence recovery) is computationally hard in the low-degree polynomial framework\footnote{Specifically, it is shown in \cite{mao2021testing} that any test statistic that is a degree-$\text{polylog}(n)$ polynomial of $(A,B)$ fails to detect correlation
$\rho= 1/\text{polylog}(n)$.} when $\rho \le 1/\text{polylog}(n)$ \cite{mao2021testing}.
\end{remark}

\subsection{On the choice of chandeliers}\label{sec:chandelier}

The key to the success of our matching algorithm is to leverage the correlation
of subgraph counts in the two graphs $A$ and $B$ as much as possible,
while suppressing the undesired correlation between different subgraph counts. In this subsection, we explain why restricting to the special family of chandeliers is crucial, as well as some basic guidelines on the choice of its parameters. 
Assume for convenience that $\pi = \id$.

First of all, we require the expected similarity score $\E[\Phi_{ij}]$ to be zero except for $i = j$. 
As discussed in the previous subsection, this is guaranteed by the uniquely rooted property of each chandelier in $\calT$.
Further, to distinguish a true pair $(i,i)$ from fake pairs $(i,j)$, we need $\Var(\Phi_{ij})$  to be much smaller than $\E[\Phi_{ii}]^2$ for any pair $(i,j)$.
More precisely, in order to apply a union bound over all fake pairs, we need $\Var(\Phi_{ij})/\expect{\Phi_{ii}}^2 = o(1/n^2)$ for all $i \ne j$. 
Later in \prettyref{eq:fake_lowerbound}, we will see that even if all the tree counts were uncorrelated, 
the variance would always be lower bounded by $\Var(\Phi_{ij})/\expect{\Phi_{ii}}^2 = \Omega(|\calT|^{-1} \rho^{-2N})$.
It follows that our class of chandeliers $\calT$ needs to satisfy
\begin{equation}
|\calT|  = \omega(n^2 \rho^{-2N}).
\label{eq:class-entropy}
\end{equation}
By choosing the parameters appropriately, we can ensure that 
$|\calT|$ grows as $(\alpha+o(1))^{-N}$, almost
at the same rate as the entire set of unlabelled rooted trees. Therefore, whenever $\rho^2>\alpha$, \prettyref{eq:class-entropy} holds by choosing $N =\Theta(\log n)$.


To further see the significance of chandeliers on the correlations between subgraph counts, let us expand out the variance of 
$\Phi_{ij}$:
\begin{align}
\Var[\Phi_{ij}]
& =\sum_{H, I \in \calT}  \aut(H)  \aut(I)  \Cov \left( W_{i,H} (\Bar{A}) W_{j,H} (\Bar{B}) , W_{i,I} (\Bar{A}) W_{j,I} (\Bar{B}) \right) \nonumber \\
& = \sum_{H, I \in \calT} \aut(H)  \aut(I)
\sum_{S_1(i), S_2(j) \cong H}
\sum_{T_1(i), T_2(j) \cong I}
\Cov \left( \bar{A}_{S_1}\bar{B}_{S_2},
\bar{A}_{T_1} \bar{B}_{T_2} \right). \label{eq:var_Phi_ij_W}
\end{align}
Here, $S_1$ and $T_1$ are labeled subgraphs of $\complete$ isomorphic to chandeliers $H$ and $I$ respectively and both rooted at $i$, and similarly for $S_2$ and $T_2$ rooted at $j$. 
It turns out that, 
thanks to centering, 
$\Cov \left( \bar{A}_{S_1}\bar{B}_{S_2},
\bar{A}_{T_1} \bar{B}_{T_2} \right) = 0$ unless  every
edge in the union graph $U \triangleq S_1\cup T_1 \cup S_2 \cup T_2$ appears \emph{at least twice} in the 4-tuple $(S_1, T_1, S_2, T_2)$. Furthermore, each covariance in \prettyref{eq:var_Phi_ij_W} is upper bounded by $\sigma^{4N} q^{-2N+e(U)}$ (cf.~\prettyref{eq:product_expect}).
To proceed, we need to enumerate all possible $4$-tuples $(S_1, T_1, S_2, T_2)$ according to the union graph $U$. Note that the number of 
different vertex labelings of $U$ (excluding vertices $i$ and $j$) is simply upper bound by $n^{v(U)-1-\indc{i\neq j}}$. However, there are many configurations for the four chandeliers $(S_1, T_1, S_2, T_2)$ to generate the same unlabeled graph $U$, which may lead to excessive correlation. 
The chandelier structure is designed specifically to limit the possible overlapping patterns and reduce the correlations.



To convey some intuitions, let us focus on a true pair $(i,i)$ and consider the simple case  where $U$ is a tree and every edge in $U$ appears exactly twice in the $4$-tuple. 
In this case, $e(U)=2N$ and $v(U)=2N+1$.
Moreover, $U$ is a chandelier with $2L$ branches, each of which belongs to exactly two out of the four chandeliers $(S_1,T_1,S_2,T_2)$. 
For example, in \prettyref{fig:whychandelier}(a), we show two branches of $U$, one comes from $S_1, S_2$ and the other comes from $T_1, T_2$. 
Using this specific structure, we can precisely enumerate all possible $4$-tuples that generate such a union graph $U$ and bound their contributions to the variance.

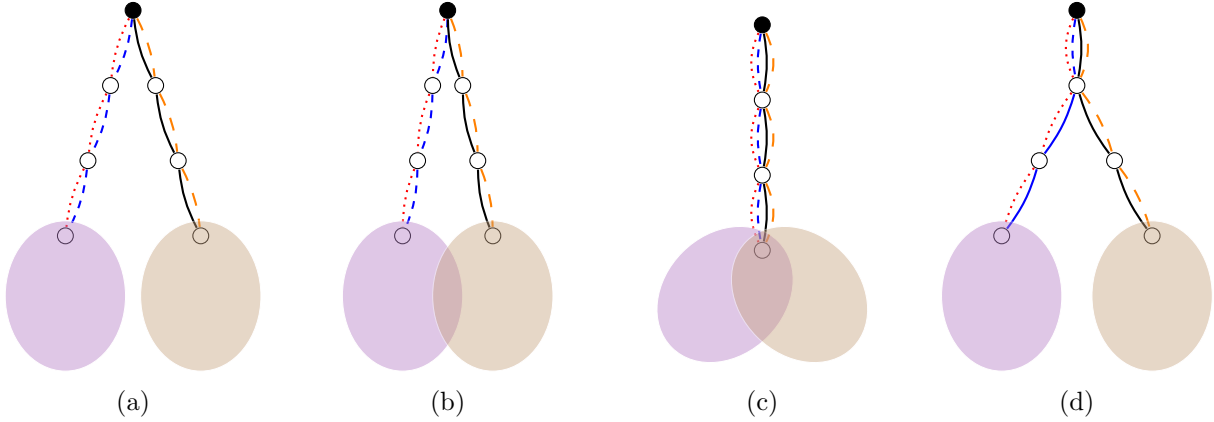
\begin{figure}
    \centering

\begin{subfigure}[b]{0.24\textwidth}
\centering
\begin{tikzpicture}[scale=2,transform shape]
\draw (0,0) node (root) [root] {};
\draw (-0.15,-0.5) node (a1) [vertexdot] {};
\draw (-0.30,-1) node (a2) [vertexdot] {};
\draw (-0.45,-1.5) node (a3) [vertexdot] {};
\draw (0.15,-0.5) node (c1) [vertexdot] {};
\draw (0.30,-1) node (c2) [vertexdot] {};
\draw (0.45,-1.5) node (c3) [vertexdot] {};

\draw (root) edge[blue, bend left=10, dashed, thick] (a1);
\draw (a1) edge[blue, bend left=10, dashed, thick] (a2);
\draw (a2) edge[blue, bend left=10, dashed, thick] (a3);
\draw (root) edge[red, bend right=10, dotted, thick] (a1);
\draw (a1) edge[red, bend right=10, dotted, thick] (a2);
\draw (a2) edge[red, bend right=10, dotted, thick] (a3);

\draw (root) edge[black, bend right=10, thick] (c1);
\draw (c1) edge[black, bend right=10, thick] (c2);
\draw (c2) edge[black, bend right=10, thick] (c3);
\draw (root) edge[orange, bend left=10, thick, dash pattern={on 5pt off 5pt on 5pt off 5pt}] (c1);
\draw (c1) edge[orange, bend left=10, thick, dash pattern={on 5pt off 5pt on 5pt off 5pt}] (c2);
\draw (c2) edge[orange, bend left=10, thick, dash pattern={on 5pt off 5pt on 5pt off 5pt}] (c3);

\filldraw[nearly transparent, color=white, fill=red!50](-0.45,-1.9) ellipse (0.4 and 0.5);
\filldraw[nearly transparent, color=white, fill=blue!50](-0.45,-1.9) ellipse (0.4 and 0.5);
\filldraw[nearly transparent, color=white, fill=black!50](0.45,-1.9) ellipse (0.4 and 0.5);
\filldraw[nearly transparent, color=white, fill=orange!50](0.45,-1.9) ellipse (0.4 and 0.5);
\end{tikzpicture}
\caption{}
\end{subfigure}
\hfill
\begin{subfigure}[b]{0.24\textwidth}
\centering
\begin{tikzpicture}[scale=2,transform shape]
\draw (0,0) node (root) [root] {};
\draw (-0.1,-0.5) node (a1) [vertexdot] {};
\draw (-0.2,-1) node (a2) [vertexdot] {};
\draw (-0.3,-1.5) node (a3) [vertexdot] {};
\draw (0.1,-0.5) node (c1) [vertexdot] {};
\draw (0.2,-1) node (c2) [vertexdot] {};
\draw (0.3,-1.5) node (c3) [vertexdot] {};

\draw (root) edge[blue, bend left=10, dashed, thick] (a1);
\draw (a1) edge[blue, bend left=10, dashed, thick] (a2);
\draw (a2) edge[blue, bend left=10, dashed, thick] (a3);
\draw (root) edge[red, bend right=10, dotted, thick] (a1);
\draw (a1) edge[red, bend right=10, dotted, thick] (a2);
\draw (a2) edge[red, bend right=10, dotted, thick] (a3);

\draw (root) edge[black, bend right=10, thick] (c1);
\draw (c1) edge[black, bend right=10, thick] (c2);
\draw (c2) edge[black, bend right=10, thick] (c3);
\draw (root) edge[orange, bend left=10, thick, dash pattern={on 5pt off 5pt on 5pt off 5pt}] (c1);
\draw (c1) edge[orange, bend left=10, thick, dash pattern={on 5pt off 5pt on 5pt off 5pt}] (c2);
\draw (c2) edge[orange, bend left=10, thick, dash pattern={on 5pt off 5pt on 5pt off 5pt}] (c3);

\filldraw[nearly transparent, color=white, fill=red!50](-0.3,-1.9) ellipse (0.4 and 0.5);
\filldraw[nearly transparent, color=white, fill=blue!50](-0.3,-1.9) ellipse (0.4 and 0.5);
\filldraw[nearly transparent, color=white, fill=black!50](0.3,-1.9) ellipse (0.4 and 0.5);
\filldraw[nearly transparent, color=white, fill=orange!50](0.3,-1.9) ellipse (0.4 and 0.5);
\end{tikzpicture}
\caption{}
\end{subfigure}
\hfill
\begin{subfigure}[b]{0.24\textwidth}
\centering
\begin{tikzpicture}[scale=2,transform shape]

\draw (0,0) node (root) [root] {};
\draw (0,-0.5) node (a1) [vertexdot] {};
\draw (0,-1) node (a2) [vertexdot] {};
\draw (0,-1.5) node (a3) [vertexdot] {};

\draw (root) edge[red, bend right=25, dotted, thick] (a1);
\draw (root) edge[blue, bend right=10, thick, dashed] (a1);
\draw (a1) edge[red, bend right=25, dotted, thick] (a2);
\draw (a1) edge[blue, bend right=10, thick, dashed] (a2);
\draw (a2) edge[red, bend right=25, dotted, thick] (a3);
\draw (a2) edge[blue, bend right=10, thick, dashed] (a3);

\draw (root) edge[black, bend left=10, thick] (a1);
\draw (root) edge[orange, bend left=25, thick, dash pattern={on 5pt off 5pt on 5pt off 5pt}] (a1);
\draw (a1) edge[black, bend left=10, thick] (a2);
\draw (a1) edge[orange, bend left=25, thick, dash pattern={on 5pt off 5pt on 5pt off 5pt}] (a2);
\draw (a2) edge[black, bend left=10, thick] (a3);
\draw (a2) edge[orange, bend left=25, thick, dash pattern={on 5pt off 5pt on 5pt off 5pt}] (a3);

\filldraw[nearly transparent, color=white, fill=red!50,rotate around={-45:(0,-1.9)}](-0.25,-2) ellipse (0.4 and 0.5);
\filldraw[nearly transparent, color=white, fill=blue!50,rotate around={-45:(0,-1.9)}](-0.25,-2) ellipse (0.4 and 0.5);
\filldraw[nearly transparent, color=white, fill=black!50,rotate around={45:(0,-1.9)}](0.25,-2) ellipse (0.4 and 0.5);
\filldraw[nearly transparent, color=white, fill=orange!50,rotate around={45:(0,-1.9)}](0.25,-2) ellipse (0.4 and 0.5);
\end{tikzpicture}
\caption{}
\end{subfigure}
\hfill
\begin{subfigure}[b]{0.24\textwidth}
\centering
\begin{tikzpicture}[scale=2,transform shape,auto,font=\scriptsize]
\draw (0,0) node (root) [root] {};
\draw (0,-0.5) node (a1) [vertexdot] {};
\draw (-0.25,-1) node (a2) [vertexdot] {};
\draw (0.25,-1) node (c2) [vertexdot] {};
\draw (-0.5,-1.5) node (a3) [vertexdot] {};
\draw (0.5,-1.5) node (c3) [vertexdot] {};

\draw (root) edge[black, bend left=10, thick] (a1);
\draw (a1) edge[black, bend right=10, thick] (c2);
\draw (c2) edge[black, bend right=10, thick] (c3);

\draw (root) edge[orange, bend left=25, thick, dash pattern={on 5pt off 5pt on 5pt off 5pt}] (a1);
\draw (a1) edge[orange, bend left=10, thick, dash pattern={on 5pt off 5pt on 5pt off 5pt}] (c2);
\draw (c2) edge[orange, bend left=10, thick, dash pattern={on 5pt off 5pt on 5pt off 5pt}] (c3);

\draw (root) edge[blue, bend right=10, dashed, thick] (a1);
\draw (a1) edge[blue, bend left=10, thick] (a2);
\draw (a2) edge[blue, bend left=10, thick] (a3);

\draw (root) edge[red, bend right=25, dotted, thick] (a1);
\draw (a1) edge[red, bend right=10, dotted, thick] (a2);
\draw (a2) edge[red, bend right=10, dotted, thick] (a3);

\filldraw[nearly transparent, color=white, fill=red!50](-0.5,-1.9) ellipse (0.4 and 0.5);
\filldraw[nearly transparent, color=white, fill=blue!50](-0.5,-1.9) ellipse (0.4 and 0.5);
\filldraw[nearly transparent, color=white, fill=black!50](0.5,-1.9) ellipse (0.4 and 0.5);
\filldraw[nearly transparent, color=white, fill=orange!50](0.5,-1.9) ellipse (0.4 and 0.5);
\end{tikzpicture}
\caption{}
\end{subfigure}
\caption{Examples of overlapping patterns of two branches from, say, $S_1 \cap S_2$ and $T_1 \cap T_2$, shown in red/blue and black/orange respectively. 
The solid vertex is the root $i$. 
(a): The two branches overlap only at the root $i$. 
(b): The two wires are disjoint and the two dangling bulbs intersect creating cycle(s).
(c): The two wires completely overlap and the bulbs can intersect into an arbitrary tree.
(d): The two wires overlap in the beginning before branching out and the bulbs are disjoint.
}
    \label{fig:whychandelier}
\end{figure}

Moving from this simple case (referred to as the baseline) to more general cases, the following three observations are crucial for bounding the total variance (although the proof does not exactly follow this classification):
\begin{itemize}
\item If bulbs from different branches overlap (\prettyref{fig:whychandelier}(b)), this
will create cycles and hence increase the excess $e(U)-v(U)$, gaining extra factors of $1/n$ in the variance bound \prettyref{eq:var_Phi_ij_W} compared to the contribution of the baseline. As a result, although the structure of $U$ is difficult to track, a crude enumeration based on $e(U)$ and $v(U)$ suffices.
Next we assume $U$ is a \emph{tree}.

    \item If two wires completely overlap (\prettyref{fig:whychandelier}(c)), both $e(U)$ and $v(U)$ are reduced by $M$ and hence we gain a factor of $(nq)^{-M}$ in the variance bound. On the other hand, the two bulbs can intersect to form an arbitrary tree which has at most $\exp(O(K))$ possibilities up to isomorphism. To ensure $(nq)^{-M}$ dominates $\exp(O(K))$, we need $M \gtrsim K/\log (nq)$. 
    
    \item If two wires first intersect then branch out (\prettyref{fig:whychandelier}(d)), the attached bulbs must be disjoint (otherwise a cycle will ensue), so that each bulb appears in 
    exactly two out of $(S_1,T_1,S_2,T_2)$.
    It turns out that the worst case occurs when the two wires share a single edge, for which there are at most $L^2$ possible ways (since each chandelier has $L$ wires). On the other hand, we gain a factor of $(nq)^{-1}$ in the variance bound \prettyref{eq:var_Phi_ij_W} (cf.~\prettyref{rmk:true_lowerbound}). 
    Thus to ensure $L^2$ is dominated by $nq^{-1}$, we need $L=o(\sqrt{nq})$.
\end{itemize}

In all, we see that it is critical for chandeliers to be a ``thin'' tree with 
only a few long wires, especially when the graphs get sparser. To further reduce the symmetry, we require the bulbs in each chandelier $H$ are all non-isomorphic so that $\aut(H)$ is given by \prettyref{eq:aut_H}, namely 
$\aut(H) = \prod_{\calB\in\calK(H)} \aut(\calB)$, and each $\aut(\calB)$ is required to be at most $R=\exp(O(K))$.

The method of counting signed trees has been applied to the \emph{detection} problem in the previous work \cite{mao2021testing}. The goal therein is to decide whether two \ER random graphs are independent or correlated using the test statistic
\begin{equation}
f(A,B) = \sum_{H \in\calT'} \aut(H) \, W_{H} (\Bar{A}) \, W_{H} (\Bar{B}),
\label{eq:test}
\end{equation}
where the weighted subgraph count $W_{H}(W)$ is  similarly defined as~\prettyref{eq:W_i_H} for unrooted $H$.
Compared to  \prettyref{eq:Phi_ij}, there are three major distinctions:
First, the trees in  \prettyref{eq:test}  are not rooted and $\aut(\cdot)$ is for unrooted graphs.
Second, trees in $\calT'$ only have $\Theta\left(\log n / \log \log n \right)$ edges, instead of $\Theta(\log n)$ edges required in this paper. 
This is because for detection, one only needs to achieve a vanishing error, instead of a specific $o(1/n^2)$ error probability for recovery in the current work. 
Third (and most importantly), $\calT'$ contains all trees without special structure, while here we choose $\calT$ to be a family of special trees called chandeliers, which, as explained earlier, is crucial for reducing the correlation between different signed tree counts.

In terms of analysis, for the detection problem in \cite{mao2021testing} the latent permutation is chosen uniformly at random, so one can average the second moment calculation over the random permutation 
which drastically simplifies the analysis of the tree counting statistic. 
In contrast, for the recovery problem in the present paper, we need to condition on the realization of the latent permutation.
As such, the second moment calculation here is much more challenging combinatorially and involves delicate enumeration procedures that revolve around the chandelier construction.
In addition, since 
the trees in \cite{mao2021testing} are much smaller with only $\Theta(\frac{\log n}{\log \log n})$ edges, so that many quantities can be bounded very crudely (e.g.,~$\aut(H)\leq v(H)!$); for the current paper since the trees have $\Theta(\log n)$ edges such simple analysis does not suffice.

\section{Statistical analysis of similarity scores}
\label{sec:statistical-analysis}


Throughout the analysis, without loss of generality, we assume $\pi = \id$. 
First, we compute the first moment of the similarity scores $\Phi_{ij}^{\calH}$ for a general collection $\calH$ of subgraphs.

\begin{proposition}\label{prop:mean_Phi_ij}
Let $\calH$ be a family of unlabeled uniquely rooted graphs with $N$ edges and $V+1$ vertices.
For any $i,j\in [n]$, we have
\begin{align}
    \expect{\Phi_{i j}^{\calH}} & = |\calH| \left(\rho \sigma^2\right) ^{N} \frac{(n-1)!}{(n-V-1)!} \indc{i = j} \, , \label{eq:mean_Phi_ij}
\end{align}
where $\sigma^2=q(1-q)$.
If $V^2 =o(n)$, then $\expect{\Phi_{i j}^{\calH}} = (1+o(1)) |\calH| \left(\rho \sigma^2\right) ^{N}  n^V$.
\end{proposition}
\begin{proof}
For a rooted graph $H$ with $N$ edges and $V+1$ vertices, the number of copies of $H$ in the complete graph $\complete$ that are rooted at $i \in [n]$ is 
\begin{align}
    \sub_n(H)\equiv\sub(H,\complete)= \frac{\binom{n-1}{V}V!}{\aut(H)}  \, , \label{eq:a_H}
\end{align}
where recall that $\aut(H)$ denotes the number of automorphisms of $H$.
For any weighted adjacency matrix $M$ and any subgraph  $S$ of $\complete$, recall that $M_S= \prod_{e \in E(S)} M_e$ as in \prettyref{eq:W_i_H}.  
Then, 
\begin{align}
    \expect{W_{i,H}(\Bar{A})W_{j,H}(\Bar{B})}  
    & = \sum_{S(i)\cong H} \sum_{T(j)\cong H} \expect{\Bar{A}_S  \Bar{B}_T} \nonumber  \\
    & \stepa{=} \sum_{S(i)\cong H,S(j)\cong H}\expect{\Bar{A}_S \Bar{B}_S }  \nonumber \\
    & \stepb{=} \left(\rho \sigma^2\right)^N \sub_n(H) \indc{i = j}  \,, \label{eq:W_i_j_H}
\end{align}
where (a) is because $\expect{\Bar{A}_S \Bar{B}_T } = 0 $ unless $S=T$ (as unrooted graphs);
(b) is because $S(i)\cong H$ and $S(j)\cong H$ imply that $i=j$, thanks to the unique-rootedness of $H$. By \prettyref{eq:Phi_ij},
\begin{align*}
    \Expect[\Phi_{ij}^{\calH}] =  ~ \sum_{H \in\calH} \aut(H) \Expect[W_{i,H} (\Bar{A}) W_{j,H} (\Bar{B})] 
    & = ~ |\calH|  \left(\rho \sigma^2\right)^N \binom{n-1}{V}V!\indc{i=j}.
\end{align*}
In view of $\binom{n-1}{V}V!=\frac{(n-1)!}{(n-V-1)!}, $ we obtain the desired~\prettyref{eq:mean_Phi_ij}.
Finally, since $\left(1-\frac{V}{n}\right)^V \le \frac{(n-1)!}{(n-V-1)!n^V} \le \left(1-\frac{1}{n}\right)^V$ and $V=o(\sqrt{n})$, we have $ \frac{(n-1)!}{(n-V-1)!} = (1+o(1))n^V$.
\end{proof}

Next, we bound the variance of the similarity scores 
$\Phi_{ij}\equiv \Phi_{ij}^{\calT}$, where $\calT$ is the collection of $(K,L,M,R)$-chandeliers, for both true pairs $i=j$ and fake pairs $i\neq j$.
In the remainder of the paper, let $\Universal$ denote a universal constant such that 
\begin{equation}
|\calJ(K)| \le \Universal^K, \quad \forall K\geq 1.
\label{eq:Universal}
\end{equation}
Such a $\Universal$ (not to be confused with Otter's constant $\alpha$) exists thanks to~\prettyref{eq:otter}.

\begin{proposition}[True pairs]
\label{prop:true_pair}
Suppose $q\le \frac{1}{2}$, $L\ge 2$, and
\begin{align}
    \frac{14 L^2  }{\rho^{2(K+M)}|\calJ|}  \le \frac{1}{2} \,, \quad \frac{11 \R^{4}(2N)^3 (11 \Universal )^{2(K+M)}}{n} \le  \frac{1}{2}\,,
    \quad  \frac{\R^{\frac{4}{M}} (11 \Universal )^{\frac{4M+4K}{M}}}{nq} \le \frac{1}{2}\,,
    \quad \frac{1+2L^2}{\rho^2 nq} \le \frac{1}{2}\,. \label{eq:true_pair_constraint}
\end{align}
Then, for any $i \in [n]$, we have
\begin{align}
    \frac{\Var\left[\Phi_{ii}\right]}{\expect{\Phi_{ii}}^2} = O
     \left(\frac{L^2 }{\rho^2 nq} + \frac{L^2 }{\rho^{2(K+M)}|\calJ|} \right) \, . \label{eq:upper_bound_true_pair}
\end{align}
\end{proposition}

\begin{proposition}[Fake pairs]
\label{prop:fake_pairs_improved}
Suppose $q\le \frac{1}{2}$, $L \ge 2$, and
\begin{align}
\frac{R^{\frac{2}{M} } (11 \Universal)^{\frac{4(K+M)}{M}}}{nq}\le \frac{1}{2}\,,
\quad
       4^{L+3} L^{ 2L \wedge (4K+2)} (11\Universal)^{8(K+M)} \R^2 (2N+1)^3 \le \frac{n}{2}\,. \label{eq:fake_pairs_constraint}
\end{align}
Then, for any $i\neq j$, we have
\begin{align}
\frac{\Var\left[\Phi_{ij}\right]}{\expect{\Phi_{ii}}^2} = O\left(\frac{1}{|\calT| \rho^{2N} }\right) \, . \label{eq:upper_bound_fake_pair}
\end{align}
\end{proposition}

The next remark shows that 
the results in Propositions 
\ref{prop:true_pair} and 
\ref{prop:fake_pairs_improved}
are essentially optimal, by identifying which configurations of $(S_1,T_1,S_2,T_2)$ in \prettyref{eq:var_Phi_ij_W} contribute predominantly to the variance.

\begin{remark}\label{rmk:true_lowerbound}
The upper bound~\prettyref{eq:upper_bound_true_pair} for true pairs is almost tight. In fact,
when $N^2 \ll n$, $q=o(1)$ and $\rho\ge 0$,
\begin{align}
    \frac{\Var\left[\Phi_{ii}\right]}{\expect{\Phi_{ii}}^2} \ge \Omega \left(\frac{L^2 }{n q }+  \frac{L^2}{\rho^{2(K+M)}|\calJ|}\right) \,.
    \label{eq:upper_bound_true_pair1}
\end{align}
For the first term in this lower bound,
fix any $H,I \in \calT$ and consider the special case where $S_1=S_2\cong H$, $T_1=T_2\cong I$, where $S_1$ and $T_1$ only intersect on one edge that connects to $i$ (see \prettyref{fig:whychandelier}(d)). Then, we can show that $  \Cov\left(\bar{A}_{S_1}\bar{B}_{S_2}, \bar{A}_{T_1} \bar{B}_{T_2}\right) = \Omega(\left(\rho\sigma^2\right)^{2N}  q^{-1})$. 
There are $\Omega(L^2 \sub_n(H) \sub_n(I)n^{-1})$ number of $(S_1,T_1,S_2,T_2)$ that satisfies the above condition. Combining this with \prettyref{eq:var_Phi_ij_W} and applying \prettyref{prop:mean_Phi_ij}, we obtain
\begin{align*}
    \frac{\Var\left[\Phi_{ii}\right]}{\expect{\Phi_{ii}}^2} 
    & \gtrsim ~ \frac{1}{|\calT|^2\left(\rho\sigma^2\right)^{2N} n^{2N}} \sum_{H,I \in \calT} \aut(H)\aut(I) \sub_n(H)\sub_n(I)\left(\rho\sigma^2\right)^{2N}   L^2  (nq)^{-1}  =  \Omega\left(\frac{L^2}{nq}\right) \,,
\end{align*}
where the last equality holds because $\aut(H) \sub_n(H) = \Omega(n^{N})$ by \prettyref{eq:a_H} and $N^2 \ll n $.

For the second term in \prettyref{eq:upper_bound_true_pair1}, suppose the chandeliers 
$H$ and $I$ only share one common bulb $\calB$ (i.e., $|\calK(H)\cap \calK(I)|=1$). Consider $(S_1,T_1,S_2,T_2)$ such that (i) $S_1$ (resp.~$T_1$) completely overlaps with $S_2$ (resp.~$T_2$) except for $\calB$ and its attached wire; (ii) 
$S_1$ (resp.~$S_2$) only overlaps with $T_1$ (resp.~$T_2$) on $\calB$ and its attached wire.
This corresponds to a baseline case as described in~\prettyref{sec:chandelier}.
Then, we can show $  \Cov\left(\bar{A}_{S_1}\bar{B}_{S_2}, \bar{A}_{T_1} \bar{B}_{T_2}\right) \ge \left(\rho\sigma^2\right)^{2N} \rho^{-2(M+K)}$, and there are $\Omega(\sub_n(H) \sub_n(I))$ number of $(S_1,T_1,S_2,T_2)$ satisfying the above conditions (i) and (ii). 
Therefore, combining this with  \prettyref{eq:var_Phi_ij_W} and  \prettyref{prop:mean_Phi_ij} yields
\begin{align*}
   \frac{\Var\left[\Phi_{ii}\right]}{\expect{\Phi_{ii}}^2} 
    & \gtrsim ~\frac{1}{|\calT|^2\left(\rho\sigma^2\right)^{2N} n^{2N}} \sum_{H,I \in \calT} \aut(H)\aut(I) \sub_n(H) \sub_n(I) \left(\rho\sigma^2\right)^{2N} \rho^{-2(M+K)}  \indc{|\calK(H)\cap\calK(I)| = 1} \\
    & = ~  \frac{\sum_{H,I \in \calT}\indc{|\calK(H)\cap\calK(I)| = 1} }{|\calT|^2\rho^{2(M+K)}} 
    \gtrsim ~\frac{L^2}{\rho^{2(K+M)}|\calJ|}\,.
\end{align*}
where the last step holds because there are $L\binom{|\calJ|}{L}\binom{|\calJ|}{L-1}$ number of pairs of $H$ and $I$ that only share a single bulb.
The upper bound~\prettyref{eq:upper_bound_fake_pair} for fake pairs is sharp. In fact, if $N^2\ll n$, $q \le 1/2$, and $\rho \ge 0$,  for any collection $\calH$ of uniquely rooted trees
 (not just chandeliers) and any fake pair $i \neq j$, we have
\begin{align}
\frac{\Var[\Phi_{ij}^{\calH}]}{\expect{\Phi_{ii}^{\calH}}^2}  \ge \Omega\left(\frac{1}{|\calH| \rho^{2N} }\right)\,. \label{eq:fake_lowerbound}
\end{align}
To see this, first note that 
for any $S_1,T_1,S_2,T_2$ where $S_1(i),S_2(j) \cong H$ and $T_1(i),T_2(j) \cong I$ with $H,I \in \calH$, 
\begin{align*}
    \Cov\left({W_{i,H} (\Bar{A}) W_{j,H} (\Bar{B})\, , \, W_{i,I} (\Bar{A}) W_{j,I} (\Bar{B}) }\right)
    & = ~ \expect{{W_{i,H} (\Bar{A}) W_{j,H} (\Bar{B}) W_{i,I} (\Bar{A}) W_{j,I} (\Bar{B}) }} 
    \ge 0 \,,
\end{align*}
where the first equality applies \prettyref{eq:W_i_j_H} for uniquely rooted trees, and the last inequality holds 
because   $ \expect{\Bar{A}_{S_1} \Bar{B}_{S_2}\Bar{A}_{S_1} \Bar{B}_{S_2} } \ge 0$ whenever $\rho \ge 0$
(cf.~\prettyref{eq:cross_moment_equal} in \prettyref{lmm:cond-exp}, \prettyref{app:pre}). 
Second, consider the special case where $H =I$ and $S_1=T_1$, $S_2=T_2$, 
$S_1$ and $S_2$ are vertex-disjoint (\ie, just focus on the diagonal terms in the expansion of the variance~\prettyref{eq:var_Phi_ij_W} and ignore the possible correlations between counts of distinct trees in $\calH$),
we get
\begin{align*}
    & \Cov\left({W_{i,H} (\Bar{A}) W_{j,H} (\Bar{B}) }, {W_{i,H} (\Bar{A}) W_{j,H} (\Bar{B}) }\right) \\
    & \ge ~ \sum_{S_1(i)=T_1(i)\cong H} \sum_{S_2(j)=T_2(j)\cong H} \indc{\text{$S_1$ and $S_2$ are vertex-disjoint}}
\Cov\left( \Bar{A}_{S_1} \Bar{B}_{S_2}, 
\Bar{A}_{T_1} \Bar{B}_{T_2}\right)\\
    & = ~ 
     \sigma^{4N} \sum_{S_1(i)\cong H} \sum_{S_2(j)\cong H} \indc{\text{$S_1$ and $S_2$ are vertex-disjoint}}    = \Omega\left(\sigma^{4N} n^{2N}/ \aut^2(H)\right) \,.
\end{align*}
Therefore, 
$$
\Var\left[\Phi_{ij}^{\calH}\right]
\ge \sum_{H \in \calH} \aut(H)^2 
\Cov\left({W_{i,H} (\Bar{A}) W_{j,H} (\Bar{B}) }, {W_{i,H} (\Bar{A}) W_{j,H} (\Bar{B}) }\right )
\ge \Omega\left( |\calH| \sigma^{4N} n^{2N}\right).
$$
Combining the above with \prettyref{prop:mean_Phi_ij} yields \prettyref{eq:fake_lowerbound}. 
\end{remark}

\subsection{Proof of \prettyref{thm:Phi_ij_almost}}


We aim to prove \prettyref{thm:Phi_ij_almost} under the assumption~\prettyref{eq:nq_rho_condition}, that is, $\rho^2 \ge \alpha+\epsilon$,
and the following more general condition than \prettyref{eq:K_L_M_R_simple}: 
\begin{align}
      & L \le   \frac{ c_1 \log n}{\log \log n} \wedge c_6\sqrt{nq}, \quad \frac{c_2}{\log (nq)} \le \frac{M}{K} \le   \frac{\log \frac{ \rho^2}{\alpha}}{2\log \frac{1}{\rho^2}} , \quad  KL \ge  \frac{c_3 \log n}{\log \frac{\rho^2}{\alpha}}, \nonumber \\
    &  K+M \le c_4 \log n, \quad R=\exp(c_5 K),
    \label{eq:K_L_M} 
\end{align}
for some absolute constants $c_1, \ldots, c_6>0$.
Indeed, the specific choice of $K,L,M,R$ in~\prettyref{eq:K_L_M_R_simple} satisfies~\prettyref{eq:K_L_M} when $nq \ge C(\epsilon)$
for a sufficiently large constant $C(\epsilon)$ that only depends on $\epsilon$.

Next, we verify that \prettyref{eq:K_L_M} with appropriately chosen $(c_1,\ldots,c_6)$
ensures that the condition 
\prettyref{eq:true_pair_constraint} in~\prettyref{prop:true_pair} and the condition \prettyref{eq:fake_pairs_constraint} in~\prettyref{prop:fake_pairs_improved} are both satisfied for all sufficiently large $n$.
To start, we note that 
\begin{align}
\frac{M}{K} \le \frac{\log \frac{\rho^2}{\alpha} }{ 2 \log \frac{1}{\rho^2} } 
\quad \Longleftrightarrow  \quad \frac{\rho^{2(K+M)/K} }{\alpha} \ge \sqrt{\frac{\rho^2}{\alpha}}. \label{eq:M_K}
\end{align}
Moreover, since $R=\exp(c_5 K)$, by choosing $c_5$ to be an appropriate absolute constant and applying \prettyref{eq:ow22}, we have that for all $K$ large enough,
\begin{align}
|\calJ| \ge \left( \alpha (1+c_0)\right)^{-K} \label{eq:bound_J},
\end{align}
where $c_0>0$ is an arbitrarily small constant.
Combining the last two displayed equation gives that 
\begin{align}
\rho^{2(K+M)} |\calJ| \ge  \left( \frac{\rho^{2(K+M)/K} }{\alpha (1+c_0)} \right)^{K}
\ge  \left( \frac{\rho^2 }{\alpha} \right)^{K/4}
\label{eq:rho_M_cond_0},
\end{align}
where the last inequality holds by choosing $c_0=\rho^2/\alpha-1\ge\epsilon/\alpha$.
Since $L\le  \frac{c_1\log n}{\log \log n}$
and $KL \ge \frac{c_3 \log n}{\log (\rho^2/\alpha)}$,
$ K \ge  \frac{c_3\log \log n} {c_1 \log (\rho^2/\alpha)}$.
We deduce from~\prettyref{eq:rho_M_cond_0} that 
\begin{align}
\rho^{2(K+M)} |\calJ|  \ge  \left( \log n \right)^{c_3/(4c_1)} \ge \omega(L^2), \label{eq:rho_L_condition}
\end{align}
where the last inequality holds by choosing $c_1, c_3$ so that 
$c_3/c_1>8$.

Assuming that $K+M \le c_4 \log n$, $L\le c_1 \frac{\log n}{\log \log n}$,
and $R=\exp(c_5 K)$,
by choosing $c_4$ to be a sufficiently small constant and noting that $N=(K+M)L$, we deduce that
$$
\frac{11 \R^{4}(2N)^3 (11 \Universal )^{2(K+M)}}{n} \le  \frac{1}{2}\,.
$$
Assuming that $M/K \ge c_2/\log (nq)$,
by choosing $c_2$ to be a sufficiently large constant,
we get that 
$$
\frac{\R^{\frac{4}{M}} (11 \Universal )^{\frac{4M+4K}{M}}}{nq} \le \frac{1}{2}.
$$
Finally, assuming that $L\le c_6 \sqrt{nq}$ and $\rho^2>\alpha$, by choosing $c_6$ to be a sufficiently small constant, we conclude that
$$
\frac{1+2L^2}{\rho^2 nq} \le \frac{1}{2}
$$
completing the verification of  \prettyref{eq:true_pair_constraint}. 

For \prettyref{eq:fake_pairs_constraint}, under the assumption $L\le c_1 \frac{\log n}{\log \log n}$,
$
L^{L} \le n^{c_1}. \label{eq:L_L}
$
Thus, under the assumptions that $K+M \le c_4 \log n$,
and $R=\exp(c_5 K)$,  by choosing $c_1, c_4$ to be sufficiently small constants,  we get that 
$$
 4^{L+3} L^{ 2L \wedge (4K+2)} (11\Universal)^{8(K+M)} \R^2 (2N+1)^3 \le \frac{n}{2}\,,
$$
hence the desired\prettyref{eq:fake_pairs_constraint}.

Now we are ready to prove \prettyref{thm:Phi_ij_almost} by 
applying Propositions \ref{prop:mean_Phi_ij} and \ref{prop:fake_pairs_improved}. Define
\begin{align}
    F= \{ i: | \Phi_{ii} - \mu | > (1-c)\mu  \} \supset \{ i: \Phi_{ii} < \tau \}   \, , \label{eq:set_F}
\end{align}
in view of $\tau = c\mu$.
Applying~\prettyref{prop:mean_Phi_ij}, \prettyref{prop:fake_pairs_improved}, and Chebyshev's inequality, we get that
for any $i \neq j$, 
\begin{align}
\prob{ \Phi_{ij} \ge \tau } =
\prob{ \Phi_{ij} - \expect{\Phi_{ij} } \ge c \expect{\Phi_{ii} }}
\le \frac{\Var\left[\Phi_{ij}\right]}{c^2 \expect{\Phi_{ii} }^2 }
=O \left( \frac{1}{|\calT| \rho^{2N}} \right). \label{eq:fake_error}
\end{align}

Note that 
\begin{align}
    |\calT| \rho^{2N} = \binom{|\calJ|}{L}  \rho^{2N} \ge \left( \frac{|\calJ|}{L} \right)^{L} \rho^{2L(K+M)} \ge \left(\frac{1}{L}\right)^L \left( \frac{\rho^2}{\alpha} \right)^{KL/4} \ge n^{c_3/4-c_1} = \omega(n^2)  \,, \label{eq:calT_rho_2N}
\end{align}
where the second inequality holds due to~\prettyref{eq:rho_M_cond_0};
the last inequality holds due to the assumptions that $L \le c_1 \log n/\log \log n$ and $KL \ge c_3 \log n/\log (\rho^2/\alpha)$; the last equality holds by choosing $c_3/4-c_1>2$. 

Hence, applying union bound together with~\prettyref{eq:fake_error} yields that 
\begin{align}
\prob{\exists i \neq j \in [n]: \Phi_{ij} \ge  \tau } = o(1) \, . \label{eq:i_neq_j_union}
\end{align}
It follows that with probability at least $1-o(1)$,
$\Phi_{ij}< \tau$ for all $i \neq j \in [n]$,
which, by our construction of $I$ and $\hat{\pi}$, further implies 
further implies $I \supset [n]\setminus F $ and $\hat{\pi}=\pi|_I$.


        


By Chebyshev's inequality and our choice of $\tau=c \expect{\Phi_{ii}} \ge 0$, for any $i\in [n]$, 
\[
\prob{| \Phi_{ii} - \mu | > (1-c) \mu } = \prob{|\Phi_{ii} - \expect{\Phi_{ii}} | > (1-c) \expect{\Phi_{ii}} } 
\le \frac{\Var\left[\Phi_{ii}\right]}{(1-c)^2\expect{\Phi_{ii}}^2  }
\triangleq  \gamma \, ,
\]

Applying \prettyref{prop:true_pair} yields that  
\begin{align}
    \gamma 
    =  O\left(\frac{L^2 }{nq} + \frac{L^2 }{\rho^{2(K+M)}|\calJ|} \right) \,. \label{eq:gamma_F}
\end{align}
It follows that 
$\expect{|F|} \le \gamma n $. For any constant $\delta \in (0,1)$, we can choose the constant $C(\epsilon,\delta)$ large enough, so that when $nq \ge C(\epsilon, \delta)$, 
the assumption $L \le c_6 \sqrt{nq} $ holds for a sufficiently small constant $c_6$ and consequently $\gamma \le \delta$. Thus, $\expect{|I|}=n-\expect{|F|}\ge (1-\delta) n$.

If $nq=\omega(1)$, then by choosing $c_6=o(1)$ we get $\gamma=o(1)$.   
Therefore, by Markov's inequality, 
$$
\prob{ |F| \ge \sqrt{\gamma } n } \le \sqrt{\gamma} =o(1). 
$$
It follows that with probability at least $1-o(1)$,
$|F| \le \sqrt{\gamma} n$ and hence $|I| \ge (1-\sqrt{\gamma}) n = (1-o(1))n$.

\section{Second moment calculation}
\label{sec:second-moment}



In this section we bound the variance of the similarity scores $\Phi_{ij}$; this computation forms the bulk of the paper.
We start with some preliminary steps that are common for both true pairs ($i=j$) and fake pairs ($i\neq j$).
To be clear, let us first define the following set operations applied to graphs that will be used in the rest of the paper. 

\begin{definition}\label{def:graph_operation}
For any graph $S$ and $T$, 
\begin{itemize}
    \item $S\cap T$ denotes the graph with $E(S \cap T) \triangleq E(S)\cap E(T)$, and $V(S \cap T) \triangleq V(S)\cap V(T)$; 
    \item $S \cup T$ denotes the graph with $E(S \cup T) \triangleq E(S)\cup E(T)$, and $V(S \cup T) \triangleq V(S)\cup V(T)$;  
    \item $S \backslash T$ denotes the graph  with
    $E(S \backslash T) \triangleq E(S)\backslash E(T)$, and $V(S \backslash T)$ as the vertex set induced by $E(S \backslash T)$. 
    \item $S \Delta T$ denotes the graph with
    $E(S \Delta T) \triangleq E(S)\Delta E(T)$, and $V(S \Delta T)$ as the vertex set induced by $E(S \Delta T)$. 
    
\end{itemize}
\end{definition}

Continuing the expansion \prettyref{eq:var_Phi_ij_W},  note that, for fixed $S_1(i), S_2(j) \cong H$ and $T_1(i), T_2(j) \cong I$,
\begin{align}
& 
\Cov \left( \bar{A}_{S_1}\bar{B}_{S_2},
\bar{A}_{T_1} \bar{B}_{T_2} \right)  \nonumber \\
\stepa{=}& ~
\Cov \left( \bar{A}_{S_1}\bar{B}_{S_2},
\bar{A}_{T_1} \bar{B}_{T_2} \right) \indc{S_1 \neq S_2 \text{ or }  T_1\neq T_2 \text{ or } V(S_1)\cap V(T_1) \neq \{i\}} \nonumber\\
\stepb{\leq}& ~
\Expect[ \bar{A}_{S_1}\bar{B}_{S_2}
\bar{A}_{T_1} \bar{B}_{T_2}] \indc{S_1 \neq S_2 \text{ or }  T_1\neq T_2 \text{ or } V(S_1)\cap V(T_1) \neq \{i\}} \nonumber \\
\stepc{=} &~
\Expect[ \bar{A}_{S_1}\bar{B}_{S_2}
\bar{A}_{T_1} \bar{B}_{T_2}] \indc{S_1 \neq S_2 \text{ or }  T_1\neq T_2 \text{ or } V(S_1)\cap V(T_1) \neq \{i\}}  \indc{S_1\Delta T_1 \subset S_2 \cup T_2, S_2\Delta T_2 \subset S_1 \cup T_1} \nonumber
\end{align}
where 
(a) is because 
$\bar{A}_{S_1}\bar{B}_{S_2}$ and
$\bar{A}_{T_1} \bar{B}_{T_2}$ are independent if 
$S_1 \cup S_2$ and 
$T_1 \cup T_2$ are edge-disjoint, in particular, when $S_1=S_2$, $ T_1=T_2$ and $S_1$ only share the root vertex with $T_1$;
(b) applies $\Expect[\bar{A}_{S_1}\bar{B}_{S_2}] \geq 0$, since we have from \prettyref{eq:W_i_j_H} 
\[
\expect{\Bar{A}_{S_1} \Bar{B}_{S_2} } \expect{\Bar{A}_{T_1} \Bar{B}_{T_2}} 
= \begin{cases}
\left(\rho \sigma^2 \right)^{2N} & \text{if $S_1=S_2$ and $T_1=T_2$}\\
0 & \text{otherwise}
\end{cases} \,;
\]
(c) is because of the following crucial observation: 
Note that $\bar{A}_{S_1}\bar{B}_{S_2}
\bar{A}_{T_1} \bar{B}_{T_2}$ is a product of monomials of in $(\bar A,\bar B)$ of maximum degree at most four, and has zero mean if a degree-one term is present. Thus, 
$\Expect[\bar{A}_{S_1}\bar{B}_{S_2},
\bar{A}_{T_1} \bar{B}_{T_2}]=0$ unless each edge in the $S_1\cup T_1 \cup S_2\cup T_2$ appears at least twice in the 4-tuple $(S_1, T_1, S_2, T_2)$.
 Combining the above with \prettyref{eq:var_Phi_ij_W} yields
\begin{align}
    \Var\left[\Phi_{i j}\right] 
    & \le ~  \sum_{H, I \in \calT} \aut(H)  \aut(I)  \sum_{S_1(i),S_2(j)\cong H} \sum_{T_1(i), T_2(j)\cong I} \nonumber 
    \left|\expect{\Bar{A}_{S_1} \Bar{B}_{S_2} \Bar{A}_{T_1} \Bar{B}_{T_2} 
   } \right|\nonumber \\
   &~~~~~ 
   \indc{S_1 \neq S_2 \text{ or }  T_1\neq T_2 \text{ or } V(S_1)\cap V(T_1) \neq \{i\}} 
   \indc{S_1\Delta T_1 \subset S_2 \cup T_2, S_2\Delta T_2 \subset S_1 \cup T_1} \triangleq \Gamma_{ij}
   \, , \label{eq:var_Phi_ij}
\end{align}

For any $i,j\in [n]$, let $\calW_{ij}$ denote the collection of $(S_1(i),T_1(i),S_2(j),T_2(j))$ where $S_1(i), S_2(j) \cong H$ and $T_1(i), T_2(j) \cong I$ for some $H,I \in \calT$ such that 
\begin{align} 
    S_1\Delta T_1 \subset S_2 \cup T_2 , \quad S_2\Delta T_2 \subset S_1 \cup T_1 \, , \label{eq:constraint1}
\end{align} 
and 
\begin{align}
 S_1 \neq S_2 \text{ or }  T_1\neq T_2 \text{ or } V(S_1)\cap V(T_1) \neq \{i\}. 
 \label{eq:constraint2}
\end{align}
Then, for any $i,j\in [n]$, we have $\aut(H)=\aut(S_1)=\aut(S_2)$, $\aut(I)=\aut(T_1)=\aut(T_2) $, and 
\begin{align}
    \Gamma_{ij}
    & \le ~ \sum_{(S_1(i),T_1(i),S_2(j),T_2(j))\in \calW_{ij}} \left(\aut(S_1)\aut(T_1)\aut(S_2)\aut(T_2)\right)^{\frac{1}{2}}\left| \expect{\Bar{A}_{S_1} \Bar{B}_{S_2} \Bar{A}_{T_1} \Bar{B}_{T_2}}  \right|\,. 
    \label{eq:var_Phi_ij_new}
\end{align}

To compute the expectation in \prettyref{eq:var_Phi_ij_new}, 
for any $(S_1,T_1,S_2,T_2) \in \calW_{ij}$ with $i,j\in [n]$, 
let us decompose their union into the following intersection graphs denoted by $\{K_{\ell,m}\}$ in \prettyref{tab:K_ell_m}:
\begin{table}[H]
\begin{center}
\begin{tabular}{ c |c| c| c }
 & $S_1 \Delta T_1$ & $S_1 \cap T_1$  & $(S_1 \cup T_1)^c$\\
 \hline
 $S_2 \Delta T_2$ & $K_{11}$ & $K_{21}$ &   $\emptyset$ \\
 \hline 
 $S_2 \cap T_2$ & $K_{12}$ & $K_{22}$ & $K_{02}$ \\  
 \hline
 $(S_2 \cup T_2)^c$ & $\emptyset$ & $K_{20}$ & 
\end{tabular}
\caption{Decomposition of the union graph $S_1 \cup T_1 \cup S_2\cup T_2 = \cup_{2 \leq \ell+m\leq 4} K_{\ell m}$, thanks to \prettyref{eq:constraint1}.}
\label{tab:K_ell_m}
\end{center}
\end{table}
Again conditioning on the latent matching being $\pi=\id$, we have
\begin{align}
    \expect{\Bar{A}_{S_1} \Bar{B}_{S_2} \Bar{A}_{T_1} \Bar{B}_{T_2}} 
= & \prod_{2\leq \ell+m\leq 4}  \prod_{(i,j)\in K_{\ell,m}}
\underbrace{ \Expect[\bar A_{ij}^\ell \bar B_{ij}^m ]}_{\triangleq \beta_{\ell m}}
= \prod_{2\leq \ell+m\leq 4}  \beta_{\ell m}^{k_{\ell m}}
\end{align}
where the cross moments satisfy (see  \prettyref{eq:cross_moment_bound} in \prettyref{lmm:cond-exp}),
\begin{equation}
\sigma^{-(\ell+m)}  |\beta_{\ell m}|\leq
\begin{cases}
|\rho| & (\ell,m)=(1,1)\\
1 & (\ell,m)=(2,0) \text{ or } (0,2) \\
\frac{1}{\sqrt{q}} & (\ell,m)=(2,1) \text{ or } (1,2) \\
\frac{1}{q} & (\ell,m)=(2,2)
\end{cases}.    
\label{eq:cross_moments1}
\end{equation}
Since $e(S_1)=e(T_1)=e(S_2)=e(T_2)=N$ and in view of the decomposition in \prettyref{tab:K_ell_m}, we have
\begin{align}
4N & = 2(e(K_{11})+e(K_{20})+e(K_{02})) + 3(e(K_{21})+e(K_{12})) + 4 e(K_{22}) \label{eq:decomp1} \\
e(S_1 \cup T_1 \cup S_2\cup T_2) & = e(K_{11})+e(K_{20})+e(K_{02}) + e(K_{21})+e(K_{12}) + e(K_{22}), \label{eq:decomp2}
\end{align}
so that
\begin{equation}
    e(K_{21})+e(K_{12})+ 2e(K_{22}) = 4N - e(S_1 \cup T_1 \cup S_2\cup T_2) \label{eq:S1_T1_S2_T2}.
\end{equation}
Then
\begin{align}
    \sigma^{-4N} \left|\expect{\Bar{A}_{S_1} \Bar{B}_{S_2} \Bar{A}_{T_1} \Bar{B}_{T_2}}\right| 
    \stepa{=} & ~ \prod_{2\leq \ell+m\leq 4}  (\sigma^{-(\ell+m)} \beta_{\ell m})^{e(K_{\ell m})} \nonumber\\
    \stepb{\leq} & ~ |\rho|^{e(K_{11}))} q^{-\frac{1}{2}\left(e(K_{21})+e(K_{12})+2e(K_{22})\right) } \nonumber\\
    \stepc{=} & |\rho|^{e(K_{11})} q^{-2N+ e(S_1 \cup T_1 \cup S_2\cup T_2)}  \,, \label{eq:product_expect}
\end{align}
where 
(a) applies \prettyref{eq:decomp1}; 
(b) applies \prettyref{eq:cross_moments1};
(c) applies \prettyref{eq:S1_T1_S2_T2}.

Combining \prettyref{eq:var_Phi_ij_new} and \prettyref{eq:product_expect},  we have
\begin{align}
    \Gamma_{ij}
    \le  \sigma^{4N}\sum_{(S_1(i),T_1(i),S_2(j),T_2(j))\in \calW_{ij}} & |\rho|^{e(K_{11})}   q^{-2N+ e(S_1 \cup T_1 \cup S_2\cup T_2)} 
    \nonumber \\ 
    &
    \left(\aut(S_1)\aut(T_1)\aut(S_2)\aut(T_2)\right)^{\frac{1}{2}} . \nonumber 
\end{align}
Let $\calW_{ij}(v,k)$ be comprised of those $(S_1(i),T_1(i),S_2(j),T_2(j))$ in $\calW_{ij}$ such that $S_1\cup T_1\cup S_2\cup T_2$ 
has $v$ distinct vertices (except for $i,j$) and excess $k$. 
Then $e(S_1 \cup T_1 \cup S_2\cup T_2) = v+k+1+\indc{i\neq j}$ and hence
\begin{align}    
    \Gamma_{ij}
    \le  ~ \sigma^{4N} \sum_{k}  \sum_{v}q^{-2N+v+k+1+\indc{i\neq j}} P_{ij}(v,k) \,, \label{eq:var_bound_P}
\end{align}
where 
\begin{align}
    P_{ij}(v,k) 
    & \triangleq ~ \sum_{(S_1(i),T_1(i),S_2(j),T_2(j))\in \calW_{ij}(v,k)}  |\rho|^{e(K_{11})}  \left(\aut(S_1)\aut(T_1)\aut(S_2)\aut(T_2)\right)^{\frac{1}{2}} \,. \label{eq:P_ij}
\end{align}

It remains to bound $P_{ij}(v,k)$, which involves enumerating $(S_1, T_1, S_2, T_2)$ in $\calW_{ij}(v,k)$ and bounding their number of automorphisms. To this end, we first introduce the notion of \emph{decorated graph} and establishes a one-to-one mapping between $(S_1, T_1, S_2, T_2)$ and the corresponding decorated graph $U =S_1 \cup T_1 \cup S_2 \cup T_2 $.  Then we carefully decompose the decorated graph $U$
into edge-disjoint parts and apply a ``divide-and-conquer'' strategy by separately enumerating each part and bounding its contribution to the number of 
automorphisms. Along the way, we crucially exploit the chandelier structure.



\subsection{Decorated graphs}\label{sec:decorated_graphs}
In this subsection, we introduce the notion of decorated graphs.
In a decorated graph $U$, every edge $e \in E(U)$ is associated with a subset $D_e$ of 4 symbols $\{\SS_1,\TT_1,\SS_2,\TT_2\}$. We call $D_e$ the decoration of $e$ in $U$. We say an edge is $\ell$-decorated if $|D_e| = \ell$, and a graph is $\ell$-decorated if every edge in the graph is $\ell$-decorated. If $e$ is not in $E(U)$, then we set $D_e = \emptyset$ by default.


Next we extend the set operation in \prettyref{def:graph_operation} to decorated graphs.
\begin{definition}
Fix two decorated graphs $U$ and $U'$ with decoration set $D_e$ and $D'_e$ respectively. 
Then, 
\begin{itemize}
    \item  For each $e \in E(U \cap U')$, 
    its decoration is $D_e \cap D_e'$;
    \item For each $e \in E(U \cup U')$, its decoration is $D_e \cup D_e'$; 
    \item For each $e \in E(U  \backslash U')$, its decoration is $D_e \backslash D_e'$;  
    \item For each  $e \in E(U  \Delta U')$, its decoration is $D_e  \Delta D_e'$\,.
\end{itemize}

According to these definitions, each intersection graph 
$K_{\ell m}$ defined in \prettyref{tab:K_ell_m}
is $(\ell+m)$-decorated. For example,
each edge in $K_{21}$ is decorated by  either
$\{\SS_1,\TT_1,\SS_2\}$ or $\{\SS_1,\TT_1,\TT_2\}$, and 
each edge in $K_{22}$ is decorated by the whole 
$\{\SS_1,\TT_1,\SS_2,\TT_2\}$.

\end{definition}


There is a natural one-to-one correspondence
between a 4-tuple $(S_1, T_1, S_2, T_2)$ and a decorated graph 
$U = S_1\cup T_1 \cup S_2 \cup T_2$, where for every edge $e\in E(U)$, we have $\SS_m \in D_e$ if and only if $e \in E(S_m)$ for $m=1,2$
and $\TT_m \in D_e$ if and only if $e \in E(T_m)$ for $m=1,2$. 
In other words, the decoration $D_e$ specifies the membership of each edge $e \in E(U)$ in $S_1, T_1, S_2, T_2$. Hence, from the decorated graph $U$, we can uniquely determine the corresponding $(S_1, T_1, S_2, T_2)$.

Now, to enumerate $(S_1, T_1, S_2, T_2) \in \calW_{ij}(v,k)$, it is equivalent to enumerate the corresponding decorated graph $U$.
 For every $ e \in E(U)$, it follows from  \prettyref{eq:constraint1} that $2 \le |D_e| \le 4$ 
and hence there are at most $\sum_{m=2}^4 \binom{4}{m}=11$ different choices of the decoration $D_e$. Since each edge in $U$ is at least $2$-decorated, we must have 
\begin{equation}
\label{eq:eU2N}
e(S_1\cup T_1 \cup S_2 \cup T_2) = e(U)  \le 2N.
\end{equation}

%




\subsection{Proof of \prettyref{prop:true_pair} for true pairs}
 
By \prettyref{eq:var_Phi_ij}, it suffices to show that under \prettyref{eq:true_pair_constraint},
\begin{align}
    \frac{\Gamma_{ii}}{\expect{\Phi_{ii}}^2} =  O
     \left(\frac{L^2 }{\rho^2 nq} + \frac{L^2 }{\rho^{2(K+M)}|\calJ|} \right)  \,.  \label{eq:Xi_ii_mean_square}
\end{align}
Recall that $\calW_{ii}(v,k)$ is the set of $(S_1(i),T_1(i),S_2(i),T_2(i))$ such that $S_1\cup T_1 \cup S_2 \cup T_2$ has $v+1$ vertices and excess $k$ with constraints \prettyref{eq:constraint1} and \prettyref{eq:constraint2} satisfied. Under the constraint \prettyref{eq:constraint1}, we must have $e(S_1\cup T_1 \cup S_2 \cup T_2)=v+k+1 \le 2N$ . Moreover, since $S_1(i),T_1(i), S_2(i), T_2(i)$ are all rooted at $i$, $S_1\cup T_1 \cup S_2 \cup T_2$ must be a connected graph, and hence $k\ge -1$.

By \prettyref{eq:true_pair_constraint}, we have $N^2=o(n)$.
Then, together with \prettyref{prop:mean_Phi_ij} and \prettyref{eq:var_bound_P}, we obtain that 
\begin{align}
   \frac{\Gamma_{ii}}{\expect{\Phi_{ii}}^2}
    & \le ~ (1+o(1)) n^{-2N} |\calT|^{-2} \rho^{-2N} \sum_{k\ge -1} \sum_{v=0}^{2N-k-1} q^{-2N+v+k+1} P_{ii}(v,k),  \label{eq:true_pair}
\end{align}
where $P_{ii}$ is given by \prettyref{eq:P_ij}.
By \prettyref{eq:eU2N}, $e(K_{11}) \le e(U) \le 2N$. Thus the RHS of \prettyref{eq:true_pair} is monotonically decreasing in $|\rho|$. 
Hence, to prove the proposition, it suffices to focus on $|\rho| \le \frac{3}{4}$. 
The following lemma provides an upper bound on $P_{ii}(v,k)$. 
\begin{lemma}\label{lmm:W_v_k_ell}
      For integers  $v, k$  with $v\le 2N-k-1$ and $k\ge -1 $, if \prettyref{eq:true_pair_constraint} holds and $|\rho| \le \frac{3}{4}$, then
\begin{align}
    P_{ii}(v,k)
    & \le ~ 2^6 n^{v} \rho^{2N}  |\calT|^2 \left(\R^4 11(2N+1)^{3} (11\Universal)^{2(K+M)}\right)^{k+1} 
    \sum_{e_2\ge 0}  \left(\R^{\frac{4}{M}} (11\Universal)^{ \frac{4M+4K}{M}} \right)^{e_2}     \nonumber \\
     &~~~~~  \sum_{e_1\ge 0} \left(1+ 2L^2 \right)^{e_1} \left(\rho^{-2e_1} \indc{e_1 \neq 0}+ \frac{12L^2}{\rho^{2(K+M)}|\calJ|}\indc{e_1 = 0} \right)   \indc{e_1+e_2 \le 2N-(v+k+1)} \, . \label{eq:W_v_k_ell}
\end{align}
 \end{lemma}
Next, by applying \prettyref{lmm:W_v_k_ell}, we get
\begin{align*}
     \frac{\Gamma_{ii}}{\expect{\Phi_{ii}}^2}
     & \le ~ (1+o(1)) 2^6 \sum_{k\ge -1}  \left(\frac{\R^4 11(2N+1)^{3} (11\Universal)^{2(K+M)}}{n}\right)^{k+1} 
    \sum_{e_2 \ge 0}  \left(\frac{\R^{\frac{4}{M}} (11\Universal)^{ \frac{4M+4K}{M}} }{nq}\right)^{e_2}     \nonumber \\
     &~~~~~  \sum_{e_1\ge 0} \left(\frac{1+ 2L^2}{ nq} \right)^{e_1} \left(\rho^{-2e_1} \indc{e_1 \neq 0}+ \frac{12L^2}{\rho^{2(K+M)}|\calJ|}\indc{e_1 = 0} \right)  \\
     &~~~~~ \sum_{v=0}^{2N-k-1}(nq)^{-2N+v+k+1 +e_1+e_2}\indc{e_1+e_2 \le 2N-(v+k+1)} \\
    & = O\left(\frac{L^2 }{\rho^2 nq} + \frac{L^2}{\rho^{2(K+M)}|\calJ|}\right) \,,
\end{align*}
where the last equality applies the following four facts: $(a)$  the last assumption in \prettyref{eq:true_pair_constraint} implies $nq \ge 2$ so that 
\begin{align*}
\sum_{v=0}^{2N-e_1-e_2-k-1} (nq)^{v+k+1-2N+e_1+e_2} \le 2 \,;
\end{align*}
$(b)$ under the last assumption in \prettyref{eq:true_pair_constraint}, we have
\begin{align*}
    &  \sum_{e_1\ge 0} \left(\frac{1+ 2L^2}{ nq} \right)^{e_1} \left(\rho^{-2e_1} \indc{e_1 \neq 0}+ \frac{12L^2}{\rho^{2(K+M)}|\calJ|}\indc{e_1 = 0} \right) \\
     & = ~
    \sum_{e_1\ge 1} \left(\frac{1+ 2L^2 }{\rho^2 nq}\right)^{e_1} +  \frac{12L^2}{\rho^{2(K+M)}|\calJ|}\\
    & \le ~ \frac{2+ 4L^2 }{\rho^2 nq} + \frac{12L^2}{\rho^{2(K+M)}|\calJ|} \,;
\end{align*}
$(c)$ under the third assumption in \prettyref{eq:true_pair_constraint},
we have 
\begin{align*}
\sum_{e_2\ge 0} \left(\frac{\R^{\frac{4}{M}}(11\Universal)^{\frac{4M+4K}{M}}}{nq}\right)^{e_2} \le 2 \,;
\end{align*}
$(d)$ under the second assumption in \prettyref{eq:true_pair_constraint}, we have
\begin{align*}
 \sum_{k\ge -1} \left(\frac{\R^{4}11(2N)^3 (11 \Universal )^{2(K+M)} }{n}\right)^{k+1} \le 2 \,.
\end{align*}

\subsubsection{Proof of \prettyref{lmm:W_v_k_ell}}
\label{sec:pf-W_v_k_ell}

As stated in \prettyref{sec:decorated_graphs},  to enumerating  $(S_1, T_1, S_2, T_2) \in \calW_{ii}(v,k)$, it is equivalent to enumerating the corresponding decorated graph $U$. Recall that $U$ has $v+1$ vertices and excess $k$ where $v+k+1\le 2N$ and $k\ge -1$.
To proceed, we first decompose $U$ into three edge-disjoint parts $\UTone$, $\UTtwo$, and $\UN$,
where $\UTone$ and  $\UTtwo$ are trees rooted at $i$, and $\UN$ is the remaining part that contains $i$ and has the same excess as $U$, 
as depicted in~\prettyref{fig:decomp_true}.

Specifically, for any neighbor $a$ of $i$ in $U$,
consider the graph $U$ with the edge $(i,a)$ removed and 
let $\calC(i,a)$ be the connected component therein that contains $a$.
Let $\calG(i,a) $ denote the graph union of $ \calC(i,a)$ and the edge $(i,a)$. 
Then we define
\begin{align}
    \calNT & = ~ \big\{ a : (i,a) \in E(U), \; \calG(i,a) \text{ is a tree} \big\} \, , \label{eq:calN_true}\\
      \calNN 
    & = ~ \{a: (i,a)\in E(U)\}\backslash \calNT \label{eq:calN_N}\,. 
\end{align} 
Further, we decompose $\calNT$ into a disjoint union $\calNTone \cup \calNTtwo$, defined by
\begin{align}
     \calNTtwo  
    & = \left\{ a \in \calNT: \left|  \{e \in E(\calG(i,a))  : |D_e| \ge 3\} \right| \ge M \right\}\label{eq:calN_T_2}, \\
     \calNTone 
    & = ~ \calNT\backslash \calNTtwo  \label{eq:calN_T_1} \,,
\end{align}
where recall that $D_e$ is the decoration set of edge $e \in U$.
Then we decompose $U$ according to $\calNTone$, $\calNTtwo$ and $\calNN$: 
\begin{align}
    \UTone \triangleq ~\bigcup_{a \in \calNTone} \calG(i,a) , \quad \UTtwo \triangleq ~\bigcup_{a \in \calNTtwo} \calG(i,a) , \quad
    \UN \triangleq ~ U \backslash (\UTone \cup \UTtwo) \, . \label{eq:U_decompose_true}
\end{align} 
Note that when $\calNTone$ (resp.\ $\calNTtwo$) is empty, 
we set $\UTone$ (resp.\ $\UTtwo$) to be the graph consisting of the single vertex $i$ by default. Similarly, if $\UN$ is empty, we set $\UN$ to be the graph consisting of the single vertex $i$. 

We pause to give some intuition behind this decomposition. Roughly speaking, $\UN$, $\UTtwo$,
and $\UTone$ correspond to the scenarios depicted in~\prettyref{fig:whychandelier} (b), (c), and (d) respectively. In particular, $\UN$ captures the overlapping part that induces cycles, while $\UTtwo$
deals with the case where two $M$-wires completely overlap. For both $\UN$ and $\UTtwo$, their sizes can be controlled and hence a crude enumeration bound solely based on their number of vertices and edges suffices. In contrast, $\UTone$ is more delicate. By definition, for each branch in $\UTone$, the number of $3$ or $4$-decorated edges is strictly less than $M$, so it cannot accommodate any $M$-wire whose edges are all at least $3$-decorated. Crucially, this further implies that each bulb in $\UTone$ must be exactly $2$-decorated -- see~\prettyref{fig:whychandelier} for an illustration and an intuitive explanation. 
We will make precise this claim and exploit it later in the proof of~\prettyref{lmm:U_T_1_true} when we bound the contribution of $\UTone$.

Next, based on the decomposition of $U$,
we assign appropriate weights to $\UTone$, $\UTtwo$, and $\UN$, so that $\left(\aut(S_1) \aut(T_1) \aut(S_2) \aut(T_2)\right)^{\frac{1}{2}} = w(\UTone) w(\UTtwo) w(\UN)$ (see \prettyref{claim:U_decompose_true}\ref{T:4} below). To this end,
for any chandelier $S$,
recall that $\calK(S)$ denotes the set of its bulbs.  For any subgraph $G \subset U$, 
define
\begin{align}
    w_S(G) \triangleq \prod_{\B \in \calK(S) , \, \B\subset G} \aut(\B)^{\frac{1}{2}}  \,.  \label{eq:w_S_U'}
\end{align}
with the understanding that $w_S(G)=1$ if $G$ contains no bulbs in  $ \calK(S)$.
Let 
\begin{align}
    w(G) \triangleq  w_{S_1}(G) w_{T_1}(G) w_{S_2}(G) w_{T_2}(G) \,. \label{eq:w_U'}
\end{align}


It is easy to verify that $\UTone$, $\UTtwo$, and $\UN$ satisfy the following properties (proved at the end of this subsection). 
\begin{figure}[ht]
    \centering
    \includegraphics[width=0.9\textwidth,trim=1cm 0cm 1cm 0cm,clip]{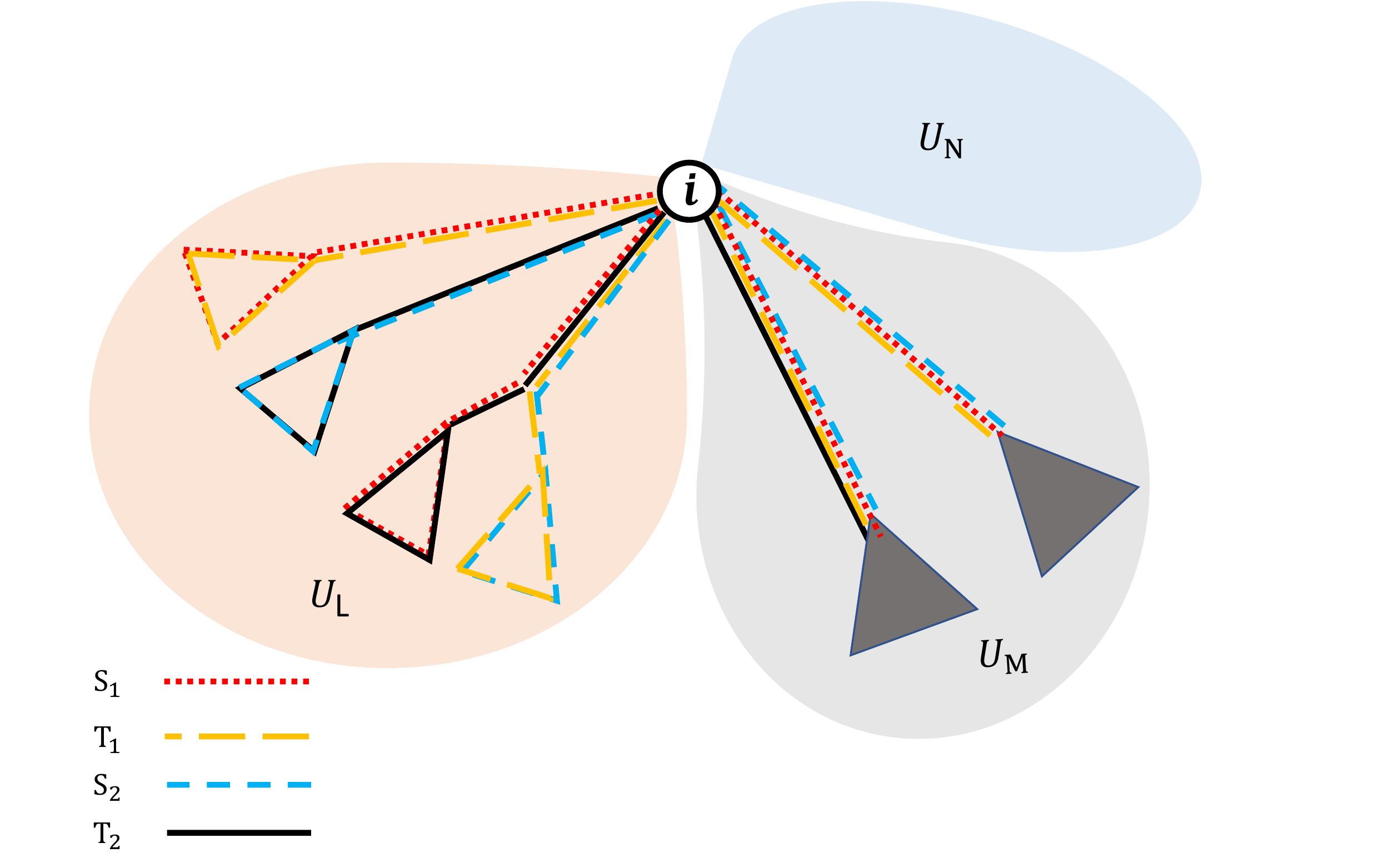}
    \caption{An illustration of decomposition of $U$ as $\UTone \cup \UTtwo \cup \UN$. In $\UTone$, each transparent triangle represents a $2$-decorated bulb with decorations indicated by the line types of its border. 
    In $\UTtwo$, each solid triangle represents a subtree whose edges are all at least $2$-decorated. The decoration of edges in the wires are indicated by their line types. 
   }
    \label{fig:decomp_true}
\end{figure}
\begin{claim}\label{claim:U_decompose_true}
\begin{enumerate}[label=(\roman*)]

  \item \label{T:0} 
  For any neighbor $a$ of $i$ in $U$, we have
  $\calC(i,a) \supset (S_1)_a \cup (T_1)_a \cup (S_2)_a \cup (T_2)_a$.
  (Recall that $(S_1)_a$ denotes the subtree of $S_1$ consisting of all descendants of $a$.)
  
  \item \label{T:1} For all $a \in \calNT$, $(S_1)_a \cup (T_1)_a \cup (S_2)_a \cup (T_2)_a = \calC(i,a) $, and the number of vertices in $\calC(i,a)$ is between $M+K$ and $2M+2K-1$. 
  
  \item \label{T:2}
  $\UTone$, $\UTtwo$ are two trees  that only share one common vertex $i$;
  $\UN$ and $\UTone$ (resp.~$\UTtwo$) share only one common vertex $i$. 
  
  \item \label{T:3}
   $\UN$ is a connected graph with the same excess as $U$ and $e(\UN) \leq |\calNN|(2K+2M)$. 
   \item \label{T:4} $ w(\UTone)w(\UTtwo)w(\UN) =(\aut(S_1)\aut(T_1)\aut(S_2)\aut(T_2))^{\frac{1}{2}}$, $w(\UTtwo) \le \R^{2|\calNTtwo|}$ and $w(\UN) \le \R^{2|\calNN|}$. 
\end{enumerate}
\end{claim}
In addition, we have the following lemma, showing that $|\calNN|$ is small when the excess $k$ is small. This will be useful for bounding the size of $w(\UN)$ and facilitating the enumeration of $\UN$. 
\begin{lemma}\label{lmm:k_true}
        For any $-1\le k\le N-1$, we have
        \begin{align*}
            |\calNN|  \le  2k+2 \, .
        \end{align*}
\end{lemma}

Next, let $\calU(v,k)$ denote the set of 
$U = S_1\cup T_ 1 \cup S_2\cup T_2$ for some $(S_1,T_1,S_2,T_2)\in \calW_{ii}(v,k)$.  
Let $\calUTone(v_1,e_1,\ell_1)$  denote the set of all possible $\UTone$ that has $v_1$ vertices excluding $i$, $\ell_1 = e(K_{11}\cap \UTone)$ and 
\[ e_1 = \frac{1}{2} \left( e( (K_{22}\cup K_{21}) \cap \UTone) +  e((K_{22}\cup K_{12}) \cap \UTone) \right)  \,,
\]
 for some $v_1,e_1,\ell_1\in\naturals$. 
 
Let $\calUTtwo(v_2,e_2)$ denote the set of all possible $\UTtwo$ that has $v_2$ vertices excluding $i$, and 
\[e_2 = \frac{1}{2} \left( e( (K_{22}\cup K_{21}) \cap \UTtwo) +  e((K_{22}\cup K_{12}) \cap \UTtwo) \right) \, .
\] 
for some $v_2, e_2\in\naturals$.  

Let $\calUN(v_N, k)$ denote the set of all possible $\UN$ with $v_N$ vertices excluding $i$ and excess $k$,  
for some $v_N \in\naturals$. 

Then, by \prettyref{eq:P_ij} and \prettyref{claim:U_decompose_true}~\ref{T:4}, for any $v\le 2N-k-1$ and $k\ge -1$, we have
\begin{align}
     P_{ii}(v,k)
     & =  ~ \sum_{U\in \calU(v,k)}   |\rho|^{e(K_{11})} w(\UTone)w(\UTtwo)w(\UN) \nonumber\\
     &  \overset{(a)}{\le} ~ \sum_{v_1,v_2,v_N \ge 0 } 
     \sum_{e_1,e_2\ge 0} \sum_{ \ell =0} |\rho|^{\ell} \PTone(v_1,e_1,\ell)  \PTtwo(v_2, e_2) \PN(v_N, k) \nonumber \\
     & ~~~~~~\indc{v_1+v_2 + v_N =v}\indc{e_1+e_2 \le 2N -(v+k+1)}  
    \,,
     \label{eq:U_decompose_true_bound}
 \end{align}
where
\begin{align}
    \PTone(v_1,e_1,\ell)
    & \triangleq ~ \sum_{\UTone \in \calUTone(v_1,e_1,\ell)}   w(\UTone) \label{eq:P_T_1_true} \\
    \PTtwo(v_2, e_2)
    & \triangleq ~\sum_{\UTtwo \in \calUTtwo(v_2,e_2)} w(\UTtwo) \label{eq:P_T_2_true} \\
    \PN(v_N, k)
    & \triangleq ~\sum_{\UN \in \calUN(v_N,k)} w(\UN) \label{eq:P_N_true} \,;
\end{align}
The inequality $(a)$ holds due to the following two facts. First
given $U$ is the union of $\UTone, \UTtwo$ and $\UN$, by \prettyref{claim:U_decompose_true}~\ref{T:2}, we have $   v_1+v_2+v_N = v(U)= v-1$. 
Second, 
\begin{align*}
    e_1+e_2 \le  \frac{1}{2} \left( e(K_{22}\cup K_{21}) +  e(K_{22}\cup K_{12})\right) = 2N - (v+k+1) \,,
\end{align*}
where the first inequality holds because $\UTone$ and $\UTtwo$ are edge-disjoint, and $\UTone\cup \UTtwo\subset U $; the last equality holds by \prettyref{eq:S1_T1_S2_T2} and $e(U)=v+k+1$. 

Hence, it remains to separately bound $\PTone(v_1,e_1,\ell)$, $\PTtwo(v_2,e_2)$ and $\PN(v_N, k)$.

\paragraph{Bounding $\PN(v_N, k)$.} 
By \prettyref{lmm:k_true},  $|\calNN| \le 2k+2$.
Hence it follows from \prettyref{claim:U_decompose_true}~\ref{T:1} and~\ref{T:3}
that
\begin{align}
    v_N & \le |\calNN| (2K+2M) \le (2k+2)(2K+2M)\,. \label{eq:v_N_true} \\
    w(\UN) & \le \R^{2|\calNN|}\le \R^{4(k+1)} \,.  \label{eq:w_UN_true}
\end{align}
Next we enumerate $\UN$ using its number of vertices and edges.
By \prettyref{lmm:enum} and \prettyref{eq:Universal}, the total number of unlabeled non-decorated graphs $\UN$ with $v_N$ vertices and excess $k$ is bounded by 
\begin{align*}
    \Universal^{v_N} {\binom{\binom{v_N+1}{2}}{k+1}}\le \Universal^{v_N} (v_N+1)^{2(k+1)} \,.  
\end{align*}

Recall that each edge can be decorated at most $11$ ways and $\UN$ contain $v_N$ vertices in addition to  $i$. Therefore, for $k\ge -1$, 
\begin{align*}
        |\calUN(v_N,k)|
        &  \le ~   \binom{n}{v_N} (v_N+1)!\Universal^{v_N} (v_N+1)^{2(k+1)}
        11^{v_N+k+1}  \nonumber \\
        & \le ~  n^{v_N}\Universal^{v_N} (v_N+1)^{2k+3} 
        11^{v_N+k+1} \,.
\end{align*}
Hence, Combining the above inequality with \prettyref{eq:P_N_true}, \prettyref{eq:v_N_true}, \prettyref{eq:w_UN_true}  yields that 
\begin{align}
       \PN(v_N,k)
        & \le ~  \R^{4(k+1)}|\calUN(v_N,e_N,k)| \indc{v_N \le (2k+2)(2K+2M)} \nonumber\\
        & \le ~ n^{v_N} 
        (11\Universal)^{v_N} \left(\R^4 11(v_N+1)^{3} \right)^{k+1}\indc{v_N \le (2k+2)(2K+2M)}
     \,. 
        \label{eq:calUN_true}
\end{align}

\paragraph{Bounding $\PTtwo(v_2,e_2)$.} 
Note that $\{e \in E(\calG(i,a))  : |D_e| \ge 3\}
= e(\calG(i,a)\cap K_{21}) + e(\calG(i,a)\cap (K_{12}))+ e(\calG(i,a)\cap (K_{22}))$.
So by definition in \prettyref{eq:calN_T_2}, for any $a\in \calNTtwo$,  $ e(\calG(i,a)\cap (K_{21}\cup K_{22})) + e(\calG(i,a)\cap (K_{12}\cup K_{22})) \ge M$. 
Then, we have $|\calNTtwo| \le \frac{2e_2}{M}$. 
By \prettyref{claim:U_decompose_true}~\ref{T:1} and~\ref{T:4}, we have
\begin{align}
      v_2 
      & \le |\calNTtwo|(2K+2M) \le  2e_2 \left( \frac{2K+2M}{M} \right) \, . \label{eq:v_2_true} \\
w(\UTtwo) & \le \R^{2|\calNTtwo|}\le \R^{\frac{4e_2}{M}} \,.  \label{eq:w_U_T_2_true}
\end{align}
Recall that each edge can be decorated by at most $11$ ways, and there are at most $\Universal^{v_2}$ number of unlabeled rooted trees with $v_2$ edges, we have
\begin{align*}
    |\calUTtwo(v_2,e_2) |
    & \le \binom{n}{v_2}v_2! \Universal^{v_2} 11^{v_2}  \le n^{v_2} (11\Universal)^{v_2} \,.
\end{align*}
Therefore, combining the above inequality with \prettyref{eq:P_T_2_true}, \prettyref{eq:v_2_true} and  \prettyref{eq:w_U_T_2_true} yields that 
\begin{align}
    \PTtwo(v_2,e_2) 
    & \le ~  n^{v_2} (11\Universal)^{v_2}  \R^{\frac{4e_2}{M}} \indc{v_2\le 2 e_2 \left( \frac{2M+2K}{M} \right)} \,.  \label{eq:calU_T_2_true}
\end{align}

\paragraph{Bounding $ \PTone(v_1,e_1,\ell)$.} This is the most challenging part of the proof. 
In contrast to $\UN$ and $\UTtwo$, 
straightforward bounds to enumeration and weight will no longer suffice. Instead,
we show that $\UTone$ can be obtained from a certain  class of chandeliers by \textit{merging}
the wires appropriately. The result of this enumeration is summarized in the following lemma.

\begin{lemma} \label{lmm:U_T_1_true}
    For any $a,b\in\naturals$ with $\frac{b}{2} \le a \le b$, define
    \begin{align}
         f_{a,b} \triangleq ~ \binom{|\calJ|}{a} \binom{a}{b-a} \binom{2a-b}{(a-L) \vee 0,(a-L) \vee 0, (2L-b)\wedge (2a-b)} 6^{ (2L-a)\wedge a} 
         \,, \label{eq:f_t_tilde_t} 
    \end{align}
    For any $v_1, e_1, \ell \in \naturals$,  define
     \begin{align}
       t \triangleq \frac{v_1+e_1}{K+M}\,,\qquad  
       \tilde{t} \triangleq ~ t \wedge (2L-\indc{e_1=0}) \wedge \left(\floor{\frac{\ell+2e_1}{2(K+M)}}{} + L\right) \,. \label{eq:tilde_t}
    \end{align}
    If $|\calJ| \ge 16 L^2$, 
    then\footnote{In \prettyref{eq:U_T_1_true}, we write $a \mid b$ if  $a$ divides $b$.}
    \begin{align}
      \PTone(v_1,e_1,\ell) \le ~  2 
      n^{v_1} \left(1+ 2L^2 \right)^{e_1}
      f_{\tilde{t},t} \indc{K+M \mid v_1+e_1   }  \indc{ v_1+e_1 \le 2N } . \label{eq:U_T_1_true}
    \end{align}

\end{lemma}

As we will see later in the proof, the factor $f_{\tilde{t}, t}$ counts all possible (unlabeled) chandeliers to start with, and the factor $(1+2L^2)^{e_1}$ counts all possible ways of merging.

\paragraph{Assembling all pieces.} 
By the first condition in \prettyref{eq:true_pair_constraint},
we have $|\calJ|\ge 16 L^2$. Hence, combining \prettyref{eq:U_decompose_true_bound},   \prettyref{eq:calUN_true},
\prettyref{eq:calU_T_2_true},
and \prettyref{eq:U_T_1_true}, 
for any $v\le 2N-k-1$ and $k\ge -1$,  
\begin{align*}
     P_{ii}(v,k)
     & \le ~ 
     \sum_{v_1,v_2,v_N\ge 0}\sum_{e_1,e_2 \ge 0}   \indc{v_1+v_2+v_N =v}\indc{e_1+e_2 \le 2N-(v+k+1)}\nonumber \\
     &~~~~~  2  n^{v}  (11\Universal)^{v_2+v_N} \R^{\frac{4e_2}{M}} \left( \R^4 11(v_N+1)^{3} \right)^{k+1}  \left(1+ 2L^2 \right)^{e_1} \\
     &~~~~~ \left( \sum_{\ell \ge 0 }
   |\rho|^{\ell}  f_{\tilde{t},t}\indc{K+M \mid v_1+e_1   }  \indc{ v_1+e_1 \le 2N }  \right) \indc{ v_2 \le 2 e_2 \left( \frac{2M+2K}{M} \right) }  \indc{v_N \le (2k+2)(2K+2M)} \\
    & \le ~ 2^3 n^{v}  \left(\R^4 11(2N+1)^{3} (11\Universal)^{2(K+M)}\right)^{k+1} 
    \sum_{e_2\ge 0}  \left(\R^{\frac{4}{M}} (11\Universal)^{ \frac{4M+4K}{M}} \right)^{e_2}     \nonumber \\
     &~~~~~  \sum_{e_1\ge 0} \left(1+ 2L^2 \right)^{e_1}  \left(\sum_{v_1=0}^v \sum_{\ell \ge 0 }
   |\rho|^{\ell} f_{\tilde{t},t}\indc{K+M \mid v_1+e_1   }  \indc{ v_1+e_1 \le 2N }  \right)  \indc{e_1+e_2 \le 2N-(v+k+1)}
    \,,
\end{align*}
where the last inequality holds because $v_N+1\le v+1 \le 2N- k  \le 2N+1$ and
\[
\sum_{v_2=0}^{2e_2 \left( \frac{2M+2K}{M} \right)}(11\Universal)^{v_2} 
\le 2  (11\Universal)^{2e_2 \left( \frac{2M+2K}{M} \right) } \,, \quad \sum_{v_N=0}^{(2k+2)(2K+M)} (11\Universal)^{v_N} \le 2 (11 \Universal)^{(2k+2)(2K+M)}  \,.
\]
Finally, applying the following lemma, we arrive at the desired \prettyref{eq:W_v_k_ell}.
\begin{lemma} \label{lmm:sum_f_t_tilde}
    If $\frac{12 L^2}{\rho^{2(K+M)}(|\calJ|-2L)} \le \frac{1}{2} $ and $|\rho| \le \frac{3}{4}$, we have
    \begin{align}
        & \sum_{v_1\ge 0} \sum_{\ell \ge 0 } |\rho|^{\ell} f_{\tilde{t},t}\indc{K+M \mid v_1+e_1   }  \indc{ v_1+e_1 \le 2N } \nonumber \\
        & \le~ 
       8 \rho^{2N}  \left(\rho^{-2e_1} \indc{e_1 \neq 0}+ \frac{12L^2}{\rho^{2(K+M)}|\calJ|}\indc{e_1 = 0} \right) |\calT|^2\,.  \label{eq:sum_f_t_tilde}
    \end{align}  
\end{lemma}

\begin{proof}[Proof of \prettyref{claim:U_decompose_true}]

Recall that for any rooted tree $S$ with root $i$, $(S)_a$ is a subtree of $S$ that consists of all descendants of $a$ with respect to the original root at $i$. If $a$ is not a neighbor of $i$, then by default, $(S_1)_a$ is an empty graph. 

\begin{enumerate}[label=(\roman*)]  

    \item If $a$ is a neighbor of $i$ in $S_1$, then $(S_1)_a \subset \calC(i,a)$ by definition. Hence, it follows that 
    for any $a$ that is a neighbor of $i$ in $U$, $\calC(i,a) \supset (S_1)_a \cup (T_1)_a \cup (S_2)_a \cup (T_2)_a$ . 
     
    \item First, we show that for any  $a \in \calNT$, $\calC(i,a)=(S_1)_a \cup (S_2)_a \cup (T_1)_a \cup (T_2)_a$. By \ref{T:0}, it suffices to show 
    $\calC(i,a) \subset
    (S_1)_a \cup (S_2)_a \cup (T_1)_a \cup (T_2)_a$,
    which further boils down to check that $E(\calC(i,a)) \subset 
    E((S_1)_a \cup (S_2)_a \cup (T_1)_a \cup (T_2)_a)$, as  there are no isolated vertices.  
    If not, then 
    there must exist $e \in E(C(i,a))$ but 
    $e \notin E((S_1)_a \cup (S_2)_a \cup (T_1)_a \cup (T_2)_a)$.
    
    Note that if $e =(i,b)$ for some $b\neq a$, then $\calG(i,a)$ contains a cycle through $i$, contradicting the definition of $a \in \calNT$. Hence, $e$
    must belong to $(S_1)_b \cup (S_2)_b \cup (T_1)_b \cup (T_2)_b$ for some $b \neq a$ that is a neighbor of $i$. However, given $a\neq b$,
    this immediately implies that $\calC(i,a)$ contains $(i,b)$, which is a contradiction.  Therefore, we conclude that $\calC(i,a)=(S_1)_a \cup (S_2)_a \cup (T_1)_a \cup (T_2)_a$. 
    

    It remains  to show that $M+K-1\le e(\calC(i,a)) \le 2M+2K-2$, given $\calC(i,a)$ is a tree. Note that for each edge in $U$, we have $2 \le |D_e | \le 4$. Given $\calC(i,a)=(S_1)_a \cup (S_2)_a \cup (T_1)_a \cup (T_2)_a$, we must have $ \frac{1}{4 } m \le e(\calC(i,a))  \le \frac{1}{2} m $, where $m = e((S_1)_a) + e((S_2)_a) + e((T_1)_a) + e((T_2)_a)$. Moreover, $(S_1)_a$, if non-empty, is a tree whose root $a$ is connected to a $K$-bulb by a path with $M-1$ edges, where $e(S_1)_a = K+M-1$. Analogous arguments hold for $(T_1)_a, (S_2)_a, (T_2)_a$ respectively. Then, the claim follows. 

    \item  Note that
    $\calG(i,a)$ only share a common vertex $i$ across different $a$. It follows that $\UTone$ and $\UTtwo$ are two vertices with only one common vertex $i$. Next, we show that $\UN$ and $\UTone$ do not share any vertex other than $i$. 
    Suppose not and let $c \neq i$ denote the common vertex shared by $\UN$ and $\UTone$. Then $c$ must belong to $\calC(i,a)$ for some $a \in \calNTone$. Also, since $\UN$ is induced by edges not in $\UTone$ and $\UTtwo$, $c$ must be incident to an edge $e$ not in $\UTone$. However, by definition of $\calC(i,a)$, $e$ must belong to $\calC(i,a)$ and hence $\UTone$, which leads to a contradiction. 
    Analogously, we can show that $\UN$ and $\UTtwo$ do not share any vertex other than $i$. 
    
    \item 
    By~\ref{T:2} and the connectivity of $U$, $\UN$ must be connected. Moreover, we have that 
    \begin{align*}
        v(U) & = v(\UTone)+v(\UTtwo)+v(\UN) - 2  \\
        e(U)  & = e(\UTone)+e(\UTtwo)+e(\UN) \, .
    \end{align*}
    Combining the above with the facts that $\UTone$ and $\UTtwo$ are trees, we get that 
        $
        e(\UN)- v(\UN) =  e(U)-  v(U) =k \, .
        $
    To prove that the number of edges in $\UN$ is bounded by $|\calNN|(2K+2M)$, it suffices to show that 
    \begin{align}
        \UN  = \bigcup_{a\in\calNN} \left(\{(i,a)\} \cup (S_1)_a \cup (S_2)_a \cup (T_1)_a \cup (T_2)_a\right) \,. \label{eq:UN_rewrite}
    \end{align}
    Note that
    \[
     U = \bigcup_{a\in\calNN \cup \calNT} \left(\{(i,a)\} \cup (S_1)_a \cup (S_2)_a \cup (T_1)_a \cup (T_2)_a\right) \, .
    \]
    By \ref{T:0}, for any $a \in \calNT$, we have $\calG(i,a)= \{(i,a)\} \cup (S_1)_a \cup (S_2)_a \cup (T_1)_a \cup (T_2)_a$. 
    Then, by \prettyref{eq:U_decompose_true}, \prettyref{eq:UN_rewrite} follows directly.

    \item

    Fix a bulb $\calB \in \calK(S_1)$. 
    Then $\calB \subset (S_1)_a$ for some   neighbor $a$  of $i$ in $S_1$. 
    By~\ref{T:0} and \ref{T:1}, $(S_1)_a \subset \calC(i,a) \subset \calG(i,a)$. 
    By \prettyref{eq:U_decompose_true}, 
    $(S_1)_a$ is contained in  $\UTone$ (resp.~$\UTtwo$,  $\UN$), if and only if $a \in \calNTone$ (resp.~$\calNTtwo $, $\calNN$). 
    Therefore, $\calB \subset \UTone$ (resp.~$\UTtwo$,  $\UN$), if and only if $a \in \calNTone$ (resp.~$\calNTtwo $, $\calNN$). 
    In particular, for any bulb $\B$ from $S_1$, it must be contained in exactly one of $\UTone$, $\UTtwo$ or $\UN$. 
    
    Recall that that $\aut(S_1) = \prod_{\calB\in \calK(S_1)}\aut(\calB)$, given $S_1$ is a chandelier with $L$ non-isomorphic bulbs.  Analogous arguments hold for $T_1,S_2,T_2$. Together with \prettyref{eq:w_S_U'} and \prettyref{eq:w_U'}, it follows that
    \[
    (\aut(S_1)\aut(T_1)\aut(S_2)\aut(T_2))^{\frac{1}{2}} = w(\UTone)w(\UTtwo)w(\UN) \,.
    \]
    
    Note that each bulb $\calB \in \calK(S_1)$ must be contained in distinct $(S_1)_a$ for $a$ that is a neighbor of $i$ in $S_1$. There are at most $|\calNTtwo|$ (resp.~$|\calNN|$) bulbs from $S_1$ contained in $\UTtwo$ (resp.~$\UN$). Then,  given $\aut(\calB)\le \R$ for any $\calB \in \calK(S_1)$, by \prettyref{eq:w_S_U'}, we have
    $w_{S_1}(\UTtwo)\le \R^{\frac{1}{2}|\calNTtwo|}$ and $w_{S_1}(\UN)\le \R^{\frac{1}{2}|\calNN|}$. Analogous arguments holds for $T_1,S_2,T_2$.  Therefore, by \prettyref{eq:w_U'}, we have 
    \[
    w(\UTtwo) \le \R^{2|\calNTtwo|} \,,  \quad w(\UN) \le \R^{2|\calNN|} \,.
    \]
    

\end{enumerate}
\end{proof}

\subsubsection{Proof of \prettyref{lmm:k_true}} 

Note that by \prettyref{claim:U_decompose_true}~\ref{T:3}, $\UN$ is a connected graph with the same excess as $U$, which is $k$. Then, it is equivalent to show that
\begin{align}
  2(e(\UN) - v(\UN)+ 1) \ge ~ |\calNN|\, . \label{eq:target_1_true}
\end{align}
Define 
\begin{align*}
    \calNNone& = ~ \{a \in \calNN(i): \text{ $ \calC(i,a) $ does not contain $i$}\} \, .
\end{align*}
Let 
\[
\UNone = ~ \bigcup_{a\in \calNNone} \calG(i,a) \,. 
\]
Note that for any $a \neq b \in \calNNone(i)$, $\calG(i,a)$ and $\calG(i,b)$ are edge disjoint and share the only one common vertex $i$, for otherwise $\calC(i,a)$ would contain $i$. Since $a \notin \calNT(i)$ (recall the definition of $\calNT(i)$ in \prettyref{eq:calN_T_i_fake}), 
each $\calG(i,a)$ has at least one cycle and hence $e(\calG(i,a)) \geq v(\calG(i,a))$. Thus, we have 
\begin{align}
    e(\UNone)- v(\UNone)+1 & \ge ~ |\calNNone| \, . \label{eq:UN_1}
\end{align}

If $\calNNone = \calNN$, our desired result follows directly.  Without loss of generality, we assume that $\calNNone \neq \calNN$.  Define
\begin{align}
    \calNNtwo  = ~ \calNN\backslash \calNNone \, . \label{eq:calN_N_2_true}
\end{align}
Let $ \UNtwo = \UN \backslash U_{N_1 } .  $ 
By definition of $\calNNone$, we have $\UNone$ only share a common vertex with $\UNtwo$. Moreover, since $\UN$ is a connected graph, $\UNtwo$ must also be a connected graph. 
Then, we have $e(\UNtwo) = e(\UN) - e(\UNone)$ and $v(\UNtwo) = v(\UN) - v(\UNone) + 1$. Then by \prettyref{eq:target_1_true} and \prettyref{eq:UN_1}, it suffices to show that 
\begin{align}
    2(e(\UNtwo) - v(\UNtwo) +  1) \ge  |\calNNtwo| \, . \label{eq:target_2_true}
\end{align}

Analogous to the proof of \prettyref{lmm:k_fake} in \prettyref{sec:k_fake}, we apply the cycle basis argument. Note that $\UNtwo$ is a connected graph. Then, it has a cycle basis of size $e(\UNtwo)- v(\UNtwo) + 1 \equiv m$. We claim that that for each $a\in \calNNtwo$, there exists a cycle containing the edge $(i,a)$. 
Therefore, this edge is contained in one of the cycle in the cycle basis. Since each cycle contains at most two edges connecting $i$, we have $2m + 2 \geq |\calNNtwo|$. This completes the proof of \prettyref{eq:target_2_true} and hence the lemma. 

It remains to justify the above claim. 
Indeed, by \prettyref{eq:calN_N_2_true}, we have 
\begin{align*}
    \calNNtwo =  ~ \{a \in \calNN(i): \text{ $ \calC(i,a) $ contains $i$}\} \, . 
\end{align*}
Then, there must be a path from $a$ to $i$ in $\calC(i,a)\subset \UNtwo$, and hence a cycle containing $(i,a)$
\subsubsection{Proof of \prettyref{lmm:U_T_1_true}}



If a chandelier has $t$ branches, we call it a $t$-chandelier.  
Before bounding $\PTone(v,k,\ell)$, we describe an auxiliary family $\calF_{t',t}$
of \emph{unlabeled} decorated graphs $W$ rooted at $i$ such that $W$ is a $t$-chandelier satisfying the following conditions: 
\begin{enumerate}
    \item There are in total $t'$ non-isomorphic bulbs where $\frac{t}{2}\le t' \le t$;
    \item Each bulb and its connecting wire share the same decoration $D_b \subset \{\SS_1,\TT_1,\SS_2,\TT_2\}$, where $|D_b|=2$;
    \item Each non-isomorphic bulb appears at most twice in $W$. If a bulb appears twice in $W$ and its decoration in one bulb is $D_b \subset \{\SS_1,\TT_1,\SS_2,\TT_2\}$, then its decoration in the other bulb must be $D_b^c = \{\SS_1,\TT_1,\SS_2,\TT_2\}\backslash D_b$, where $|D_b|=|D_b^c| = 2$; 
    \item There are at least $(t'-L)\vee 0$ non-isomorphic bulbs decorated by $\{\SS_1, \SS_2\}$ and at least $(t'-L)\vee 0$ non-isomorphic bulbs decorated by $\{\TT_1, \TT_2\}$, and all these bulbs appear only once in $W$.
\end{enumerate}

\begin{lemma}\label{lmm:decoupled_chandelier_tree}
Let $f_{t',t}$ be defined by \prettyref{eq:f_t_tilde_t}. Then $|\calF_{t',t}| 
     \le f_{t',t}$.
\end{lemma}
\begin{proof}
For each $W$, there are $t-t'$ non-isomorphic bulbs that appear twice in $W$, and $2t'-t$ non-isomorphic bulbs that appear once in $W$. 
Hence, the total number of unlabeled non-decorated trees $W$ is upper bounded by $\binom{|\calJ|}{t'} \binom{t'}{t-t'}$. 

Given a non-decorated $W$, to see how many ways we can decorate it, we fix a subtree $W'$ of $W$ that is a $t'$-chandelier with $t'$ non-isomorphic bulbs. 
Then $W'$ contains all the bulbs that appear once in $W$. 
For each bulb that appears twice in $W$, it must appear both inside and outside $W'$. If its decoration inside $W'$ is $D_b\subset\{\SS_1,\TT_1,\SS_2,\TT_2\}$ where $|D_b| = 2$, then its decoration outside $W'$ must be $D_b^c$.
Therefore, to decorate $W$, it suffices to decorate $W'$.

Note that $W'$ 
has $2t'-t$ bulbs that appear once in $W$, among which at least $(t'-L)\vee 0$ bulbs are decorated by $\SS_1$ and $\SS_2$, and at least $(t'-L)\vee 0$ bulbs are decorated by $\TT_1$ and $\TT_2$. For the remaining bulbs in $W'$, there are at most $\binom{4}{2}=6$ ways to decorate each of them. 
Hence, the total number of ways to decorate $W'$ is upper bounded by
\begin{align*}
    & \binom{2t'-t}{(t'-L) \vee 0,(t'-L) \vee 0, (2L-t) \wedge (2t'-t)}  6^{ (2L-t')\wedge t'}  \, .
\end{align*}
Consequently, by the definition of $f_{t',t}$ in \prettyref{eq:f_t_tilde_t}, \prettyref{lmm:decoupled_chandelier_tree} is proved.
\end{proof}


The key is to recognize that any unlabeled $\UTone \in \calUTone(v_1,e_1,\ell)$ can be generated by \textit{merging} a chandelier in $\calF_{t',t}$ for some $t' \le t$. To this end, let us  first describe this merging procedure for $2$-chandeliers, as illustrated in~\prettyref{fig:merging}.
Let $T$ denote an unlabeled $2$-chandelier, with one branch (a bulb $\B_1$ and the attached wire) decorated by $D_{1}$, and the other branch (a bulb $\B_2$ and the attached wire) decorated by $D_{2}$. 
Consider the path $P_1$ with $m$ edges starting from the root on one wire, and another path $P_2$ with $m$ edges starting from the root on the other wire. 
We merge $P_1$ and $P_2$ so that they become a wire of length $m$ decorated by $D_1\cup D_2$ throughout.
After merging, we call the resulting unlabeled $T'$ as a \emph{tangled chandelier}. (The analogy is that two bulbs hanging from the ceiling have their wires partly tangled up.)
In addition, the merged wire is called the \textit{parent wire} and the unmerged wires the \textit{children wires}. We denote the merging procedure as $\calM_m$, where $T' = \calM_m(T)$ and $m$ is the number of merged edges. 
We call the two paths of length $M-m$ attached to the parent wire as the children wires. 

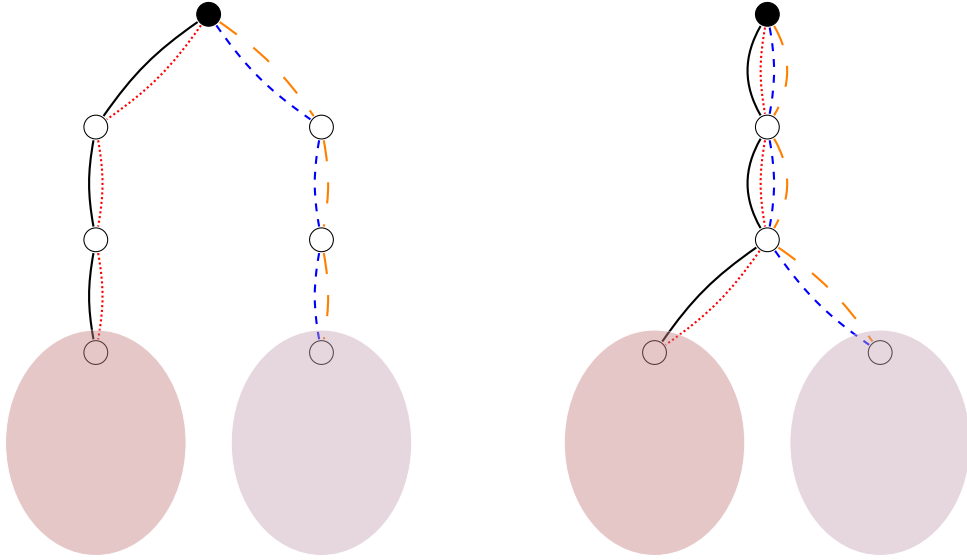
\begin{figure}[ht]
\centering
\begin{tikzpicture}[scale=3,transform shape]
\draw (0,0) node (root) [root] {};
\draw (-0.5,-0.5) node (a1) [vertexdot] {};
\draw (-0.5,-1) node (a2) [vertexdot] {};
\draw (-0.5,-1.5) node (a3) [vertexdot] {};
\draw (0.5,-0.5) node (c1) [vertexdot] {};
\draw (0.5,-1) node (c2) [vertexdot] {};
\draw (0.5,-1.5) node (c3) [vertexdot] {};

\draw (root) edge[black, bend right=10, thick] (a1);
\draw (root) edge[red, bend left=10, thick, densely dotted] (a1);
\draw (a1) edge[black, bend right=10, thick] (a2);
\draw (a1) edge[red, bend left=10, thick, densely dotted] (a2);
\draw (a2) edge[black, bend right=10, thick] (a3);
\draw (a2) edge[red, bend left=10, thick, densely dotted] (a3);

\draw (root) edge[blue, bend right=10, thick,  dashed] (c1);
\draw (root) edge[orange, bend left=10, thick, dash pattern={on 8pt off 8pt on 8pt off 8pt} ] (c1);
\draw (c1) edge[blue, bend right=10, thick,  dashed] (c2);
\draw (c1) edge[orange, bend left=10, thick, dash pattern={on 8pt off 8pt on 8pt off 8pt} ] (c2);
\draw (c2) edge[blue, bend right=10, thick,  dashed] (c3);
\draw (c2) edge[orange, bend left=10, thick, dash pattern={on 8pt off 8pt on 8pt off 8pt} ] (c3);

\filldraw[nearly transparent, color=white, fill=black!50](-0.5,-1.9) ellipse (0.4 and 0.5);

\filldraw[nearly transparent, color=white, fill=red!50](-0.5,-1.9) ellipse (0.4 and 0.5);

\filldraw[nearly transparent, color=white, fill=blue!50](0.5,-1.9) ellipse (0.4 and 0.5);
\filldraw[nearly transparent, color=white, fill=orange!50](0.5,-1.9) ellipse (0.4 and 0.5);
\end{tikzpicture}
\hspace{50pt}
\begin{tikzpicture}[scale=3,transform shape,auto,font=\scriptsize]
\draw (0,0) node (root) [root] {};
\draw (0,-0.5) node (a1) [vertexdot] {};
\draw (0,-1) node (a2) [vertexdot] {};
\draw (-0.5,-1.5) node (a3) [vertexdot] {};
\draw (0.5,-1.5) node (c3) [vertexdot] {};

\draw (root) edge[black, bend right=30, thick] (a1);
\draw (root) edge[red, bend right=10, thick, densely dotted] (a1);
\draw (a1) edge[black, bend right=30, thick] (a2);
\draw (a1) edge[red, bend right=10, thick, densely dotted] (a2);
\draw (a2) edge[black, bend right=10, thick] (a3);
\draw (a2) edge[red, bend left=10, thick, densely dotted] (a3);

\draw (root) edge[blue, bend left=10, thick, dashed] (a1);
\draw (root) edge[orange, bend left=30, thick, dash pattern={on 8pt off 8pt on 8pt off 8pt} ] (a1);
\draw (a1) edge[blue, bend left=10, thick, dashed] (a2);
\draw (a1) edge[orange, bend left=30, thick, dash pattern={on 8pt off 8pt on 8pt off 8pt} ] (a2);
\draw (a2) edge[blue, bend right=10, thick, dashed] (c3);
\draw (a2) edge[orange, bend left=10, thick, dash pattern={on 8pt off 8pt on 8pt off 8pt} ] (c3);

\filldraw[nearly transparent, color=white, fill=black!50](-0.5,-1.9) ellipse (0.4 and 0.5);

\filldraw[nearly transparent, color=white, fill=red!50](-0.5,-1.9) ellipse (0.4 and 0.5);

\filldraw[nearly transparent, color=white, fill=blue!50](0.5,-1.9) ellipse (0.4 and 0.5);
\filldraw[nearly transparent, color=white, fill=orange!50](0.5,-1.9) ellipse (0.4 and 0.5);
\end{tikzpicture}
\caption{A step in the merging procedure. {Left:} A 2-chandelier. {Right:} A tangled chandelier obtained from the 2-chandelier by merging the first two edges in each of the two wires. Here $M=3$ and $m=2$. The 4-decorated edges form the parent wire, and the 2-decorated edges form the children wires. 
} 
\label{fig:merging} 
\end{figure}

More generally, we can merge an unlabeled $t$-chandelier for any integer $t \ge 2$, by repeatedly merging $2$-chandeliers, to obtain a union of $1$-chandeliers and tangled chandeliers with a common root. 
For any unlabeled $t$-chandelier tree $T$, let $\calM_e(T)$ denote the set of all possible unlabeled trees that can be obtained from this procedure by merging a total of $e$ edges. Then, we have the following lemma.
\begin{lemma}\label{lmm:merge_bound}
For any unlabeled $t$-chandelier tree $T$, 
\begin{align}
|\calM_e(T)| \le~ \left(1+ t^2/2 \right)^{e} \, . \label{eq:merge_bound}
\end{align}
\end{lemma}
\begin{proof}
To generate a tree by merging $e$ edges in $T$, we need to pick $\hat{t}$ number of pairs of wires to merge, where $0\le \hat{t} \le e$. 
The total number of ways to pick $\hat{t}$ pairs of wires is equal to 
$\binom{t}{2\hat{t}}(2\hat{t})!/2^{\hat{t}}$.  After fixing $\hat{t}$ pairs of wires, it remains to specify the length of the parent wires after merging the $\hat{t}$ pairs of wires, \ie, $(m_1, m_2, \ldots, m_{\hat{t}})$ such that $\sum_{j=1}^{\hat{t}}m_j=e$ and $m_j \ge 1$ for all $1 \le j \le \hat{t}$. The number of such sequences is same as the number of ways of putting $\hat{t}-1$ bars among $e$ items, which is $\binom{e-1}{\hat{t}-1}\le \binom{e}{\hat{t}}$. 
Hence, we have that
\begin{align*}
|\calM_e(T)| \le \sum_{\hat{t} = 0}^{e} \binom{t}{2\hat{t}}\frac{(2\hat{t})!}{2^{\hat{t}}}\binom{e}{\hat{t}}  \le 
\sum_{0\le \hat{t} \le e } \binom{e}{\hat{t}}  \left(\frac{1}{2}t^2\right)^{ \hat{t}}
\le \left(1+ t^2/2 \right)^{e }.
\end{align*}
\end{proof}

Then, we present the following lemma. 
\begin{lemma}\label{lmm:calU_T_1_bound}
If $\calUTone(v_1,e_1,\ell) \neq \emptyset$, then $K+M$ divides $v_1+e_1$ and $v_1+e_1 \le 2N$.
Furthermore, any 
$\UTone \in \calUTone(v_1,e_1,\ell)$ satisfies $ \UTone \in \calM_{e_1}(W)$ for some $W \in \bigcup_{t' \le \tilde{t}}\calF_{t',t}$ with $t$ and $\tilde{t}$ defined in \prettyref{eq:tilde_t}, and 
$w(\UTone) \leq \aut(\UTone)$.
\end{lemma}

Let $\widetilde{\calUTone}(v_1,e_1, \ell)$ denote the set of unlabeled $\UTone \in \calUTone(v_1,e_1, \ell)$. Then, by \prettyref{eq:P_T_1_true}, we have

\begin{align*}
    \PTone(v_1,e_1,\ell) 
     & \le ~  \sum_{\UTone\in \widetilde{\calUTone}(v_1,e_1,\ell)} w(\UTone) \binom{n}{v_1}\frac{v_1!}{\aut(\UTone)}\\
     & \le ~ \binom{n}{v_1}v_1! |\widetilde{\calUTone}(v_1,e_1,\ell)| \\
     & \le ~ 
     2 
      n^{v_1} \left(1+ 2L^2\right)^{e_1}
     f_{\tilde{t},t} \indc{K+M \mid v_1+e_1   }  \indc{ v_1+e_1 \le 2N } \,,
\end{align*}
where the second inequality holds because $\aut(\UTone) \ge w(\UTone)$ by \prettyref{lmm:calU_T_1_bound}, and the last inequality holds because 
\begin{align*}
    |\widetilde{\calUTone}(v_1,e_1,\ell)| 
    & \overset{(a)}{\le }~
     \sum_{t'=0}^{\tilde{t}} |\calF_{t',t}| \left(\max_{T\in\calF_{t',t} }|\calM_{e_1}(T)|\right) \indc{K+M \mid v_1+e_1   }  \indc{ v_1+e_1 \le 2N }   \\
    & \overset{(b)}{\le }~ ~ (1+2L^2)^{e_1} \sum_{t'=0}^{\tilde{t}} f_{t',t} \indc{K+M \mid v_1+e_1   }  \indc{ v_1+e_1 \le 2N } \\
    & \overset{(c)}{\le }~ ~ 2 (1+2L^2)^{e_1}  f_{\tilde{t},t} \indc{K+M \mid v_1+e_1   }  \indc{ v_1+e_1 \le 2N } \,,
\end{align*}
where $(a)$ holds by \prettyref{lmm:calU_T_1_bound}; $(b)$ holds by $t \le 2L$, \prettyref{lmm:decoupled_chandelier_tree}, and \prettyref{eq:merge_bound}; $(c)$ follows from $\frac{f_{t',t}}{f_{t'+1,t}} \le \frac{1}{2}$ by \prettyref{eq:f_a_b_mono} in \prettyref{claim:f_a_b} and the assumption $|\calJ| \ge 16L^2$. 
This completes the proof of \prettyref{lmm:U_T_1_true}.



It remains to prove \prettyref{lmm:calU_T_1_bound}. 
\begin{proof}[Proof of \prettyref{lmm:calU_T_1_bound}]
Fix any $\UTone \in \calUTone(v_1,e_1,\ell)$.
Let $t'$ denote the number of non-isomorphic bulbs in $\UTone$.
Then, we have the following claim. 
\begin{claim}\label{claim:U_T_1_structure}
\begin{enumerate}[label=(\roman*)]
    \item \label{UT:1} For all $a \in \calNTone$, we have $|D_{i,a}| = 2$ or $4$. 
    \begin{itemize}
        \item If $|D_{i,a}| = 2$, then $\calG(i,a)$ is a $1$-chandelier which is decorated by $D_{i,a}$ throughout. 
        \item If $|D_{i,a}| = 4$, then $\calG(i,a)$ is a tangled chandelier such that each edge on the parent wire is $4$-decorated, and each edge outside the parent wire is $2$-decorated. Moreover, each child wire and its incident bulb have the same decoration. %
    \end{itemize}
      Overall, every bulb in $\UTone$ is 2-decorated.
    \item \label{UT:2} Each non-isomorphic bulb appears at most twice in $\UTone$. If a bulb appears twice in $\UTone$ and its decoration in one bulb is $D_b \subset \{\SS_1,\TT_1,\SS_2,\TT_2\}$, then its decoration in another bulb must be $D_b^c = \{\SS_1,\TT_1,\SS_2,\TT_2\}\backslash D_b$, where $|D_b|=|D_b^c| = 2$;
    
    \item \label{UT:3} There are at least $(t'-L)\vee 0$ non-isomorphic bulbs decorated by $\{\SS_1, \SS_2\}$ and at least $(t'-L)\vee 0$ non-isomorphic bulbs decorated by $\{\TT_1, \TT_2\}$, and all these bulbs appear only once in $W$.  
\end{enumerate}
\end{claim} 

By \prettyref{claim:U_T_1_structure}~\ref{UT:1}, 
$\UTone$ consists of 1-chandeliers and tangled chandeliers. So the unlabeled
$\UTone$ can be generating by applying the merging procedure to some unlabeled $t$-chandelier $W$.
Moreover, $|D_e|=2$ or $4$ for each edge $e$ in $\UTone$.
Recall that $K_{12}$ and $K_{21}$ are both 3-decorated. So $E(\UTone\cap K_{21}) = \emptyset$ and $E(\UTone\cap K_{12}) = \emptyset$ and then
\begin{align}
    e_1 = \frac{1}{2} \left( e( (K_{22}\cup K_{21}) \cap \UTone) +  e((K_{22}\cup K_{12}) \cap \UTone) \right)  = e(\UTone\cap K_{22}) \,. \label{eq:e_1}
\end{align}
In other words, $e_1$ is precisely the number of $4$-decorated edges in $\UTone$, which are obtained by merging 2-decorated edges in $W$. Hence, $\UTone \in \calM_{e_1}(W)$. 
As a result, $e(W) = e(\UTone)+e_1=v_1+e_1$. Since $W$ is a $t$-chandelier with $t(K+M)$ edges, it follows that $K+M$ divides $v_1+e_1$, and $t = \frac{v_1+e_1}{K+M}$.  

Let $t'$ denote the number of non-isomorphic bulbs in $\UTone$. By \prettyref{claim:U_T_1_structure}~\ref{UT:2}, each non-isomorphic bulb must appear at most twice . Since there are $t$ bulbs in total, we have $\frac{t}{2} \le t' \le t$. Hence, by \prettyref{claim:U_T_1_structure}~\ref{UT:1}-\ref{UT:3}, it follows that $W \in \calF_{t',t}$. Moreover, we have 
$$
2v_1+2e_1 = \sum_{e \in E(\UTone)} |D_e| \le  \sum_{e \in E(U)} |D_e| = 4N,
$$
and hence $v_1+e_1 \le 2N$. 

Next, we show that $t' \le \tilde{t}$, where $\tilde{t}$ was defined in \prettyref{eq:tilde_t}. First, we show that $t' \le \floor{\frac{\ell+2e_1}{2(K+M)}}{}+L$. It suffices to show
\begin{align}
    \ell\ge 2 ((t'-L)\vee 0) (K+M) -2 e_1 \label{eq:ell_1_2_e_1}\,.
\end{align}
Fix a bulb $\B$  in $\UTone$. Then, $\B\subset \calG(i,a)$ for some $a \in \calNTone$. If $\calG(i,a)$ is a $1$-chandelier, then the wire connecting $\B$ should have the same decoration as $\B$. If $\calG(i,a)$ is a tangled chandelier, then the children wire connecting $\B$ should have the same decoration as $\B$, while the parent wire is $4$-decorated. 
By \prettyref{claim:U_T_1_structure}~\ref{UT:3}, the total number of bulbs that are decorated by $\{\SS_1, \SS_2\}$ or $\{\TT_1, \TT_2\}$ is at least $2 ((t'-L)\vee 0)$. Hence, in total, there are at least $2 ((t'-L)\vee 0)(K+M)-2 e_1$ edges are decorated by $\{\SS_1, \SS_2\}$ or $\{\TT_1, \TT_2\}$, where we subtract $2 e_1$ because some edges
may appear in parent wires and hence are $4$-decorated.
Moreover, for any edge $e$ in $\UTone$, if it is colored by $\SS_1$ and $\SS_2$ (resp.~$\TT_1$ and $\TT_2$), it must belong to $K_{11}$. Then, \prettyref{eq:ell_1_2_e_1} follows.


Second, we show that $t' \le 2L-\indc{e_1=0}$. 
It suffices to show that if $e_1 = 0$, $t'\le 2L-1$. Suppose not. Then, we have $t'= t=2L$. 
By \prettyref{claim:U_T_1_structure}~\ref{UT:3}, we have $L$ number of bulbs are decorated by $\{\SS_1, \SS_2\}$, and the other $L$ number of bulbs are decorated $\{\TT_1,\TT_2\}$. Hence, we have $S_1=S_2$ and $T_1=T_2$ in this case, and
$U=\UTone$. Then, by \prettyref{eq:e_1}, we have $e_1 = e(K_{22}) = e(S_1\cap T_1)$. 
If $e_1=0$,  $E(S_1\cap T_1) =\emptyset$. Given $S_1\cup T_1 $ is a tree, then $V(S_1\cap T_1) = \{i\}$,  which contradicts with our assumption that $V(S_1\cap T_1) \neq \{i\}$, given $S_1= S_2$ and $T_1=T_2$,  in \prettyref{eq:constraint2} for any $(S_1,T_1,S_2,T_2)\in \calW_{ii}$. 
Hence, it follows that $t'\le \tilde{t}$.  

Lastly, we show that $\aut(\UTone) \ge w(\UTone)$. 
By \prettyref{claim:U_T_1_structure}~\ref{UT:1}, for each $a \in \calNTone$,
$\calG(i,a)$ is either a $1$-chandelier or a tangled chandelier. 
Let $\calK(\calG(i,a))$ denote the set of $K$-bulbs in $\calG(i,a)$. 
For any bulb $\calB \in \calK(\calG(i,a))$, each of its automorphism 
naturally induces a distinct automorphism of the entire 
$\UTone$ by permuting the node labels of $\B$ according to the automorphism and keeps all other node labels unchanged. Therefore,\footnote{In general, \prettyref{eq:autUT1}  can be an strict inequality, for example, when a tangled chandelier has two isomorphic bulbs, in which case swapping the two bulbs and the attached children wires is another automorphism.}
\begin{equation}
\aut(\UTone) \ge \prod_{a \in \calNTone} \prod_{\calB\in \calK(\calG(i,a))} \aut(\calB).     \label{eq:autUT1}
\end{equation}

Next, we show that $w(\UTone) = \prod_{a \in \calNTone} \prod_{\calB\in \calK(\calG(i,a))} \aut(\calB)$, 
where $w(\UTone)$ is defined per \prettyref{eq:w_U'}. 
Fix any bulb $\B$ from $S_1$ contained in $\UTone$. Then, $\B \subset (S_1)_a \subset \calG(i,a)$ for some neighbor $a$ of $i$ in $S_1$.  By \prettyref{claim:U_T_1_structure}~\ref{UT:1}, $\calB$ in $\calG(i,a)$ must be $2$-decorated and has the same decoration for all its edges. 
Without loss of generality, assume that $\calB$ in $\calG(i,a)$ is decorated by $\{\SS_1,\TT_1\}$. 
By \prettyref{claim:U_T_1_structure}~\ref{UT:1}, it follows that $(S_1)_a=(T_1)_a$
and thus 
$\B \in \calK(T_1)$.  
Analogous arguments hold for bulbs from $T_1$, $S_2$ and $T_2$. 
Hence
\begin{align*}
w(\UTone) & =\prod_{\calB \in \calK(S_1): \calB \subset \calUTone} \aut^{1/2}(\calB)
\prod_{\calB \in \calK(T_1): \calB \subset \calUTone} \aut^{1/2}(\calB)
\prod_{\calB \in \calK(S_2): \calB \subset \calUTone} \aut^{1/2}(\calB)
\prod_{\calB \in \calK(T_2): \calB \subset \calUTone} \aut^{1/2}(\calB) \\
& =\prod_{a \in \calNTone} \prod_{\calB\in \calK(\calG(i,a))} \aut(\calB).
\end{align*}
In conclusion, we have that $\aut(\UTone) \ge w(\UTone)$.

It remains to prove \prettyref{claim:U_T_1_structure}.
\begin{enumerate}[label=(\roman*)]
    \item 
   
    Fix $a \in \calNTone $. Without loss of generality, assume that $(i,a) \in S_1$. Then $\calG(i,a) \cap S_1$ consists of an $M$-wire $P \equiv (i,a,\ldots,b)$ and a bulb $\calB \equiv (S_1)_b$ rooted at $b$. Consider an edge $e$ in $\calB$. Then there is a path $P'$ in $\calB$ from the root $b$ to $e$. Since $\calG(i,a)$ is a tree, the path $(P,P')$ is the unique path from $i$ to $e$ in $S_1 \cup T_1\cup S_2 \cup T_2$. If $e\in T_1$, then there exists a path $P''$ from $i$ to $e$ in $T_1$. By the uniqueness of the path $(P,P')$, $(P,P')$ must coincide with $P''$
    and hence belong to $T_1$. Analogous argument holds for $S_2,T_2$. 
    It follows that any edge $e'$ on $(P,P')$ in $\calG(i,a)$ must have decoration $D_{e'} \supset D_e$. We further claim that $|D_e| =2$.  If instead $|D_e|\ge 3$, then $P$ is $3$-decorated in $\calG(i,a)$, and hence there are at least $M$ edges  that are at least $3$-decorated in $\calG(i,a)$, which contradicts with \prettyref{eq:purebulb}. 
    
    Next, we show that $\calB$ in $\calG(i,a)$ must be decorated by $D_e$. 
    Suppose that there is another edge $\tilde{e}$ in $\calB$ has decoration $D_{\tilde{e}} \neq D_e$. Let $\tilde{P}$ denote the path in $\calB$ from the root $b$ to $\tilde{e}$. Similarly, we can show that any edge $e'$ on $(P,\tilde{P})$ in $\calG(i,a)$ must have decoration $D_{e'} \supset D_{\tilde{e}}$. Hence, any edge $e'$ on $P$ must have $ D_{e'} \supset D_e \cup D_{\tilde{e}}$. Since  $|D_e|, |D_{\tilde{e}}| \ge 2$ and $D_e \neq D_{\tilde{e}}$, it follows that $|D_{e'}| \ge 3$. 
    Therefore, there are at least $M$ edges that are at least $3$-decorated in $\calG(i,a)$, which contradicts with \prettyref{eq:purebulb}. 

    Without loss of generality, assume $\calB$ is decorated by $\{\SS_1, \TT_1\}$. Then, it follows that $(S_1)_a = (T_1)_a$. If $(i,a) \in S_2$, given $\calB$ is only decorated by $\{\SS_1, \TT_1\}$, we have $(S_2)_a \neq (S_1)_a$ and $(S_2)_a \neq (T_1)_a$. 
    By an analogous argument as the above paragraph, we must have $(i,a) \in T_2$ and $(S_2)_a = (T_2)_a$. Hence, all edges in $\calG(i,a)$ is either $2$-decorated or $4$-decorated. 
    Let $ \calG'(i,a) $ denote the subgraph of $\calG(i,a)$ induced by all edges $e$ with $|D_e| = 4$. We claim that $\calG'(i,a)$ must be an parent wire with length less than $M$. For any $4$-decorated edge $e$ in $\calG(i,a)$, let $\tilde{P}$ denote the path connecting $i$ to $e$.  Then $\tilde{P}$ must be a unique path in $S_1 \cup T_1\cup S_2 \cup T_2$ that connects $i$ to $e$. Given $e$ belongs to $S_1$, the path $\tilde{P}$ must belong to $S_1$.  Analogous argument holds for $T_1,S_2,T_2$. Hence $P$ must be $4$-decorated. In view of \prettyref{eq:purebulb}, $ \calG'(i,a) $ must be a parent wire with length less than $M$. It follows that $\calG(i,a)$ is a tangled chandelier with all edges on parent wire is $4$-decorated, and each edge outside parent wire is $2$-decorated. 
    
  If $(i,a) \not \in S_2$, then $ (i,a) \not \in T_2$; otherwise, by following the argument as above, we arrive at contradiction $(i,a) \in S_2$. It follows that $\calG(i,a)$ is a $2$-decorated $1$-chandelier with decoration $\{\SS_1,\TT_1\}$.

    \item  
    This follows directly from the facts that each non-isomorphic bulb appears at most once in $S_1$ (resp.~$T_1$, $S_2$, $T_2$) and each bulb is $2$-decorated. 
    
    \item  
    Suppose $S_1, S_2 \cong H$, $T_1,T_2 \cong I$ for some $H,I \in \calT$. Note that $H$ has $L$ non-isomorphic bulbs and $I$ also has $L$ non-isomorphic bulbs. Recall that there are  $t'$ non-isomorphic bulbs in $\UTone$. Thus there are least $(t'-L) \vee 0$ bulbs belonging to $\calJ(H)\backslash \calJ(I)$, and  $(t'-L) \vee 0$ bulbs belonging to $\calJ(I)\backslash \calJ(H)$ in $\UTone$. If a bulb belongs to $\calJ(H)\backslash \calJ(I)$, then it must be decorated by $\SS_1$ and $\SS_2$ and only appear once in $\UTone$; similarly for a bulb belonging to $\calJ(I)\backslash \calJ(H)$. 


\end{enumerate}
\end{proof}

\subsubsection{Proof of \prettyref{lmm:sum_f_t_tilde}}

Before proving \prettyref{eq:sum_f_t_tilde}, we first present the following claim. 
\begin{claim}\label{claim:f_a_b}
Suppose $|\calJ| \ge 2L$. 
For any $a,b\in\naturals$ such that 
$\frac{b}{2} \le a\le b-1$ and $b \le 2L$, we have 
\begin{align}
     \frac{ f_{a,b}}{ f_{a+1,b} } \le  \frac{6L^2}{|\calJ|-2L}\, ,  \label{eq:f_a_b_mono} 
\end{align}
and for any $a \in \naturals$, 
\begin{align}
    \frac{f_{a,a}}{f_{a,a+1}} \le 1 \, . \label{eq:f_a_a}
\end{align}
\end{claim}

Note that for any $v_1,e_1,\ell \in \naturals$ with $  K+M$  divides $v_1+e_1 $ and $v_1+e_1\le 2N$, we have $0 \le  \tilde{t} \le 2L-\indc{e_1=0}$, and $\tilde{t} \le  t \le (2L)\wedge 2\tilde{t}$, in view of \prettyref{eq:tilde_t}. 
Then, we have 
\begin{align}
    & \sum_{v_1\ge 0}  \sum_{\ell\ge 0} |\rho|^{\ell}   f_{\tilde{t},t} \, \, \indc{K+M \mid v_1+e_1   }  \indc{ v_1+e_1 \le 2N }  \nonumber \\
    & \le ~  \sum_{v_1\ge 0}  \sum_{\ell\ge 0}  \sum_{a = 0}^{ 2L-\indc{e_1=0}} \sum_{b=a}^{(2L)\wedge (2a)}   |\rho|^{\ell}  f_{a,b}   \indc{b = \frac{v_1+e_1}{K+M}}  \indc{a \le \floor{\frac{\ell+2e_1}{2(K+M)}}{} + L}  \nonumber\\
    & \le ~   \sum_{a = 0}^{ 2L-\indc{e_1=0}} \sum_{b=a}^{(2L)\wedge (2a)}  f_{a,b} \,\, \sum_{\ell \ge 2(a-L)(K+M)-2e_1} |\rho|^{\ell} \nonumber \\
    & \le ~ 4 \sum_{a = 0}^{ 2L-\indc{e_1=0}} \underbrace{\rho^{2(a-L)(K+M)-2e_1} \sum_{b=a}^{(2L)\wedge (2a)}   f_{a,b}}_{\triangleq h_a} \, , \label{eq:e_1_sum_true_1}
    \end{align}
where in the second-to-last inequality we interchanged the order of the summation 
and used the fact that if $a  \le \floor{\frac{\ell+2e_1}{2(K+M)}} + L$, then
$
\ell \ge 2\left(a - L \right)(K+M) - 2e_1;
$
the last inequality holds by our assumption that $|\rho| \le \frac{3}{4}$.

Note that for $0 \le a \le 2L-1$, 
\begin{align*}
    \frac{h_{a}}{ h_{a+1}}
    & \le \frac{2\sum_{b=a+1}^{(2L)\wedge (2a)} f_{a,b}}{\rho^{2(K+M)}\sum_{b=a+1}^{(2L)\wedge (2a+2)} f_{a+1,b}
    \, \, } \le \frac{12 L^2}{\rho^{2(K+M)}(|\calJ|-2L)
    \, \, } \le \frac{1}{2} \,,
\end{align*}
    where the first inequality holds because $f_{a,a} \leq f_{a,a+1}$ by \prettyref{eq:f_a_a};
the second inequality holds due to
  \prettyref{eq:f_a_b_mono} and $|\calJ| \ge 28L^2$, in view of  
$\frac{12 L^2}{\rho^{2(K+M)}(|\calJ|-2L)} \le \frac{1}{2} $; and the third inequality holds by our assumption that $\frac{12 L^2}{\rho^{2(K+M)}(|\calJ|-2L)} \le \frac{1}{2}$ in \prettyref{eq:true_pair_constraint}.

It follows from~\prettyref{eq:e_1_sum_true_1} that 
\begin{align*}
 & \sum_{v_1\ge 0}  \sum_{\ell\ge 0} |\rho|^{\ell}   f_{\tilde{t},t} \, \, \indc{K+M \mid v_1+e_1   }  \indc{ v_1+e_1 \le 2N } \nonumber \\
 & \le ~ 8 h_{2L-\indc{e_1=0}} \nonumber \\
& = ~ 8 \rho^{2N- 2(K+M) \indc{e_1=0} -2e_1} \sum_{b=2L-\indc{e_1=0}}^{2L}   f_{2L-\indc{e_1=0},b} \, .
\end{align*}

In the case of $e_1>0$, note that 
\begin{align*}
    f_{2L,2L} = ~ \binom{|\calJ|}{2L}\binom{2L}{L,L} \le \binom{|\calJ|}{L}^2\, .
\end{align*}
In the case of $e_1=0$, since $ \frac{f_{2L-1,2L-1}}{f_{2L-1,2L}} \le 1$ by \prettyref{eq:f_a_a} in \prettyref{claim:f_a_b}, it follows
\begin{align*}
     \rho^{-2(K+M)} (f_{2L-1,2L-1}+f_{2L-1,2L})  
     & \le ~ 2 \rho^{-2(K+M)}f_{2L-1,2L} \nonumber  \\
     & = ~ 12 \rho^{-2(K+M)} \binom{|\calJ|}{2L-1}\binom{2L-1}{1}\binom{2L-2}{L-1,L-1,0} \\
      & \le ~  \frac{12 L^2 }{\rho^{2(K+M)}|\calJ|} \binom{|\calJ|}{L}^2\, .
\end{align*}
Combining the last three displayed equations yields that 
\begin{align}
 &  \sum_{v_1\ge 0}  \sum_{\ell\ge 0} |\rho|^{\ell}   f_{\tilde{t},t} \, \, \indc{K+M \mid v_1+e_1   }  \indc{ v_1+e_1 \le 2N }  \nonumber \\
  & \le ~ 8 \rho^{2N}  \left(\rho^{-2e_1} \indc{e_1 \neq 0}+ \frac{12L^2}{\rho^{2(K+M)}|\calJ|}\indc{e_1 = 0} \right) \binom{|\calJ|}{L}^2\, ,\label{eq:e_1_sum_true}
\end{align}
which proves the desired result since $|\calT| = \binom{|\calJ|}{L}$.

\begin{proof}[Proof of \prettyref{claim:f_a_b}]
We have
\begin{align*}
    \frac{f_{a,a}}{f_{a,a+1}}
    & = \frac{ \binom{a}{(a-L) \vee 0,(a-L) \vee 0, (2L-a) \wedge a} }{  a \binom{a-1}{(a-L) \vee 0,(a-L) \vee 0, (2L-a-1) \wedge (a-1)}  }.
\end{align*}
If $|\calJ|\ge 2L$, for any $a\in\naturals$ such that $a \le (b-1)\wedge (L-1)$, we have
\begin{align*}
    \frac{f_{a,b}}{f_{a+1,b}}
    & = ~ \frac{\binom{|\calJ|}{a} \binom{a}{b-a} \binom{2a-b}{2a-b}  6^{ a}}{\binom{|\calJ|}{a+1} \binom{a+1}{b-a-1} \binom{2a+2-b}{2a+2-b}  6^{ a+1}}  = ~ \frac{\left(2a-b+2\right)\left(2a-b+1\right)}{6\left(b-a\right)(|\calJ|-a)}  \le \frac{L^2}{6(|\calJ|-L)} \, ,
\end{align*}
and for any $a\in\naturals$ such that $L \le a \le b-1 \le 2L-1$,
\begin{align*}
    \frac{f_{a,b}}{f_{a+1,b}}
    & = ~ \frac{\binom{|\calJ|}{a} \binom{a}{b-a} \binom{2a-b}{a-L, a-L, 2L-b }  6^{ 2L-a }}{\binom{|\calJ|}{a+1} \binom{a+1}{b-a-1} \binom{2a+2-b}{a+1-L,a+1-L, 2L-b}  6^{2L-a-1}} = ~ \frac{6\left(a-L+1\right)^2}{\left(b-a\right)(|\calJ|-a)} \le \frac{6L^2}{|\calJ|-2L+1} \,. 
\end{align*}
Hence, we have that for any $a\le b-1$, we have
\begin{align*}
    \frac{ f_{a,b} }{ f_{a+1,b}  }
      \le ~ \frac{6L^2}{|\calJ|-2L}\,.
\end{align*}
Also, we have that $\frac{ f_{0,0} }{ f_{0,1}  } = 1$, and for any $1 \le a  \le 2L-1$,
\begin{align*}
    \frac{f_{a,a}}{f_{a,a+1}}
= \frac{ \binom{a}{(a-L) \vee 0,(a-L) \vee 0, (2L-a) \wedge a} }{  a \binom{a-1}{(a-L) \vee 0,(a-L) \vee 0, (2L-a-1) \wedge (a-1)}  }=\frac{1}{2L-a} \le 1\,. 
\end{align*}

\end{proof} 

\subsection{Proof of \prettyref{prop:fake_pairs_improved} for fake pairs}
\label{sec:fake_pair}

In this subsection, we analyze the case when $i$ and $j$ are a fake pair, i.e., $i\neq j$. 
By \prettyref{eq:var_Phi_ij}, it suffices to show that under \prettyref{eq:fake_pairs_constraint},
\begin{align}
    \frac{\Gamma_{ij}}{\expect{\Phi_{ii}}^2} = O\left(\frac{1}{|\calT| \rho^{2N} }\right)  \,.  \label{eq:Xi_ij_mean_square}
\end{align}
Recall that $\calW_{ij}(v,k)$ is the set of $(S_1(i),T_1(i),S_2(j),T_2(j))$ such that $S_1\cup T_1 \cup S_2 \cup T_2$ has $v+2$ vertices and excess $k$ with constraints \prettyref{eq:constraint1} and \prettyref{eq:constraint2} satisfied. Under the constraint \prettyref{eq:constraint1}, we must have $e(S_1\cup T_1 \cup S_2 \cup T_2) = v+k+2 \le 2N$. 
Since $S_1(i)$ and $S_2(i)$ are rooted at $i$,  and $T_1(j)$ and $T_2(j)$ are rooted at $j$, $S_1\cup T_1 \cup S_2 \cup T_2 $ must contain at most two components so that $k\ge -2$.
By \prettyref{eq:fake_pairs_constraint}, we have $N^2=o(n)$.
Thus combining \prettyref{prop:mean_Phi_ij} and \prettyref{eq:var_bound_P} yields that
\begin{align}
    \frac{\Gamma_{ij}}{\expect{\Phi_{ii}}^2} 
     & \le ~ (1+o(1)) n^{-2N} |\calT|^{-2} \rho^{-2N}  \sum_{k \ge -2} \sum_{v=0}^{2N-k-2}  q^{-2N+v+k+2}  P_{ij}(v,k) \, \label{eq:fake_pair}.
\end{align}

We first proceed with the simple case of $k=-2$, where $S_1\cup T_1\cup S_2\cup T_2$ is a disconnected graph consisting of two connected components.
By \prettyref{eq:constraint1}, we must have $S_1=T_1$ and $S_2 =T_2$, and  $v = 2N $. 
In this case, recalling the definition from \prettyref{tab:K_ell_m}, we have $E(K_{11})=\emptyset$.
Then, it follows from the definition of $P_{ij}$ in \prettyref{eq:P_ij} that
\begin{align}
    P_{ij}(2N,-2) 
    &  = ~ \sum_{(S_1(i),T_1(i),S_2(j),T_2(j))\in\calW(2N,-2)}  \left(\aut(S_1)\aut(T_1)\aut(S_2)\aut(T_2)\right)^{\frac{1}{2}} 
    \notag \\
    & \overset{(a)}{\le} ~  \sum_{H \in \calT} \aut(H)^2 |\{S_1(i): S_1 \cong H\}|  |\{S_2(j): S_2 \cong H\}|  \notag \\
    & \overset{(b)}{=}~
     |\calT| \, n^{2N}  \,, \label{eq:case-k=-2}
\end{align}
where $(a)$ holds because $\calW_{ij}(v,k)$ is a subset of $\calW_{ij}$ and then if $S_1,S_2\cong H$ and $T_1,T_2\cong I$, given $S_1=T_1$ and $S_2=T_2$, we must have $H=I$ and $ \left(\aut(S_1)\aut(T_1)\aut(S_2)\aut(T_2)\right)^{\frac{1}{2}} = \aut(H)^2$;  $(b)$ holds because for any rooted tree $H\in \calT$, 
$|\{S (i): S \cong H\}| = |\{S_2(j): S_2 \cong H\}|= \sub_n(H) \le \frac{n^N}{\aut(H)}$ following from \prettyref{eq:a_H}.


The case of $k\ge -1$ is much more sophisticated, for which we derive the following upper bound on $P_{ij}(v,k)$. 
\begin{lemma}\label{lmm:W_v_k}
       For any $v, k$ where $k\ge -1$ and $v\le 2N-k-2$,  we have
       \begin{align}
     P_{ij}(v,k)  
     \le &  n^{v}  
     |\calT| \, 4^{L}  L^{ 2L \wedge (4K+2)} ( 6\Universal)^{4K+4M}    \nonumber \\
     & \left( 11 (2N+1)^3 R^4 (11\Universal)^{4(K+M)} \right)^{k+2} 
     \left(R^{\frac{2}{M} } (11 \Universal)^{\frac{4(K+M)}{M}}  \right)^{2N-v-k-2}
          \,,
\label{eq:W_v_k}
       \end{align}
where we recall that $\R$ is an upper bound on the number of automorphisms of each bulb of any chandelier in $\calT$.
 \end{lemma}

\begin{remark}
To appreciate the significance of~\prettyref{eq:W_v_k}, it is instructive to consider the special case $k=-1$, in which $S_1\cup T_1 \cup S_2 \cup T_2 $ is a tree. In this case, we can get a straightforward bound $|P_{ij}(v,-1)|\le \R^{2L} n^v \Universal^{v+1}11^{v+1}$, 
which follows from the definition of $P_{ij}$ in \prettyref{eq:P_ij} together with the following two facts: $(a)$ $\left(\aut(S_1)\aut(T_1)\aut(S_2)\aut(T_2)\right)^{\frac{1}{2}} \le \R^{2L}$ in view of $\aut(H) \le \R^L$ for any $L$-chandelier $H$; and $(b)$ $|\calW_{ij}(v,-1)| \le n^v\Universal^{v+1}11^{v+1}$,
where $\Universal^{v+1}$ counts all possible unlabeled trees with $v+1$ edges and $11^{v+1}$ enumerates all possible assignment of edges in $E(S_1\cup T_1 \cup S_2 \cup T_2) $
to $E(S_1), E(T_1), E(S_2),E(T_2)$. 
However, this simple bound is too loose when plugged in the RHS of \prettyref{eq:fake_pair}, for the following two reasons: (i) Since $R$ will be chosen to $\exp(\Theta(K))$, the factor of $R^L$ can be too large; (ii) The $(11\Universal)^v$ factor can be much larger than $n^2$ when $v=2N-1$ and 
$N=C \log n$ for some big constant $C$.  To circumvent these catastrophic factors, in the proof of \prettyref{lmm:W_v_k} we  examine the anatomy of $S_1\cup T_1 \cup S_2 \cup T_2 $
and crucially leverage the chandelier structure.
\end{remark}
Applying \prettyref{lmm:W_v_k} and \prettyref{eq:case-k=-2} to \prettyref{eq:fake_pair},
we have
\begin{align*}
    \frac{\Gamma_{ij}}{\expect{\Phi_{ii}}^2} 
      & \le 
     ~ (1+o(1))|\calT|^{-1} \rho^{-2N} \Biggl\{ 1 + 
        4^{L}  L^{ 2L \wedge (4K+2)} ( 6\Universal) ^{4K+4M}\\
     &~~~~~ \sum_{k \ge -1} \left(\frac{11 (11\Universal)^{4(K+M)} \R^4 (2N+1)^3 }{n}\right)^{k+2} \\
     &~~~~~ \sum_{v=0}^{2N-k-2} \left(\frac{R^{\frac{2}{M} } (11 \Universal)^{\frac{4(K+M)}{M}} }{nq} \right)^{2N -(v+k+2)}  \Biggr\}\\
     & \overset{(a)}{\le }    (1+o(1))|\calT|^{-1}\rho^{-2N} \Biggl\{ 1 +  4^{L+1}  L^{ 2L \wedge (4K+2)} ( 6\Universal) ^{4K+4M}  \\
     & ~~~~~ 
     \left( \frac{11 (11\Universal)^{4(K+M)} \R^2 (2N+1)^3 }{n} \right)\Biggr\}\\
     & \overset{(b)}{=}O\left( |\calT|^{-1} \rho^{-2N} \right) \,,
\end{align*}
where $(a)$ holds because by the condition 
$\frac{R^{\frac{2}{M} } (11 \Universal)^{\frac{4(K+M)}{M}} }{nq} \le \frac{1}{2}$ in \prettyref{eq:fake_pairs_constraint},
\begin{align*}
     \sum_{v=0}^{2N-k-2} \left(\frac{R^{\frac{2}{M} } (11 \Universal)^{\frac{4(K+M)}{M}} }{nq} \right)^{2N -v-k-2} \le 2 \,, 
\end{align*}
and moreover, by the second condition in \prettyref{eq:fake_pairs_constraint}, we have $\frac{11 (11\Universal)^{4(K+M)} \R^2 (2N+1)^3 }{n} \le \frac{1}{2}$ so that
\begin{align*}
    \sum_{k \ge -1} \left(\frac{11 (11\Universal)^{4(K+M)} \R^2 (2N+1)^3 }{n}\right)^{k+2} \le 2\left( \frac{11 (11\Universal)^{4(K+M)} \R^2 (2N+1)^3 }{n} \right)\, ;
\end{align*}
$(b)$ holds by the second condition in \prettyref{eq:fake_pairs_constraint}. 





\subsubsection{Proof of \prettyref{lmm:W_v_k}}
Recall that there is a one-to-one correspondence between the $4$-tuple $(S_1(i), T_1(i), S_2(j), T_2(j))$ and the decorated union graph $U=S_1\cup T_1 \cup S_2 \cup T_2$. To bound $P_{ij}(v,k)$, it suffices to enumerate $U$
and bound the numbers of automorphisms. As such, we apply the idea of ``divide-and-conquer''.
In particular, we  decompose $U$ into multiple parts and bound them separately.

\begin{figure}[ht]
    \centering
    \includegraphics[width=0.9\textwidth]{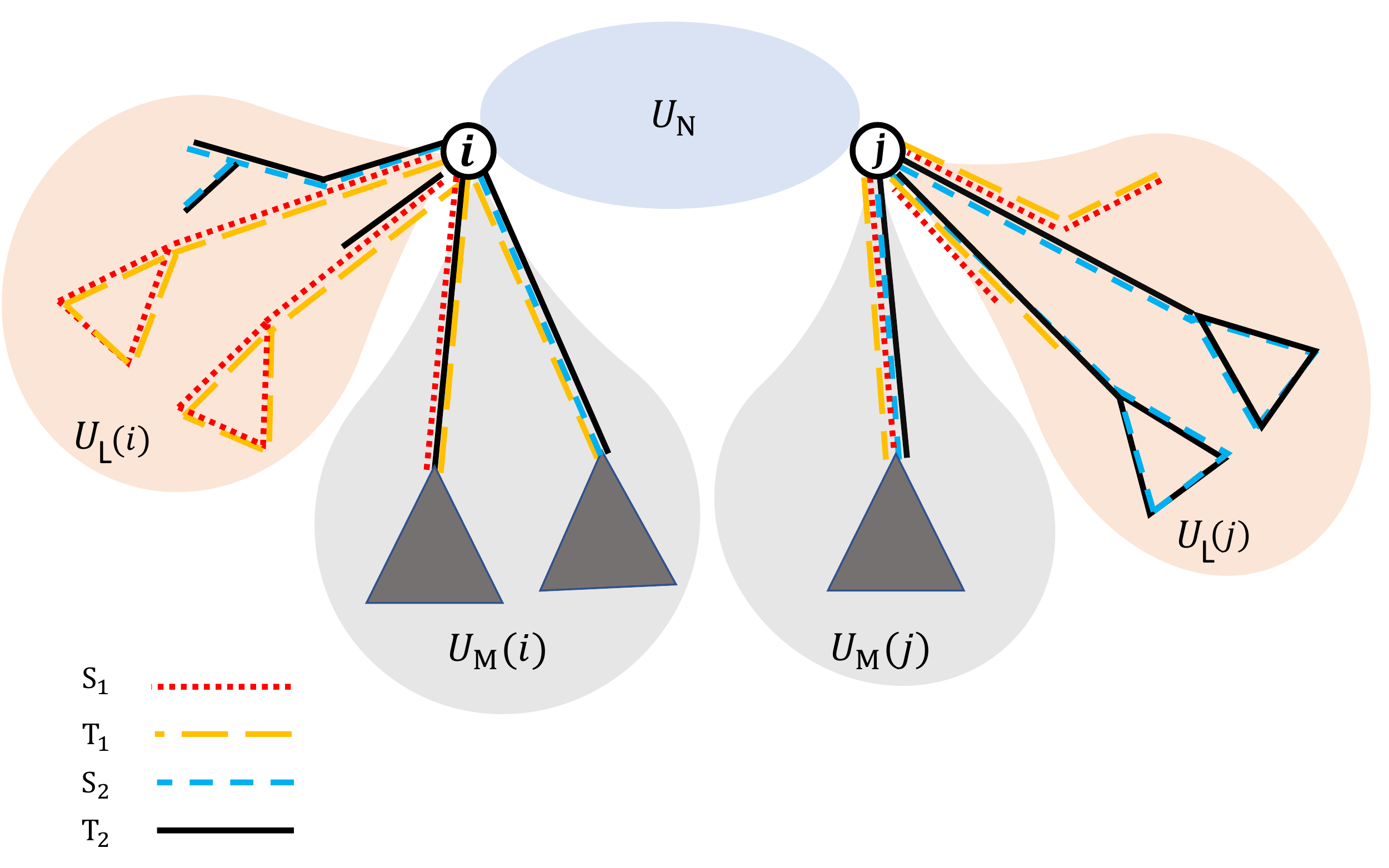}
\centering
    \caption{
 An illustration of the decomposition of $U$ as $\UTone(i) \cup \UTtwo(i) \cup \UTone(j) \cup \UTtwo(j) \cup \UN$ in the same convention as \prettyref{fig:decomp_true}.
    } 
    \label{fig:fake_pair_decomp}
\end{figure}

We first decompose $U$ into three edge-disjoint parts $\UT(i)$, $\UT(j)$, and $\UN$,
where $\UT(i)$ (resp.\ $\UT(j)$) is a tree rooted at $i$ (resp.\ $j$), and $\UN$ is a connected graph containing both $i$ and $j$ with the 
same access $k$, as depicted in~\prettyref{fig:fake_pair_decomp}.

Specifically, for any neighbor $a$ of $i$ in $U$,
consider the graph $U$ with the edge $(i,a)$ removed and 
let $\calC(i,a)$ be the connected component therein that contains $a$.
Let $\calG(i,a) $ denote the graph union of $ \calC(i,a)$ and the edge $(i,a)$. Let 
\begin{align}
    \calNT(i) 
    & = ~ \big\{ a : D_{(i,a)}  \cap \{\SS_1, \TT_1\} \neq \emptyset\, , \,\; \calG(i,a) \text{ is a tree  not containing $j$} \big\}  \label{eq:calN_T_i_fake}\\
    \calNN(i) 
    & = ~ \big\{ a : D_{(i,a)}  \cap \{\SS_1, \TT_1\} \neq \emptyset  \big\} \backslash \calNT(i) \label{eq:calN_N_i_fake}. 
\end{align}
Similarly, define
\begin{align}
    \calNT(j) 
    & =~ \big\{ a :  D_{(j,a)} \cap \{ \SS_2, \TT_2\} \neq \emptyset\,, \;\, \calG(j,a) \text{ is a tree not containing $i$} \big\} \, ,\label{eq:calN_T_j_fake} \\
    \calNN(j) 
    & =~ \big\{ a :  D_{(j,a)} \cap \{ \SS_2, \TT_2\} \neq \emptyset \big\} \backslash  \calNT(j)\, .\label{eq:calN_N_j_fake}
\end{align} 

Then we decompose $U$ according to $\calNT(i)$ and $\calNT(j)$:
\begin{align}
    \UT(i) \triangleq ~\bigcup_{a \in \calNT(i)} \calG(i,a) , \quad \UT(j) \triangleq ~\bigcup_{a \in \calNT(j)} \calG(j,a) , \quad
    \UN \triangleq ~ U \backslash (\UT(i) \cup \UT(j)) \, . \label{eq:U_decompose_fake}
\end{align} 
If $\calNT(i)$ (resp.\ $\calNT(j)$) is empty, 
we define $\UT(i)$ (resp.\ $\UT(j)$) as the graph consisting of the single vertex $i$ (resp.\ $j$) by default. 

Next, we further decompose $\UT(i)$ into two edge-disjoint subtrees and similarly for $\UT(j)$. See~\prettyref{fig:fake_pair_decomp} for an illustration. In particular, 
define
\begin{align}
\calNTtwo(i) & = \left\{ a \in \calNT(i): \left|  \{e \in E(\calG(i,a))  : |D_e| \ge 3\} \right| 
\ge M \right\},  \label{eq:mixbulb} \\
\calNTone(i)
& = \calNT(i) \backslash \calNTtwo(i).
     \label{eq:purebulb}
\end{align}
This leads to the partition $\UT(i) = \UTone(i) \cup \UTtwo(i) $, where
\begin{align}
    \UTone(i) \triangleq \bigcup_{a \in \calNTone(i)} \calG(i,a) \, , \quad 
    \UTtwo(i) \triangleq \bigcup_{a \in \calNTtwo(i)} \calG(i,a) \, . \label{eq:U_i_decompose_fake}
\end{align}
When $\calNTone(i)$ (resp.\ $\calNTtwo(i)$) is empty, we define $\UTone(i)$ (resp.\ $\UTtwo(i)$)
to be the graph with a single vertex $i$, respectively.  
Similarly, we define $\calNTone(j)$, $\calNTtwo(j)$, $\UTone(j)$, and $\UTtwo(j)$. 
We then define
\begin{align}
    \UTone = \UTone(i) \cup \UTone(j) \,, \quad \UTtwo = \UTtwo(i) \cup \UTtwo(j) \,. \label{eq:U_T_1_U_T_2}
\end{align}

Overall we have 
$U = \UN \cup \UTone\cup \UTtwo$, which will be shown to be an edge-disjoint union.
The intuition behind this decomposition is as follows. Analogous to the case of true pairs in \prettyref{sec:pf-W_v_k_ell}, roughly speaking, $\UN$ captures the overlapping part that contains cycles, and $\UTtwo$ deals with the case where each branch contains
at least $M$ edges that are $3$ or $4$-decorated.
Again, we can control the sizes of  both $\UN$ and $\UTtwo$ and apply a crude counting argument. In comparison, $\UTone$ is more delicate. Following a similar reasoning as in the analysis of true pairs, each bulb in $\UTone$ is exactly $2$-decorated. Moreover, here since $S_1, T_1$
are rooted at $i$ and $S_2, T_2$ are rooted at $j \neq i$, we can further argue that the bulbs in $\UTone(i)$ must be decorated by $\{\SS_1, \TT_1\}$, 
while the bulbs in $\UTone(j)$ must be decorated by $\{\SS_2, \TT_2\}$. In other words, within $\UTone$, only the wires of $S_1\cup T_1$
and $S_2 \cup T_2$ can overlap and their bulbs cannot; see~\prettyref{fig:decomp_true} for an illustration. 
This observation is crucial for bounding the contribution of $\UTone$ in \prettyref{lmm:fake_enumeration}.

Based on the decomposition of $U$, we further assign 
appropriate weights to $\UTone$, $\UTtwo$ and $\UN$ to represent their contributions to the number of  automorphisms:
 \begin{align}
w(\UTone) & \triangleq ~ w_{S_1}(\UTone(i)) w_{T_1}(\UTone(i)) w_{S_2}(\UTone(j)) w_{T_2}(\UTone(j)) \,, \label{eq:w_UL} \\
w(\UTtwo) & \triangleq ~ w_{S_1}(\UTtwo(i)) w_{T_1}(\UTtwo(i)) w_{S_2}(\UTtwo(j)) w_{T_2}(\UTtwo(j)) \,,
\label{eq:w_U'_fake} \\
   w(\UN) & \triangleq ~w_{S_1}(U\backslash \UT(i)) w_{T_1}(U\backslash \UT(i)) w_{S_2}(U\backslash \UT(j)) w_{T_2}(U\backslash \UT(j))\,. \label{eq:w_UN}
\end{align}
where we applied the notation \prettyref{eq:w_S_U'}, namely, for any chandelier $S$
and any graph $G \subset U$, $
w_S(G) \triangleq \prod_{\B \in \calK(S), \, \B\subset G} \aut(\B)^{\frac{1}{2}}$.

It is easy to verify that $\UT(i)$, $\UT(j)$, and $\UN$ satisfy the following properties (proved at the end of this subsection). 
\begin{claim}
\label{claim:U_decompose_fake}
Assume that $k\geq -1$. 
\begin{enumerate}[label=(\roman*)]
\item \label{N:0} For all $a \in \calNT(i)$, $\calC(i,a) = (S_1)_a \cup (T_1)_a \cup (S_2)_a \cup (T_2)_a$, and
the number of vertices in $\calC(i,a)$ is between $M+K$ and $2M+2K-1$; the same holds for $\calC(j,a)$ for all $a \in \calNT(j)$.
  \item \label{N:1}
  $\UT(i) $ and $\UT(j)$ are two vertex-disjoint trees, where $\UTone(i)$ and $\UTtwo(i)$ (resp.~$\UTone(j)$ and $\UTtwo(j)$) only share one common vertex $i$ (resp.~$j$).
  
   \item  \label{N:3}
   Each of the graphs $\UTone \cap \UTtwo$, $\UTone \cap \UN$, and $\UTtwo \cap \UN$ consists of only two isolated vertices $i$ and $j$. 
  \item \label{N:4}
  $\UN$ is a connected graph with the same excess as $U$, and the number of edges in $\UN$ is bounded by $ \left( |\calNN(i)|+ |\calNN(j)|\right)(M+K)$.
 
  \item \label{N:5}
    $ w(\UTone)w(\UTtwo)w(\UN) =(\aut(S_1)\aut(T_1)\aut(S_2)\aut(T_2))^{\frac{1}{2}}$, $w(\UTtwo) \le \R^{|\calNTtwo(i)|+|\calNTtwo(j)|}$, and $w(\UN) \le \R^{|\calNN(i)|+ |\calNN(j)|}$. 
     \end{enumerate}
\end{claim}

The following lemma shows that $|\calNN(i)| + |\calNN(j)| $ is small when the excess $k$ is small. This will be useful for bounding the size of  $\UN$ and $w(\UN)$, and facilitating the enumeration of $\UN$. 
\begin{lemma}\label{lmm:k_fake}
        For any $-1\le k\le N-1$, we have
        \begin{align*}
                |\calNN(i)|+|\calNN(j)| \le 4(k+2) \,.
        \end{align*}
        
\end{lemma}


Then by \prettyref{claim:U_decompose_fake}\ref{N:3}, we have
\begin{align}
|V(\UTone) \backslash\{i,j\}|+|V(\UTtwo)\backslash\{i,j\}|+|V(\UN)\backslash\{i,j\}| = v. \label{eq:v_sum_bound_fake}
\end{align}
By \prettyref{eq:mixbulb}, for any $a\in \calNTtwo(i)$, we have $ e(\calG(i,a)\cap (K_{21}\cup K_{22})) + e(\calG(i,a)\cap (K_{12}\cup K_{22})) \ge M$. 
An analogous statement holds for $a\in \calNTtwo(j)$. By \prettyref{eq:S1_T1_S2_T2},  $ \frac{1}{2} \left( e (K_{22}\cup K_{21}) +  e (K_{22}\cup K_{12})\right) =  2N - (v+k+1)$. Then, we have
\begin{align}
     |\calNTtwo(i)|+|\calNTtwo(j)| \le  \frac{2\left(2N-(v+2+k)\right) }{M} \triangleq \bTtwo . \label{eq:calNM_i_j}
\end{align}
It follows by~\prettyref{claim:U_decompose_fake}\ref{N:0}, \ref{N:5} and \prettyref{eq:calNM_i_j} that
\begin{align}
|V(\UTtwo)\backslash\{i,j\}| & \le 2(M+K) \left(|\calNTtwo(i)|+|\calNTtwo(j)| \right)
\le 2(M+K) \bTtwo  \label{eq:v_2_fake} \,, \\
w(\UTtwo) & \le 
\R^{|\calNTtwo(i)|+|\calNTtwo(j)|} =  \R^{\bTtwo}.  \label{eq:w_U_T_2_fake}
\end{align}
By~\prettyref{claim:U_decompose_fake}\ref{N:4}, \ref{N:5} and \prettyref{lmm:k_fake}, in view of $k\ge -1$, we have
\begin{align}
|V(\UN)\backslash\{i,j\}| & \le \left( |\calNN(i)|+ |\calNN(j)|\right)(M+K) -k-2 \le 4(k+2)(M+K) -1 \label{eq:v_N_fake} \,, \\
w(\UN) & \le \R^{|\calNN(i)|+|\calNN(j)|}\le \R^{4(k+2)} \,.  \label{eq:w_UN_fake}
\end{align}




Let $\calU(v,k)$ denote the set of all possible $U$ with $v+2$ vertices and excess $k$. Let $\calUTone(\ell)$ and $\calUTtwo(\ell)$ denote the set of 
all possible $\UTone$ and $\UTtwo$ with $\ell+2$ vertices, respectively. Let $\calUN(\ell,k)$ denote the set of 
all possible $\UN$ with $\ell+2$ vertices and excess $k$.

By \prettyref{eq:P_ij} and $|\rho|\le 1$, for any $v\le 2N-k-2$ and $k\ge -1$, 
we have 
 \begin{align}
     P_{ij}(v,k) 
     & \le  ~ \sum_{(S_1,T_1,S_2,T_2)\in \calW_{ij}(v,k)}   (\aut(S_1)\aut(T_1)\aut(S_2)\aut(T_2))^{\frac{1}{2}} \nonumber\\
     & \overset{(a)}{=} ~ \sum_{U\in \calU(v,k)} w(\UTone)w(\UTtwo)w(\UN) \nonumber\\
     &  \overset{(b)}{\le}  ~
      R^{4(k+2)+\bTtwo }
    \sum_{\ell=0}^{2(M+K)\bTtwo} \sum_{m=0}^{4(k+2)(M+K) -1}
     \PTone (v-\ell-m)  \left| \calUTtwo(\ell) \right|
     \left| \calUN(m, k) \right|,
     \label{eq:U_decompose_fake_bound}
 \end{align}
where 
\begin{align}
    \PTone(\ell)
    & = ~ \sum_{\UTone \in \calUTone(\ell)}  w(\UTone) \label{eq:P_T_1_fake};
\end{align}
$(a)$ holds by \prettyref{claim:U_decompose_fake}\ref{N:5}; 
 $(b)$ holds due to~\prettyref{eq:v_sum_bound_fake}--\prettyref{eq:w_UN_fake}.

First, we bound $|\calUTtwo(\ell)|$.
Recall that $\UTtwo=\UTtwo(i)\cup \UTtwo(j)$,
where 
$\UTtwo(i)$ and $\UTtwo(j)$ are two vertex-disjoint trees rooted at $i$ and $j$, respectively. 
Thus there are at most $\ell \Universal^{\ell}$  unlabeled, non-decorated $\UTtwo$ with $\ell$ edges.
Recall that each edge can be decorated in at most $11$ ways. Hence, 
\begin{align}
    |\calUTtwo(\ell) | \le ~ \binom{n}{\ell}\ell! \ell \Universal^{\ell} 11^{\ell} \le ~   \ell  (11\Universal n)^{\ell}   \,.  \label{eq:calU_T_2_fake}
\end{align}

Next, we bound $\calUN(\ell,k)$. By \prettyref{lmm:enum} and \prettyref{eq:Universal}, the total number of unlabeled non-decorated graphs $\UN$ with $\ell+2$ vertices and excess $k$ is bounded by 
\[
     \Universal^{\ell+1} {\binom{\binom{\ell+2}{2}}{k+1}}\le \Universal^{\ell+1} (\ell+2)^{2(k+1)}  \,.
\]
Recall that each edge can be decorated in at most $11$ ways, and $\UN$ contains vertices $i$ and $j$ and has excess $k$. Therefore,
\begin{align}
        |\calUN(\ell,k)|
        &  \le ~ \binom{n}{\ell} (\ell+2)!\Universal^{\ell+1} (\ell+2)^{2(k+1)}
        11^{\ell+2+k}  \nonumber \\
        & \le ~ n^{\ell}\Universal^{\ell+1} (\ell+2)^{2(k+2)}
        11^{\ell+2+k} \,.
         \label{eq:fake_enumeration_non_tree}
\end{align}

Finally, we provide the following lemma to bound $\PTone(\ell)$. In contrast to $\UTtwo$ and $\UN$, straightforward bounds on the enumeration and the weights will no longer suffice. Instead, the proof relies on showing that $\UTone$ is close to a
chandelier.
\begin{lemma}\label{lmm:fake_enumeration}
For $0 \le \ell \le v \le 2N-k-2$ and $ k\ge -1$,  we have
\begin{align}
    \PTone(\ell) 
    & \le~ n^{\ell} |\calT| 4^{L}  L^{ 2L \wedge (4K+2)} ( 6\Universal) ^{4K+4M-2}
    \label{eq:fake_enumeration_tree>0} \,.
\end{align}
\end{lemma}

Substituting  \prettyref{eq:calU_T_2_fake}, \prettyref{eq:fake_enumeration_non_tree}, and \prettyref{eq:fake_enumeration_tree>0}, 
into \prettyref{eq:U_decompose_fake_bound}
yields that 
for $v \le 2N-k-2$ and $k \ge -1$,
\begin{align*}
     P_{ij}(v,k)
     \le &   R^{4(k+2)+\bTtwo } n^{v}
     |\calT| 4^{L}  L^{ 2L \wedge (4K+2)} ( 6\Universal)^{4K+4M-2}  \left( 2N+1 \right)^{2k+5}  \\
   &  \sum_{\ell=0}^{2(M+K)\bTtwo}  (11\Universal )^{\ell}  \sum_{m=0}^{4(k+2)(M+K) -1}
     \Universal^{m+1}  11^{m+2+k} \\
     \le &   R^{4(k+2)+\bTtwo } n^{v}
     |\calT| 4^{L}  L^{ 2L \wedge (4K+2)} ( 6\Universal)^{4K+4M}    \\
     & \left( 11 (2N+1)^3 \right)^{k+2} \left( 11 \Universal\right)^{2(M+K)\bTtwo+4(k+2)(K+M)},
\end{align*}
where the first inequality holds because $\ell, m \le v \le 2N-1$; the last inequality holds due to $k \ge -1$
and 
\begin{align*}
\sum_{\ell=0}^{2(M+K)\bTtwo} \left( 11 \Universal \right)^\ell & \le  2\left( 11 \Universal \right)^{2(M+K)\bTtwo} ,\\
\sum_{m=0}^{4(k+2)(M+K) -1}
     \left(11 \Universal\right)^m  
     & \le \left(11 \Universal\right)^{4(k+2)(M+K)}.
\end{align*}

\begin{proof}[Proof of Claim~\ref{claim:U_decompose_fake}]

\begin{enumerate}[label=(\roman*)]

\item The proof is analogous to that of~\prettyref{claim:U_decompose_true}~\ref{T:1}, so we omit the details. 
There is  only a minor difference when showing 
any edge $e$ in $\calC(i,a)$
must belong to $ (S_1)_a \cup (S_2)_a \cup (T_1)_a \cup (T_2)_a$. 
If not, then 
there must exist $e \in E(C(i,a))$ but 
$e \not\in E((S_1)_a \cup (S_2)_a \cup (T_1)_a \cup (T_2)_a)$. Analogous to the proof  of~\prettyref{claim:U_decompose_true}~\ref{T:1}, we can show that if $e = (i,b)$ or $e $ belongs to $(S_1)_b \cup (S_2)_b \cup (T_1)_b \cup (T_2)_b$  for any $b\neq a$, $\calG(i,a)$ contains a cycle through $i$, contradicting the definition of $a \in \calNT$. Hence, $e\not \in E(S_1\cup T_1 \cup (S_2)_i \cup (T_2)_i)$, which further implies that 
$e \in E(S_2 \backslash (S_2)_i) \cup E(T_2 \backslash (T_2)_i)$.
It follows that $\calC(i,a)$ must contain $j$, contradicting the definition of $a \in \calNT$.

    \item Note that for distinct vertices $a, b \in \calNT(i)$, 
    $\calG(i,a)$ and  $\calG(i,b)$ only share a common vertex $i$, for otherwise $\calG(i,a)$ would contain a cycle through $i$. Therefore, $\UT(i)$ is a tree rooted at $i$, where $\UTone(i)$ and $\UTtwo(i)$ only share one common $i$. 
    The same conclusions hold for $G(j,a)$, $\UT(j)$, $\UTone(j)$ and $\UTtwo(j)$. 
    Moreover, $\UT(i)$ and $\UT(j)$ are vertex-disjoint; otherwise $\calG(i,a)$ would contain $j$ for some $a \in \calNT(i)$.

    \item 
    By~\ref{N:1},
    it follows that $\UTone$ and $\UTtwo$ only share two isolated common vertices $i$ and $j$. 
    Next, we show $\UN$ and $\UTone$ (resp.~$\UTtwo$) only share two isolated common vertices $i$ and $j$. We first show $\UN$ contains both $i$ and $j$. Since $U$ is connected, there exists a path $P$ from $i$ to $j$ in $U$. Such a path $P$ must be edge-disjoint from $\UT(i)$ and $\UT(j)$; otherwise, $\calG(i,a)$ 
    would contain $j$  for some $a \in \calNT(i)$ or $G(j,a)$ would contain $i$ for some $a \in \calNT(j)$. Therefore, the path $P$
    must be in $\UN$ and hence $\UN$ contains both $i$ and $j$. Then, we show that $\UN$ and $\UT(i)$ do not share any vertex other than $i$. Suppose not and let $c \neq i$ denote a common vertex of $\UN$ and $\UT(i)$. Then $c$ must belong to $\calC(i,a)$ for some $a \in \calNT(i)$. Also, since $\UN=U\backslash (\UT(i) \cup \UT(j))$, $c$ must be incident to an edge $e$ not in $\UT(i)$. However, by definition of $\calC(i,a)$, $e$ belongs to $\calC(i,a)$ and hence $\UT(i)$, which is a contradiction. 
    Analogously, we can show that $\UN$ and $\UT(j)$ do not share any vertex other than $j$. Since $\UTone=\UTone(i) \cup \UTone(j)$ and $\UTtwo=\UTtwo(i) \cup \UTtwo(j)$, our claim follows.  
  
  \item By~\ref{N:3} and the  connectivity of $U$, $\UN$ must be connected. Moreover, we have that 
   \begin{align*}
    v(U) & = v(\UT(i))+v(\UT(j))+v(\UN) - 2  \\
    e(U) & = e(\UT(i))+e(\UT(j))+e(\UN)\, .
   \end{align*}
   Combining the above with the facts that $\UT(i)$ and $\UT(j)$ are trees, we get that 
   $
   e(\UN)- v(\UN) =  e(U)-  v(U) =k.
   $

 
 Since $i$ has at least $L$ neighbors in $S_1\cup T_1$, by \prettyref{eq:calN_T_i_fake} and \prettyref{eq:calN_T_j_fake}, we have 
  $ |\calNT(i)| +  |\calNN(i)| \ge L$. Similarly, we can show that  $ |\calNT(j)| +  |\calNN(j)| \ge L$. 
  Then, we have
   \begin{align*}
       |\calNT(i)|+ |\calNT(i)| \ge 2 L-  |\calNN(i)|-  |\calNN(j)| \,.
   \end{align*}
   By~\ref{N:0}, we have
   \begin{align*}
       e(\UT(i))+e(\UT(j))\ge (  |\calNT(i)|+ |\calNT(i)| )(M+K) \,.
   \end{align*}
   By~\ref{N:1} and \ref{N:3}, we know that $\UN$, $\UT(i)$ and $\UT(j)$ are pairwise edge-disjoint. By \prettyref{eq:eU2N}, $e(U)\le 2N = 2L(M+K)$ and hence
   \begin{align*}
       e(\UN) \le (|\calNN(i)|+ |\calNN(j)|)(M+K) \,.
   \end{align*}

   \item 

    Fix a bulb $\B$ from $S_1$, $\B \subset (S_1)_a$ for some $a$ that is a neighbor of $i$ in $S_1$. Then, by \prettyref{eq:calN_T_i_fake} and \prettyref{eq:calN_N_i_fake}, $a\in \calNT(i) \cup \calNN(i)$. By~\ref{N:0}, for any $a \in \calNT(i)$, $(S_1)_a \subset \calC(i,a) \subset \calG(i,a)$. It follows from \prettyref{eq:U_i_decompose_fake} that $\B \subset \UTone(i)$ (resp.~$\UTtwo(i)$) if and only if $a \in \calNTone(i)$ (resp.~$ \calNTtwo(i)$). Moreover, $\B \subset U\backslash \UT(i)$, if and only if $a \in \calNN(i)$. Therefore,
   $\calB$ is contained in 
   exactly one of $\UTone(i)$, $\UTtwo(i)$,
   and $U\backslash \UT(i)$. 
   
    Recall from~\prettyref{eq:aut_H} that $\aut(S_1) = \prod_{\calB\in \calK(S_1)}\aut(\calB)$, given $S_1$ is a chandelier with $L$ non-isomorphic bulbs. Analogous arguments hold for $T_1,S_2,T_2$.
   Therefore, in view of~\prettyref{eq:w_S_U'} and \prettyref{eq:w_U'_fake}, it follows that
    \[
    (\aut(S_1)\aut(T_1)\aut(S_2)\aut(T_2))^{\frac{1}{2}} = w(\UTone)w(\UTtwo)w(\UN) \,.
    \]

   Note that each bulb $\calB \in \calK(S_1)$ must be contained in distinct $(S_1)_a$ for $a$ that is a neighbor of $i$ in $S_1$. There are at most $|\calNTtwo(i)|$ (resp.~$|\calNN(i)|$) bulbs from $S_1$ contained in $\UTtwo(i)$ (resp.~$U\backslash \UT(i)$). Then, given $\aut(\calB)\le \R$ for any $\calB \in \calK(S_1)$, by \prettyref{eq:w_S_U'}, we have $w_{S_1}(\UTtwo(i))\le \R^{\frac{1}{2}|\calNTtwo(i)|}$ (resp.~$w_{S_1}(U\backslash \UT(i))\le \R^{\frac{1}{2}|\calNN(i)|}$). Analogous arguments hold for $T_1,S_2,T_2$. Hence, it follows that 
   \[
   w(\UTtwo) \le \R^{|\calNTtwo(i)|+|\calNTtwo(j)|}\,, \quad w(\UN) \le \R^{|\calNN(i)|+ |\calNN(j)|} \,.
   \]

\end{enumerate}

\end{proof}

\subsubsection{Proof of \prettyref{lmm:k_fake}}\label{sec:k_fake}

By \prettyref{claim:U_decompose_fake}\ref{N:4}, the excesses of $U$ and $\UN$ are both $k$. It suffices to prove that 
\begin{align}
  e(\UN) - v(\UN)+ 2 \ge ~ \frac{1}{4} (|\calNN (i)|+|\calNN (j)|)\, . \label{eq:target_1}
\end{align}

Define 
\begin{align*}
    \calNNone(i) & = ~ \{a \in \calNN(i): \text{ $ \calC(i,a) $ does not contain $i$ or $j$}\} \, ,\\
    \calNNone(j) & = ~\{a \in \calNN(j): \text{ $ \calC(j,a) $ does not contain $i$ or $j$}\} \, .
\end{align*}
Let 
\[
\UNone(i) = ~ \bigcup_{a\in \calNNone(i) } \calG(i,a) \,, \quad 
\UNone(j) = ~ \bigcup_{a\in \calNNone(j) } \calG(j,a) \,. 
\]
Then we have 
\begin{align}
    e(\UNone(i))- v(\UNone(i))+1 & \ge ~ |\calNNone(i)|, \label{eq:E_U_1_i}\\ 
    e(\UNone(j)) - v(\UNone(j))|+1 & \ge ~  |\calNNone(j)|  \, .\label{eq:E_U_1_j}
\end{align} 
Indeed, for any $a \neq b \in  \calNNone(i)$, $\calG(i,a)$ and $\calG(i,b)$ are edge disjoint and share the only one common vertex $i$; otherwise $\calC(i,a)$ would contain $i$.
Furthermore, 
for each $a \in \calNNone(i)$, by definition of $ \calNNone(i)$, $\calC(i,a)$ does not contain $j$, which implies that $\calG(i,a) = \calC(i,a) \cup (i, a)$ does not either. 
Since $a \notin \calNT(i)$ (recall the definition of $\calNT(i)$ in \prettyref{eq:calN_T_i_fake}), 
each $\calG(i,a)$ has at least one cycle and hence $e(\calG(i,a)) \geq v(\calG(i,a))$. 
Therefore,
\begin{align*}
e(\UNone(i))
=\sum_{a \in \calNNone(i)} e(G(i,a))
\ge \sum_{a \in \calNNone(i)} v(G(i,a))
= v(\UNone(i))
+|\calNNone(i)|-1,
\end{align*}
where the last equality holds because the root $i$ has been over-counted by 
$|\calNNone(i)|-1$ times.
The proof of \prettyref{eq:E_U_1_j} follows similarly.

By definition of $\calNNone(i)$, for any $a \in \calNNone(i)$, $\calG(i,a)$ must not contain $j$;  for any $a \in \calNNone(j)$, $\calG(j,a)$ must not contain $i$. 
It follows that  $\UNone(i)$ and $ \UNone(j)$ are vertex-disjoint. Define
\begin{align}
    \calNNtwo(i)  = ~ \calNN(i)\backslash \calNNone(i) \,, \quad 
    \calNNtwo(j)  = ~ \calNN(j)\backslash \calNNone(j)  \, .  \label{eq:calN_N_2_i_j}
\end{align}
Let 
\begin{align} 
\UNtwo = \UN \backslash (\UNone(i) \cup \UNone(j))  \, . \label{eq:UN_2}
\end{align}
Since $\UNone(i)$ and $ \UNone(j)$ are vertex-disjoint, we have
$e(\UNtwo) = e(\UN)- e(\UNone(i))-e(\UNone(j))$, and $ v(\UNtwo) = v(\UN)- v(\UNone(i)-v(\UNone(j)) + 2$. 
By \prettyref{eq:target_1}, \prettyref{eq:E_U_1_i} and \prettyref{eq:E_U_1_j}, 
it suffices to show
\begin{align}
    e(\UNtwo) - v(\UNtwo) +  2 \ge \frac{1}{4}\left( |\calNNtwo(i)|+ |\calNNtwo(j)| \right) \, . \label{eq:target_2}
\end{align} 

To this end, first note that $\UNtwo$ is a connected graph. Indeed, by  \prettyref{claim:U_decompose_fake}, $\UN$ is a connected graph. Since $i,j\in V(\UN)$, there exists a path, say $i,u,\ldots,v,j$, where $u\in \calNN(i)$ and $v\in \calNN(j)$. In fact, by definition of $\calNNtwo(i),\calNNtwo(j)$, we have 
$u\in \calNNtwo(i)$ and $v\in \calNNtwo(j)$. Therefore this path is contained in $\UNtwo$. Since there is a path between any vertex and $i$ or $j$, we conclude that $\UNtwo$ is connected.

Next, we define a new graph $\widehat{U}_{\sfN_2}$ by adding a new vertex $d \not \in [n]$ and connect it to $i$ and $j$. 
Then, we have
\[
e(\widehat{U}_{\sfN_2})- v(\widehat{U}_{\sfN_2}) + 1=  
e(\UNtwo)- v(\UNtwo) + 2 \, .
\]
So to show \prettyref{eq:target_2}, it suffices to prove
\begin{align}
    e(\widehat{U}_{\sfN_2})- v(\widehat{U}_{\sfN_2}) + 1 \ge \frac{1}{4} \left(|\calNNtwo(i)|+|\calNNtwo(j)|\right)\,. \label{eq:target_3}
\end{align}
To show this, it is helpful to recall the following result (cf.~e.g.~\cite[Theorem 1.9.5]{diestel}): For any connected graph with $v$ vertices and $e$ edges, there exist a total of $e-v+1$ cycles (called a \emph{cycle basis}), such that any cycle can be expressed as a set symmetric difference of those in the cycle basis.

Since $\widehat{U}_{\sfN_2}$ is connected, it has a cycle basis of size $e(\widehat{U}_{\sfN_2})- v(\widehat{U}_{\sfN_2}) + 1 \equiv m$.
We claim that for each $a\in \calNNtwo(i)$, there exists a cycle in $\widehat{U}_{\sfN_2}$ that contains the edge $(i,a)$. Therefore, this cycle can be expressed as a set symmetric difference of cycles in the cycle basis and hence this edge $(i,a)$ is contained in one of the $m$ cycles in the cycle basis of $\widehat{U}_{\sfN_2}$. Since each cycle contains at most two edges incident to $i$, we have $2m \geq |\calNNtwo(i)|$.
Similarly, $2m \geq |\calNNtwo(j)|$.
This completes the proof of \prettyref{eq:target_3} and hence the lemma.

It remains to justify the above claim. 
Indeed, by \prettyref{eq:calN_N_2_i_j}, we have 
\begin{align*}
    \calNNtwo(i) & = ~  \calNN(i)\backslash \calNNone(i) =~ \{a \in \calNN(i): \text{ $ \calC(i,a) $ contains $i$ or $j$}\}  \, .
\end{align*}
The claim follows by considering the following two cases:
\begin{itemize}
    \item Suppose that $\calC(i,a)$ contains $i$. Then there exists a path from $a$ to $i$ in $\calC(i,a)\subset \UNtwo$, hence a cycle in $\widehat{U}_{\sfN_2}$ containing $(i,a)$;

    \item Suppose that $\calC(i,a)$ contains $j$ \emph{but not} $i$. By the same reasoning, there exists a path $(a, u, \ldots v, j)$ in $\calC(i,a)\subset \UNtwo$, such that $i$ does not lie on this path. Using the newly added vertex $d$, we have found a a cycle $(a,u,\ldots,v,j,d,i,a)$ in $\widehat{U}_{\sfN_2}$  that contains $(i,a)$.
\end{itemize}

\subsubsection{Proof of \prettyref{lmm:fake_enumeration}}


To enumerate $\UTone(i) \cup \UTone(j)$, the key is to recognize that, by \prettyref{eq:purebulb}, 
\begin{align}
    \calNTone(i)
= \calNT(i) \backslash \calNTtwo(i)
& = \left\{ a \in \calNT(i): \left|  \{e \in E(\calG(i,a))  : |D_e| \ge 3\} \right| < M \right\},
\end{align}
so that the overlap between $S_1\cup T_1$ and $S_2\cup T_2$ in $\calUTone$ only occurs on the wire part, not the bulb part. 
See~\prettyref{fig:fake_pair_decomp} 
for an illustration of $\UTone(i)$ and $\UTone(j)$. 

We first define a special family $\calF$ of \emph{unlabeled} decorated graphs.  
Fix integers $t,m,d$.
Fix a decorated $t$-chandelier $U'$ (recall \prettyref{def:chandelier}),
such that each edge of $U'$ is decorated by both $\SS_1$ and $\TT_1$. 
Let $U''$ denote a decorated tree rooted at $i$ satisfying the following conditions:
\begin{enumerate}
\item $U''$ has at most $m$ edges and the degree of $i$ in $U''$ is at most $d$;
    \item Each edge of $U''$ is decorated by $\SS_2$ or $\TT_2$, or both;
    \item All edges in $U' \cap U''$ must be on the wires of $U'$.
\end{enumerate}
For a fixed $U'$, let $\calF_{t,m,d}(U')$ denote the set of all possible $U'\cup U''$ as decorated unlabeled graphs. 
\begin{lemma}\label{lmm:diluted_chandelier_tree}
\begin{align}
\left| \calF_{t,m,d} (U')\right| \le  \sum_{e=0}^m (6\Universal)^e \sum_{\ell=0}^d \binom{t}{\ell} \ell! \le (6\Universal)^{m+1}  t^d.
\label{eq:calF_bound}
\end{align}
\end{lemma}
\begin{proof}
Recall that $U'$ is fixed. To enumerate $U'\cup U''$ up to isomorphisms, it suffices to enumerate the unlabeled version of $U''$ and its subgraph $V$, together with the embedding of $V$ as a subgraph of $U'$. 

First, suppose that $U''$ has $0 \le e \le m $ edges. Since $U''$ is a tree, 
its isomorphism class has at most $\Universal^{e}$ possibilities.
Moreover, the decoration of each edge has at most $3$ possibilities. Hence, there are at most $ (3\Universal)^{e}$ possibilities of $U''$. 
Next, choose $V=U'\cap U''$ to be one of $2^e$ different subgraphs of $U''$. 
By assumption, $V$ is be a tree rooted at $i$ with $\ell$ incident paths (of possibly different lengths) for some $0 \le \ell \le d$, and these paths are subgraphs of the $t$ wires in $U'$.
Therefore, the number of ways to embed $V$ in $U'$ is at most $\sum_{\ell=0}^d \binom{t}{\ell} \ell!$. 
Combining these factors together yields the first inequality in~\prettyref{eq:calF_bound}.
The second inequality holds due to $\sum_{e=0}^m (6\Universal)^e \le 2 (6\Universal)^m$ and 
$\sum_{\ell=0}^d \binom{t}{\ell} \ell! \le 2 t^d$.
\end{proof}

Next, we apply~\prettyref{lmm:diluted_chandelier_tree} to count $\UTone(i) \cup \UTone(j)$. 
Let $\Vprime$ (resp.~$\Vdoubleprime$) denote the subgraph of $\UTone(i)$ induced by all edges $e$ with $\{\SS_1, \TT_1\} \subset D_e$ (resp.~$\{\SS_2, \TT_2\} \cap D_e \neq \emptyset$) with decoration set $D'_e=\{\SS_1,\TT_1\}$ (resp.\ $D''_e=\{\SS_2, \TT_2\} \cap D_e$).
 Note that $\Vprime$ and $\Vdoubleprime$ need not be edge-disjoint; 
 nevertheless, we have 
 \[
 \UTone(i) = \Vprime \cup \Vdoubleprime
 \]
 because $|D_e|\geq 2$ for each edge $e$ so if 
 $\{\SS_2, \TT_2\} \cap D_e =\emptyset$, then we must have
 $D_e = \{\SS_1, \TT_1\}$.
 Similarly, let $\Wprime$ (resp.\ $\Wdoubleprime$) denote the subgraph of $\UTone(j)$ induced by all edges $e$ with
 $\{\SS_2, \TT_2\} \subset D_e$ (resp.\ $\{\SS_1, \TT_1\} \cap D_e \neq \emptyset$) with decoration
 $D'_e=\{\SS_2,\TT_2\}$
 (resp.\ $D''_e=\{\SS_1, \TT_1\} \cap D_e $),
 so that 
 $\UTone(j)=\Wprime\cup \Wdoubleprime$. 
%
 
Then we  have the following claim. 
\begin{claim}\label{claim:U_20}
\begin{enumerate}[label=(\roman*)]
    \item \label{C:4} $\Vprime$ (resp.\ $\Wprime$) is a  chandelier rooted at $i$ (resp.\ $j$) 
   

    \item \label{C:7} The total number of non-isomorphic bulbs in $\Vprime$ and $\Wprime$ is at most $L$. 
    
    \item \label{C:5} $\Vdoubleprime$ (resp.\ $\Wdoubleprime$) is either empty or a 
    tree rooted at $i$  (resp.\ $j$) with at most $2K+2M-2$ edges,
    and the degree of the root is at most $L \wedge (2K+1)$.

    \item \label{C:6} All edges in $\Vprime \cap \Vdoubleprime$ (resp.\ $\Wprime \cap \Wdoubleprime$) must be on the wires of $\Vprime$ (resp.\ $\Wprime$).
    
    \item \label{C:8} 
    $w(\UTone) \ge \aut(\UTone(i) )\aut(\UTone(j) )$, where $w(\UTone)$ is defined in \prettyref{eq:w_UL}.
    
\end{enumerate}
 \end{claim}

Now, we are ready to enumerate $\UTone(i) \cup \UTone(j)$. 
First, we enumerate the bulbs in the chandeliers $\Vprime$ and $\Wprime$. Note that there are at most a total of $L$ non-isomorphic bulbs, and each of them can be in $\Vprime$ or $\Wprime$ or both or none. 
Thus there are at most $ \binom{|\calJ|}{L} 4^ L = |\calT| 4^ L$  different choices for the isomorphic classes of $\Vprime$ and $\Wprime$. 
Next, fix one such choice.
Then $\Vprime \cup \Vdoubleprime$ belongs to 
$\calF_{t,e,d}(\Vprime)$ for some fixed $t \le L$, $e \le 2K+2M-2$, and $d \le L \wedge (2K+1)$,
and similarly for $\Wprime \cup \Wdoubleprime$. Therefore, applying \prettyref{lmm:diluted_chandelier_tree}, we get that the total number of unlabeled decorated $\UTone(i) \cup \UTone(j) = 
(\Vprime \cup \Vdoubleprime)\cup 
(\Wprime \cup \Wdoubleprime)$ is at most
$$
|\calT| 4^L (6\Universal)^{4K+4M-2}  L^{2L \wedge (4K+2)} .
$$


    

    
Finally, since there are at most 
$ \frac{n^{\ell}}{\aut(\UTone (i))\aut(\UTone (j))}$ different vertex labelings of $\UTone(i)\cup \UTone(j)$ (excluding $i$ and $j$) and by \prettyref{claim:U_20}~\ref{C:8}, $w(\UTone) \ge \aut(\UTone(i) )\aut(\UTone(j) )$, we deduce that 
\begin{align*}
\PTone(\ell) & =  \sum_{\UTone \in \calUTone(\ell)}  w(\UTone) \\
& \le |\calT| 4^L (6\Universal)^{4K+4M-2}  L^{2L \wedge (4K+2)}
\frac{n^{\ell}}{\aut(\UTone (i))\aut(\UTone (j))}  \aut(\UTone(i) )\aut(\UTone(j) )\\
& \le n^{\ell}|\calT| 4^L (6\Universal)^{4K+4M-2}  L^{2L \wedge (4K+2)},
\end{align*}
yielding the desired \prettyref{eq:fake_enumeration_tree>0}.

We finish by proving \prettyref{claim:U_20}. It suffices to focus on $\Vprime$ and $\Vdoubleprime$, as
the proofs for $\Wprime$ and $\Wdoubleprime$ are entirely analogous.
\begin{enumerate}[label=(\roman*)]
    \item 
In terms of graphs, we have 
\begin{align}
\Vprime &= \UTone(i) \cap S_1 \cap T_1 = 
\bigcup_{a \in \calNTone(i)} \calG(i,a) \cap S_1 \cap T_1  \label{eq:U20p}  \\
\Vdoubleprime&=  \UTone(i) \cap (S_2 \cup  T_2) =
\bigcup_{a \in \calNTone(i)} \calG(i,a) \cap (S_2 \cup T_2).\label{eq:U20pp}
\end{align}

We show that for each $a \in \calNTone(i)$, 
$\calG(i,a) \cap S_1 \cap T_1$ consists of a wire followed by a bulb. Furthermore, this bulb 
is contained in
both $S_1$ and $T_1$ but not $S_2$ or $T_2$.

Fix $a \in \calNTone(i) \subset \calNT(i)$. By definition, 
$(i,a) \in S_1 \cup T_1$. Without loss of generality, assume that $(i,a) \in S_1$. Then $\calG(i,a) \cap S_1$ consists of an $M$-wire $P \equiv (i,a,\ldots,b)$ and a $K$-bulb $\calB \equiv (S_1)_b$ rooted at $b$.
Consider an edge $e$ in $\calB$. 
Then there is a path $P'$ in $\calB$ from the root $b$ to $e$. Since $\calG(i,a)$ is a tree, the path $(P,P')$ is the unique path from $i$ to $e$ in the union graph $S_1\cup T_1\cup S_2\cup T_2$.
 Next we show that the entire bulb $\calB$ is contained in $S_1 \cap T_1$. Suppose that this fails for some $e$ in $\calB$. Then $e \in S_2 \cup T_2$, say, $e\in S_2$.
Since $S_2$ is a tree rooted at $j$, there exists a path $P''$ from $j$ to $e$ in $S_2$. 
This path must contain $(i,a)$, for otherwise $\calC(i,a)$ would contain $j$.
By the uniqueness of the path $(P,P')$, we conclude that 
\[
P''=(j,\ldots,\underbrace{i,a,\ldots,b}_{P},P').
\]
In particular, $P \in S_1\cap S_2$ and $e\in (S_2)_b$.
Furthermore, by the chandelier structure of $S_2$, since the subpath of $P''$ from $j$ to $b$ has at least $M+1$ edges, so that $|(S_2)_b| \leq K-1 < K = |\calB|$. 
Thus there exists an edge $e' \in \calB$ not in $S_2$, so $e' \in T_1 \cup T_2$, say $e'\in T_2$. Applying the same argument to $e'$, the $M$-wire $P$ is also in $T_2$ and hence in 
$S_1 \cap S_2 \cap T_2$, contradicting the definition of $a\in \calNTone(i)$.

In all, in terms of decorations, we have shown that for any $a\in \calNTone(i)$, any edge $e$ in the bulb part of $\calG(i,a)$ satisfies $D_e=\{\SS_1,\TT_1\}$ and any $e$ in the wire part satisfies $D_e \supset \{\SS_1,\TT_1\}$.

\item Note that $S_1 \cong S_2 \cong H$ and $T_1 \cong T_2 \cong I$ for some chandelier $H,I \in \calT$. In the proof of~\ref{C:4}, we have shown that every bulb  in $\Vprime$ is decorated by exactly $\{\SS_1, \TT_1\}$
and hence isomorphic to some rooted tree in $\calJ(H) \cap \calJ(I)$. The proof is complete since $|\calJ(H)| = |\calJ(I)| = L$.


\item 
Suppose $\Vdoubleprime \neq \emptyset$. In view of \prettyref{eq:U20pp}, we aim to show that for each $a\in \calNTone(i)$, 
$\calG(i,a) \cap (S_2\cup T_2)$ is either empty or a tree. 
Consider a connected component $C$ of $\calG(i,a) \cap (S_2\cup T_2)$. It suffices to show that $C$ contains the edge $(i,a)$.
Indeed, consider an edge $e \neq (i,a)$ in $C$. 
Then $e \in S_2\cup T_2$. Then there exists a path from $j$ to $e$ in $C$. This path must contain $(i,a)$ for otherwise $\calG(i,a)$ contains $j$. This shows $(i,a) \in C$ as desired, and, in addition, $\calG(i,a) \cap (S_2\cup T_2) \subset (S_2)_i \cup (T_2)_i$ for every $a \in \calNT(i)$.
As a result, $\Vdoubleprime \subset (S_2)_i \cup (T_2)_i$.
Note that since $i\neq j$, the subtree $(S_2)_i$ contains at most $M-1+K$ edges. Thus $\Vdoubleprime$ has at most $2K+2M-2$ edges. 
Finally, the neighbors of $i$ in $\Vdoubleprime$ are in $\calNTone(i)$ and $|\calNTone(i)| \le L$. 
Also,
there are at most $2K$ neighbors of $i$ in $(S_2)_i \cup (T_2)_i \supset \Vdoubleprime$. Therefore, the degree of $i$ in $\Vdoubleprime$ is at most $L \wedge 2K$.

\item This follows from the proof of~\ref{C:4}, where we have shown that $D_e=\{\SS_1, \TT_1\}$ for all edges $e$ in each bulb.  

\item 
Define $w(\UTone(i)) \triangleq ~ w_{S_1}(\UTone(i)) w_{T_1}(\UTone(i))$ and $w(\UTone(j))\triangleq ~ w_{S_2}(\UTone(j)) w_{T_2}(\UTone(j))$. It suffices to show that $w(\UTone(i)) \ge \aut(\UTone(i))$ and $w(\UTone(j)) \ge \aut(\UTone(j))$. 
 
 First, we show that $\aut(\UTone(i)) \ge \prod_{\calB\in \calK(\Vprime)} \aut(\calB)$, where recall that $\Vprime$ is a chandelier rooted at $i$ by \prettyref{claim:U_20}~\ref{C:4}.
 For any bulb $\B \in \calK(\Vprime)$, any of its automorphism naturally induces an automorphism of the entire $\UTone(i)$ by only permuting the node labels of $\calB$ according
 to the automorphism and keeping all the other node labels unchanged. Hence $\aut(\UTone(i)) \ge \prod_{\calB\in \calK(\Vprime)} \aut(\calB)$.

 Next, we show that $w(\UTone(i))=\prod_{\calB\in \calK(\Vprime)} \aut(\calB)$. 
 Fix any bulb $\B$ from $S_1$ contained in $\UTone$. Then, $\B \subset (S_1)_a \subset \calG(i,a)$ for some neighbor $a$ of $i$ in $S_1$.
 Since $\calB$ in $\calG(i,a)$ must be decorated by $\{\SS_1, \TT_1\}$, $\B$ must be contained in $\Vprime$ and $(S_1)_a =(T_1)_a$. Hence, $\B \in \calK(T_1)$. 
 Similarly, we can show that for any  bulb $\B$ from $T_1$ contained in $\UTone$, 
 $\B \in \calK(S_1)$. Therefore,
 $$
 \{\calB \in \calK(S_1), \calB \subset \UTone \}
 = \{\calB \in \calK(T_1), \calB \subset \UTone \}
 =\{\calB \in \calK(\Vprime)\}.
 $$
It follows that 
\begin{align*}
w(\UTone(i)) &= \prod_{\calB \in \calK(S_1), \calB \subset \UTone} \aut^{1/2}(\calB)
 \prod_{\calB \in \calK(T_1), \calB \subset \UTone} \aut^{1/2}(\calB) \\
& =\prod_{\calB\in \calK(\Vprime)} \aut(\calB).
\end{align*}
Therefore, we have $\aut(\UTone(i)) \ge w(\UTone(i))$. Analogous argument holds for $\aut(\UTone(j)) \ge w(\UTone(j))$.




\end{enumerate}

\section{Approximated similarity scores by color coding} \label{sec:color_coding}

In this section, following~\cite{mao2021testing}, 
we provide a polynomial-time algorithm to approximately compute the similarity scores $\{\Phi_{ij}\}_{i,j\in [n]}$  in \prettyref{eq:Phi_ij} when $\calT$ is the family of chandeliers\footnote{In fact, the algorithm does not rely on the chandelier structure and works for any trees.} 
of size $O(\log n)$, using the idea of color coding~\cite{alon1995color,alon2008biomolecular}. 

\paragraph{Approximate signed rooted subgraph count}
Let $H$ be a rooted connected graph with $N+1$ vertices.
For each $i\in [n]$,
 we first approximately count the signed graphs rooted at $i$ that are isomorphic to $H$. Specifically, given a weighted adjacency matrix $M$ on $[n]$, we generate a random coloring $\mu: [n] \to [N+1]$ that assigns a color to each vertex of $M$ from the color set $[N+1]$ independently and uniformly at random. 
Given any $V\subset[n]$, let $\chi_{\mu} (V)$ indicate that $\mu(V)$ is colorful, i.e., $\mu(x)\neq \mu(y)$ for any distinct $x,y \in V$. In particular, if $|V|=N+1$, then $\chi_\mu (V)=1$ with probability 
\begin{equation}
r \triangleq \frac{(N+1)!}{(N+1)^{N+1}}.
\label{eq:r}
\end{equation} 
Define
  \begin{align}
        X_{i,H} (M,\mu)
        & \triangleq  \sum_{S(i)\cong H} \chi_{\mu}(V(S)) \prod_{(u,v)\in E(S)} M_{uv} \, .\label{eq:X_H_M}
    \end{align}
Then $\expect{  X_{i,H}(M,\mu)}=r W_{i,H}(M)$, where $W_{i,H}(M)$ is defined in \prettyref{eq:W_i_H}. Hence, $X_{i,H} (M, \mu) /r$ is an unbiased estimator of $W_{i,H}(M)$.

When $H$ is a tree, the color coding together with the recursive tree structure enables us to use dynamic programming
to count colorful trees and compute 
$X_{i,H} (M, \mu)$ efficiently. 
This is summarized as \cite[Algorithm 2]{mao2021testing} for unrooted trees and the same algorithm with minor adjustments also works for rooted trees. First, since
$H$ is already a rooted tree, the step of assigning an arbitrary vertex of $H$ as its root is not needed and thus the rooted tree $T_N$  constructed is exactly $H$ itself. 
Second, as an intermediate step, \cite[Algorithm 2]{mao2021testing} computes $Y(i,T_N,  [N+1], \mu )$, which is the same as
$\aut(H) X_{i,H} (M, \mu)$. Hence, 
we can simply output $ \frac{1}{\aut(H)}  Y(i,T_N,  [N+1], \mu )$ as the rooted tree count $X_{i,H} (M, \mu)$.



Finally, we generate independent random colorings $\mu_1,\ldots,\mu_t$ and  average over  
$X_{i,H} (M ,\mu_m )$'s to better approximate $W_{i,H}(M)$, where we set
\[
t \triangleq \ceil{1/r}{}.
\]

\paragraph{Approximate similarity scores}

To approximate $\Phi_{ij} \equiv  \Phi_{ij}^{\calT}$ in \prettyref{eq:Phi_ij}, we apply the above idea to each chandelier $H\in \calT$. 
Generate $2t$ random colorings $\{\mu_a\}_{a=1}^t$ and $\{\nu_a\}_{a=1}^t$ which are independent copies of $\mu$ that map $[n]$ to $[N+1]$. Define
    \begin{align}
         \tilde{\Phi}_{ij}
        & \triangleq   \frac{1}{r^2} \sum_{H\in \calT}\aut(H)
         \left(\frac{1}{t} \sum_{a=1}^t  X_{i,H}(\Bar{A},\mu_a ) \right) \left(\frac{1}{t} \sum_{a=1}^t  X_{j,H} (\Bar{B},\nu_{a}) \right) 
        \label{eq:tilde_Phi_ij} \, .
   \end{align} 
Then $\expect{\tilde{\Phi}_{ij} \mid A, B}= \Phi_{ij}$. Moreover,
 the following result bounds the approximation error under the same 
 conditions as those in Propositions \ref{prop:true_pair}
 and \ref{prop:fake_pairs_improved} for the second moment calculation.

\begin{proposition}\label{prop:tilde_Phi_ij}
For any $i\in[n]$, if \prettyref{eq:true_pair_constraint} holds,
\begin{align}
    \frac{\Var[\tilde{\Phi}_{ii} - \Phi_{ii}]}{ \expect{\Phi_{ii} }^2} 
    =  O
     \left(\frac{L^2 }{\rho^2 nq} + \frac{L^2 }{\rho^{2(K+M)}|\calJ|} \right) \, ; \label{eq:tilde_upper_bound_true_pair}
\end{align}
for any $i\neq j$, if \prettyref{eq:fake_pairs_constraint} holds, 
\begin{align}
    \frac{\Var[\tilde{\Phi}_{ij} -  \Phi_{ij}]}{ \expect{\Phi_{ii} }^2} =O\left(\frac{1}{|\calT| \rho^{2N} }\right)  \,.
    \label{eq:tilde_upper_bound_fake_pair}
\end{align}
\end{proposition}

Finally, we show that the approximate similarity scores $\tilde{\Phi}_{ij}$ can be computed efficiently using \prettyref{alg:approximated_Phi_ij}. 
\begin{algorithm}
    \caption{Approximate similarity scores via color coding}\label{alg:approximated_Phi_ij}
    \begin{algorithmic}[1]
    \State{\bfseries Input:} Centered adjacency matrices $\bar A$ and $\bar B$ and integers $K,L,M,N, R$.
    \State 
    Apply the algorithm for generating rooted trees in \cite[Sec.~5]{beyer1980constant}
    to
    list all non-isomorphic rooted trees with $K$ edges, compute $\aut(H)$ for each rooted tree using the 
    automorphism
    algorithm for trees in \cite[Sec.~2]{colbourn1981linear}, and return $\calJ$ as the subset of rooted trees whose number of automorphisms is at most $R$.
    \State Generate $(K,L,M,R)$-chandeliers using $\calJ$ to obtain $\calT$ per \prettyref{def:chandelier}. 
    \State Generate \iid random colorings $\{\mu_a\}_{a=1}^t$ and $\{\nu_a\}_{a=1}^t$ mapping $[n]$ to $[N+1]$. 
    \For{each $a = 1 ,\cdots, t$}
    \State{\bfseries} 
    For each $H \in \calT$, compute $\{X_{i,H}(\overline{A},\mu_a)\}_{i\in [n]}$ and $\{X_{j,H}(\overline{B},\nu_a)\}_{j\in [n]}$ via \cite[Algorithm 2]{mao2021testing} with adjustments described after~\prettyref{eq:X_H_M}.
    \EndFor
    \State{\bfseries Output:}  $\{\tilde{\Phi}_{ij}\}_{i,j \in [n]}$ according to  \prettyref{eq:tilde_Phi_ij}.
    \end{algorithmic}
\end{algorithm}

\begin{proposition}\label{prop:Y_calT_computation}
\prettyref{alg:approximated_Phi_ij} computes $\{\tilde{\Phi}_{ij}\}_{i,j\in [n]}$ in time $O\left(n^2 (3e \alpha)^N \right)$. 
Furthermore, when $nq \ge 2$, under
the choice of $K,L,M, \R \in \naturals$ as per~\prettyref{eq:K_L_M_R_simple}, the time complexity is $O(n^{ c/\epsilon})$, where
$\epsilon$ is from \prettyref{eq:K_L_M_R_simple}
and
$c$ is an  absolute constant.
\end{proposition}


\begin{proof}[Proof of \prettyref{thm:tilde_Phi_ij_almost}]
Note that 
\begin{align}
    \Var[\tilde{\Phi}_{ij}] 
    & =  ~ \Var[\tilde{\Phi}_{ij} - \Phi_{ij}] +  \Var[\Phi_{ij}] + 2\Cov\left(\tilde{\Phi}_{ij} - \Phi_{ij},\Phi_{ij}\right) \nonumber \\
    & = ~ \Var[\tilde{\Phi}_{ij} - \Phi_{ij}] +  \Var[\Phi_{ij}] \,, \label{eq:Var_Phi_ii}
\end{align}
where the last equality holds because 
$\expect{\tilde{\Phi}_{ij} |A,B} =\Phi_{ij}$  and so
\begin{align*}
    \Cov\left(\tilde{\Phi}_{ij} - \Phi_{ij},\Phi_{ij}\right) 
  =\expect{ \expect{(\tilde{\Phi}_{ij} - \Phi_{ij})|A,B}\Phi_{ij}} = 0 \,.
\end{align*}


Under the assumption of \prettyref{thm:Phi_ij_almost},
both 
\prettyref{eq:true_pair_constraint}
and \prettyref{eq:fake_pairs_constraint} hold.
Since $ \expect{\tilde{\Phi}_{ij}} = \expect{ \Phi_{ij}}$, applying  \prettyref{prop:tilde_Phi_ij} yields
\begin{align*}
    \frac{\Var[\tilde{\Phi}_{ii}]}{ \expect{\tilde{\Phi}_{ii}}^2} 
    =  O\left(\frac{L^2 }{\rho^2 nq} + \frac{L^2 }{\rho^{2(K+M)}|\calJ|} \right) \, ;
\end{align*}
for all $i$ and 
\begin{align*}
    \frac{\Var[\tilde{\Phi}_{ij} ]}{ \expect{\tilde{\Phi}_{ii} }^2} =O\left(\frac{1}{|\calT| \rho^{2N} }\right)  \,.
\end{align*}
for all $i\neq j$. 
In other words, 
Propositions \ref{prop:true_pair}-- \ref{prop:fake_pairs_improved} and hence \prettyref{thm:Phi_ij_almost} continue to hold with $\tilde{\Phi}_{ij} $ in place of $\Phi_{ij}$. The time complexity follows from  \prettyref{prop:Y_calT_computation}.
\end{proof}

\subsection{Proof of \prettyref{prop:tilde_Phi_ij}}


The proof is similar to that  of \cite[Proposition 3 in Section C.1]{mao2021testing}. 
Define
 \begin{align}
        Y_{ij}(\mu,\nu)
        & \triangleq \sum_{H \in \calT}  \aut(H) X_{i,H}(\Bar{A},\mu)   X_{j,H}(\Bar{B},\nu)   \, . \label{eq:Y_calT_ij} 
\end{align}
where $\mu,\nu$ are two independent $(N+1)$-coloring of the vertices in $[n]$.
Then 
\begin{align*}
        \tilde{\Phi}_{ij}= \frac{1}{r^2 t^2} \sum_{a=1}^t \sum_{b=1}^t Y_{ij}(\mu_a,\nu_b )  \,.
\end{align*}
Note that for any $1\le a,b\le t$, $ Y_{ij}(\mu_a,\nu_b) /r^2 $ is an unbiased estimator of $\Phi_{ij}$ as
   \begin{align}
       \expect{Y_{ij}(\mu_a,\nu_b)  \mid A,B }
       & =  r^2 \sum_{H\in \calT} \aut(H) W_{i,H}(\bar{A}) W_{j,H}(\bar{B}) = r^2 \Phi_{ij} \, . \label{eq:E_X_calH_A_B}
   \end{align} 
Moreover, $\{Y_{ij}(\mu_a,\nu_b)\}_{1\le a,b \le t}$ are identically distributed. And conditional on $A$ and $B$, for any $1\le a , b ,c ,d\le t$,  $Y_{ij}(\mu_a,\nu_b)$ and $Y_{ij}(\mu_c,\nu_d)$ are independent if and only if $ a \neq c$ and $b \neq d$. Hence, we have $\expect{\tilde{\Phi}_{ij} - \Phi_{ij}} = 0 $.

Next, we bound the variance of $\tilde{\Phi}_{ij}$. In particular,  we get that
   \begin{align}
       \Var\left(\tilde{\Phi}_{ij} - \Phi_{ij}\right) 
       & = \Var \left(
       \expect{\tilde{\Phi}_{ij} - \Phi_{ij} \mid A, B} \right)
       + \Expect\left[ \Var\left(\tilde{\Phi}_{ij} - \Phi_{ij} \mid A, B \right)\right] \notag \\
       & =  \Expect\left[ \Var\left(\tilde{\Phi}_{ij} \mid A, B \right)\right] \nonumber \\
       & \le 
       \frac{1}{r^4 t^4} \sum_{a=1}^t\sum_{b=1}^t \sum_{c=1}^t\sum_{d=1}^t  \Expect\left[ \Cov\left( Y_{ij}(\mu_a,\nu_b),  Y_{ij}(\mu_c,\nu_d) \mid A, B \right)\right]\, ,\label{eq:total_var}
   \end{align}
   where the second equality holds because 
   $
\expect{\tilde{\Phi}_{ij} \mid A, B} =  \Phi_{ij}$.

Next, we introduce an auxiliary result, bounding 
the conditional covariance.

\begin{lemma} \label{lmm:cov_Phi_ij}
  For any $1\le a,b,c,d \le t$,  and $i,j \in [n]$,
  \begin{align*}
    &\expect{\Cov\left(Y_{ij}(\mu_a,\nu_b),  Y_{ij}(\mu_c,\nu_d) \mid A, B \right)} \le ~  \left( r^{2+ \indc{b \neq d }  + \indc{a \neq c} } -r^4\right) \Gamma_{ij} \, ,
  \end{align*}
  where $\Gamma_{ij}$ is defined in \prettyref{eq:var_Phi_ij}.
\end{lemma}
Then, by \prettyref{eq:total_var}, we get
\begin{align*}
    \Var\left[\tilde{\Phi}_{ij}- \Phi_{ij}  \right]
    & \le~ \frac{1}{r^4 t^4} \sum_{a=1}^t \sum_{b=1}^t \sum_{c=1}^t \sum_{d=1}^t \left( r^{2+  \indc{b \neq d }  + \indc{a \neq c} } -r^4\right)   \Gamma_{ij}\\
    & \le~  \left(\frac{1}{t^2r^2}  + \frac{2}{tr}  \right)   \Gamma_{ij}\\
    & \le ~ 3 \Gamma_{ij} \, ,
\end{align*}
where the inequality holds because $t = \ceil{1/r}{}$. 
By \prettyref{eq:true_pair_constraint} and \prettyref{eq:Xi_ii_mean_square}, \prettyref{eq:tilde_upper_bound_true_pair} follows; by \prettyref{eq:fake_pairs_constraint} and \prettyref{eq:Xi_ij_mean_square}, \prettyref{eq:tilde_upper_bound_fake_pair} follows. Hence, \prettyref{prop:tilde_Phi_ij}
follows.

We are left to prove \prettyref{lmm:cov_Phi_ij}.
\begin{proof}[Proof of \prettyref{lmm:cov_Phi_ij}]

\begin{align*}
     & \Cov\left( Y_{ij}(\mu_a,\nu_b),  Y_{ij}(\mu_c,\nu_d) \mid A, B \right) \\
     & = ~ \expect{Y_{ij}(\mu_a,\nu_b) Y_{ij}(\mu_c,\nu_d) \mid A, B  } - \expect{Y_{ij}(\mu_a,\nu_b) \mid A, B  }\expect{Y_{ij}(\mu_c,\nu_d) \mid A, B  }\\
     & = ~ \expect{Y_{ij}(\mu_a,\nu_b) Y_{ij}(\mu_c,\nu_d) \mid A, B  } - r^4 \Phi_{ij}^2 \,,
\end{align*}
where the last equality holds because $Y_{ij}(\mu,\nu)/r^2$ is an unbiased estimator of $\Phi_{ij}$. 
Next, by \prettyref{eq:Y_calT_ij} and \prettyref{eq:X_H_M}, we get 
\begin{align*}
    \expect{Y_{ij}(\mu_a,\nu_b) Y_{ij}(\mu_c,\nu_d) \mid A, B } 
    & = ~ \sum_{H, I\in \calT} \aut(H)  \aut(I)\sum_{S_1(i),S_2(j) \cong H}  \sum_{T_1(i),T_2(j) \cong I} \\
    &~~~~~  \expect{\chi_{\mu_a} (V(S_1)) \chi_{\mu_c}(V(T_1))}  \expect{\chi_{\nu_b}(V(S_2)) \chi_{\nu_d}(V(T_2))}\\ 
    &~~~~~ 
    \Bar{A}_{S_1}  \Bar{B}_{S_2}  \Bar{A}_{T_1} 
    \Bar{B}_{ T_2} \,.
\end{align*}

Then, we have
\begin{align*}
    \expect{\Cov\left( Y_{ij}(\mu_a,\nu_b),  Y_{ij}(\mu_c,\nu_d) \mid A, B \right)} 
    & = ~ \sum_{H, I\in \calT} \aut(H)  \aut(I)\sum_{S_1(i),S_2(j) \cong H}  \sum_{T_1(i),T_2(j) \cong I}\\
    &~~~~~ \left( \expect{\chi_{\mu_a} (V(S_1)) \chi_{\mu_c}(V(T_1))}  \expect{\chi_{\nu_b}(V(S_2)) \chi_{\nu_d}(V(T_2))} -r^4 \right)\\ 
    &~~~~~ 
    \expect{\Bar{A}_{S_1}  \Bar{B}_{S_2}  \Bar{A}_{T_1} 
    \Bar{B}_{ T_2} } \,.
\end{align*}
Hence, we have
\begin{align*}
    & \expect{\Cov\left( Y_{ij}(\mu_a,\nu_b),  Y_{ij}(\mu_c,\nu_d) \mid A, B \right)} \\
    & \overset{(a)}{\le} ~ \sum_{H, I\in \calT} \aut(H)  \aut(I)\sum_{S_1(i),S_2(j) \cong H}  \sum_{T_1(i),T_2(j) \cong I}  \expect{\Bar{A}_{S_1}  \Bar{B}_{S_2}  \Bar{A}_{T_1} 
    \Bar{B}_{ T_2} }\\
    &~~~~~ \left( \expect{\chi_{\mu_a} (V(S_1)) \chi_{\mu_c}(V(T_1))}  \expect{\chi_{\nu_b}(V(S_2)) \chi_{\nu_d}(V(T_2))} - r^4\right)\\
    &~~~~~   \indc{S_1 \neq S_2 \text{ or } T_1 \neq T_2, \text{ or }V(S_1)\cap V(T_1) \neq \{i\}}  \indc{S_1\Delta T_1 \subset S_2 \cup T_2 \, , \, S_2 \Delta T_2 \subset S_1 \cup T_1 }\\
    & \overset{(b)}{\le} ~ \left( r^{2+ \indc{b \neq d }  + \indc{a \neq c} } -r^4\right) \\
    &~~~~~ \sum_{H, I\in \calT} \aut(H)  \aut(I)\sum_{S_1(i),S_2(j) \cong H}  \sum_{T_1(i),T_2(j) \cong I} \left|\expect{\Bar{A}_{S_1}  \Bar{B}_{S_2}  \Bar{A}_{T_1} 
    \Bar{B}_{ T_2} } \right|  \\
    &~~~~~   \indc{S_1 \neq S_2 \text{ or } T_1 \neq T_2, \text{ or }V(S_1)\cap V(T_1) \neq \{i\}}  \indc{S_1\Delta T_1 \subset S_2 \cup T_2 \, , \, S_2 \Delta T_2 \subset S_1 \cup T_1 } \,,
\end{align*}
where $(a)$ holds because if $S_1 = S_2$, $T_1=T_2$, and $V(S_1)\cap V(T_1) = \{i\}$, 
\[
\expect{\chi_{\mu_a} (V(S_1)) \chi_{\mu_c}(V(T_1))} = \expect{\chi_{\nu_b}(V(S_2)) \chi_{\nu_d}(V(T_2))} = r^2 \,;
\]
$(b)$ holds because 
\[
\expect{\chi_{\mu_{a}}(V(S_1)) \chi_{\mu_{c}}(V(T_1))} \le r^{1+\indc{a \neq c }}, \quad \expect{\chi_{\nu_b}(V(S_2)) \chi_{\nu_{d}}(V(T_2))}\le r^{1+\indc{b \neq d }} \, .
\]

\end{proof}

\subsection{Proof of \prettyref{prop:Y_calT_computation}}

 The algorithm in \cite{dinneen2015constant,wright1986constant}  generates all unlabeled rooted trees with $K$ edges in time $O\left(\alpha^K\right)$.
 Moreover,  for each generated rooted trees, the algorithm in  \cite{colbourn1981linear} computes $\aut(H)$ in time $O(K)$. Hence, the total time complexity to generate $\calJ$ is $O \left( K\alpha^K \right)$. Given $\calJ$, the time complexity to generate $\calT$ is $O(\binom{|\calJ|}{K})=O(|\calT|)$.


As mentioned after~\prettyref{eq:X_H_M}, we can slightly adjust \cite[Algorithm 2]{mao2021testing} to compute $\{X_{i,H}(M,\mu)\}_{i\in [n]}$ given any rooted tree $H$ with $N$ edges, a weighted graph $M$ on $[n]$, and a coloring $\mu: [n] \to [N+1]$. Thus, by~\cite[Lemma 2]{mao2021testing},
for each iteration $m$ in \prettyref{alg:approximated_Phi_ij},  $\{X_{i,H}(\overline{A},\mu_m)\}_{i\in [n]}$ and $\{X_{j,H}(\overline{B},\nu_m)\}_{j\in [n]}$ for all $H \in \calT$ can be computed in $O(|\calT| N 3^{N} n^2)$ time.

 Since $ t = \ceil{1/r} = O(e^{N})$, the total time complexity of \prettyref{alg:approximated_Phi_ij} to output $\{\tilde{\Phi}_{ij}\}_{i,j\in [n]}$ is 

\[
    O\left( K \alpha^K + |\calT| + |\calT| N (3e)^{N} n^2\right) =   O\left(\binom{|\calJ|}{L} N (3e)^{N} n^2 \right) = O\left(\left(3e\beta \right)^{N} n^2 \right) \, ,
\]
where the last equality holds because $(K+M)L =N $, and $|\calJ| \le \beta^{K}$ in view of~\prettyref{eq:Universal}.

Finally, under condition \prettyref{eq:K_L_M_R_simple} and $nq \ge 2$, we have
$$
N=(K+M) L = \frac{C_1 C_2 \log n }{ \epsilon } \left( 1+ \frac{C_3}{\log (nq)} \right)
\le \frac{C_1 C_2 \log n }{ \epsilon } \left( 1+ \frac{C_3}{\log 2} \right).
$$
Hence, the total time complexity is $O \left(n^C\right) $ 
where 
$
C = c/\epsilon
$
for some universal constant $c>0$.

\section{Seeded graph matching}
\label{sec:boost}



Recall that with high probability~\prettyref{alg:GMCC} applied to the class $\calT$ of chandeliers finds a  set $I$ with $|I|=n-o(n)$ and recovers the latent permutation $\pi$ on $I$.
In this section, we develop a seeded graph
matching subroutine (\prettyref{alg:recovery}) that matches the remaining vertices, thereby achieving exact recovery.
Since the seed set $I$ depends on graphs $A$ and $B$, 
we need to show that \prettyref{alg:recovery} succeeds even if the seed set $I$ is chosen adversarially as long as  $|I|=(1-o(1)) n$.

Given $I'\subset [n]$ and an injection $\pi': I' \to [n]$, for any vertex $i$ in $A$ and vertex $j$ in $B$, denote by $\sfN_{\pi'} (i,j)$ the number of common neighbors 
of $i$ and $j$ under the vertex correspondence $\pi'$, namely, 
the number of vertex $u \in  I'$ such that $u$ is a neighbor of $i$ in $A$ and $\pi'(u)$ is a neighbor of $j$ in $B$. 
\begin{algorithm}[H]
\caption{Seeded graph matching}\label{alg:recovery}
\begin{algorithmic}[1]
\State{\bfseries Input:} $A$ and $B$,  a mapping $\hat{\pi}: I \to [n]$, and $\gamma$.
\State{\bfseries}Let $J= I$ and $\tilde{\pi} = \hat{\pi}$. 
\While {there exists $i \notin J$
and $j \notin \tilde{\pi}(J)$ such that $\sfN_{\tilde{\pi}}(i,j) \ge \gamma (n-2) q^2$}
\State Add $i$ to $J$ and let $\tilde{\pi}(i)=j$.
\EndWhile
\State{\bfseries Output:}  $\tilde{\pi} $.
\end{algorithmic}
\end{algorithm}

\prettyref{alg:recovery} keeps adding vertices as new seeds  once we are confident that they are true pairs based on the current seed set, in a similar fashion as the percolation graph matching proposed in \cite{yartseva2013performance}. It is a simplified version of \cite[Algorithm 3.22]{barak2019nearly}, since our initial seeds are guaranteed to be error-free (thanks to \prettyref{thm:Phi_ij_almost}) and so there is no need to clean up any mismatch. 
This allows us to show our \prettyref{alg:recovery} succeeds under the information-theoretic necessary condition of  $nq(q+ \rho(1-q) )\ge (1+\epsilon) \log n $, whereas their algorithm requires $nq(q+ \rho(1-q) )>\log^C n$ for some constant $C>1$. 
Another similar algorithm in prior work is \cite[Algorithm~4]{mao2021exact}, which however requires $nq \le \sqrt{n}/\log n$.

The following proposition gives sufficient conditions for our seeded algorithm to achieve exact recovery.
Let 
\begin{align}
h(x)=x \log x - x +1 \label{eq:rate_fun}
\end{align}
for $x>0$, which is a convex function with the minimum value $0$ achieved at $x=1$.

\begin{proposition}\label{prop:seeded}
Fix an arbitrarily small constant $\epsilon>0$.
 Suppose $A,B \sim \calG(n,q,\rho)$
with $ q \le \frac{1}{2}$, $nq(q+ \rho(1-q) ) \ge (1+\epsilon) \log n$,  and $ \rho \ge \epsilon$. 
Let $\hat{\pi}\equiv \hat{\pi}(A,B)$ 
denote a mapping: $I \to [n]$
such that $\hat{\pi} = \pi|_{I}$ and $|I| \ge \left(1- \epsilon/16 \right )n$.
Let $\gamma$ denote the unique solution in $(1, +\infty)$
to $h(\gamma)= \frac{3 \log n}{(n-2)q^2}$.
Then with probability at least $1-o(1)$, 
\prettyref{alg:recovery} with inputs $\hat{\pi}$ and $\gamma$ outputs $ \tilde{\pi} = \pi$   in $O(n^3q^2)$ time. 
\end{proposition}



\begin{proof}[Proof of \prettyref{thm:exact_reovery}]

\prettyref{thm:Phi_ij_almost} ensures that, with probability $1-o(1)$, 
\prettyref{alg:GMCC} returns a mapping $\hat{\pi}: I \to [n]$ in time $O(n^C)$ such that $\hat{\pi}=\pi|_I$ and $I \ge (1-\epsilon/16) n$. 
Furthermore, \prettyref{prop:seeded} implies that, with probability $1-o(1)$, 
\prettyref{alg:recovery} outputs $\tilde{\pi} = \pi$  in $O(n^3 q^2)$ time.
Hence, \prettyref{thm:exact_reovery} follows. 
\end{proof}

\subsection{Proof of \prettyref{prop:seeded}}
To show~\prettyref{alg:recovery} eventually matches all vertices, we need a key lemma below, proving that an \ER random graph $G\sim \calG(n,p)$ satisfies a nice expansion property with high probability.
It extends \cite[Lemma 3.26]{barak2019nearly} from
$np \ge c \log n$ for a large constant $c$ to $np \ge (1+\epsilon) \log n$ by restricting to $|I| \le \epsilon n /16$.


\begin{lemma}\label{lmm:crossing_edge}
Suppose $G \sim \calG(n,p)$ with $np \ge (1+\epsilon)\log n$. 
Let $e_G(I,I^c)$ denote the total number of edges between vertices in $I$ and vertices in $I^c \triangleq [n]\setminus I$ in graph $G$.
With probability at least $1-n^{-\epsilon/8}$, for all subsets $I \subset [n]$ with $|I|\le \frac{\epsilon}{16}n$,
$e_{G}(I,I^c) \ge  \eta |I||I^c| p$, where 
$\eta$ is the unique solution in $(0,1)$ such that 
$h(\eta) = \frac{(1+\epsilon/8)\log n}{(1-\epsilon/16) np}$. In particular,
\begin{align}
\eta \ge \max \left\{ \frac{\epsilon}{16}, 1- \sqrt{ \frac{2(1+\epsilon/8)\log n}{(1-\epsilon/16) np} }\right\}.
\label{eq:eta_lower}
\end{align}

\end{lemma}
\begin{proof}
Note that the conclusion trivially holds when $I$ is empty. 
Fix a nonempty set $I\subset[n]$ with size $k$ where $1 \le k < \frac{\epsilon}{16}n$. Then $e_{G}(I, I^c) \sim \Binom(k(n-k),p)$. Letting $\mu=k(n-k)p$ and using the multiplicative Chernoff bound for binomial distributions~\cite[Theorem 4.5]{ProbabilityComputing05}, we get that for any $x \in (0,1)$, 
$$
\prob{e_{G}(I,I^c) \le x k(n-k) p} 
 \le \exp\left( - \mu h(x) \right),
$$
where $h(x)$ is defined in~\prettyref{eq:rate_fun}. 
Since  $np \ge (1+\epsilon) \log n$, it follows that
\begin{align}
\frac{(1+\epsilon/8) \log n}{(1-\epsilon/16) np} \le \frac{(1+\epsilon/8)}{(1-\epsilon/16)(1+\epsilon)}<1 .
\label{eq:expansion_1}
\end{align}
Therefore, there exists a unique solution $\eta$ in $(0,1)$ such that 
$h(\eta)=\frac{(1+\epsilon/8)\log n}{(1-\epsilon/16) np}$. 
By choosing $x=\eta$, we have
$$
\prob{e_{G}(I,I^c) \le \eta k(n-k) p} 
 \le \exp\left( - \mu \frac{(1+\epsilon/8)\log n}{(1-\epsilon/16) np}  \right)
 \le \exp \left( - (1+\epsilon/8) k \log n \right) ,
$$
where the last inequality holds because 
$\mu =k(n-k)p \ge k(1-\epsilon/16) np$.

 Note that the total number of $I \subset [n]$ with $|I|=k$ is $\binom{n}{k} \le n^k$. 
Thus, by a union bound,
\begin{align*}
    & \prob{\exists I \subset [n] \text{ s.t.  $1\le |I| < \frac{\epsilon}{16} n$}, \, \, e_{G}(I,I^c) 
    \le \eta |I|\left(n-|I|\right) p} \\
    & \le ~ \sum_{k=1}^{\floor{\frac{\epsilon}{16} n}{}}
    \exp\left(k  \log n - \left(1+\frac{\epsilon}{8}\right) k \log n \right) \\
    & = ~ \sum_{k=1}^{\floor{\frac{\epsilon}{16} n}{}} \exp\left( -  \frac{\epsilon}{8}k\log n  \right)  
 = O\left(n^{-\frac{\epsilon}{8}}\right)\,.
\end{align*}


Finally, we show $\eta$ satisfies the lower bound  in~\prettyref{eq:eta_lower}. 
Since 
$$
h\left( \frac{\epsilon}{16} \right) \ge \frac{(1+\epsilon/8)}{(1-\epsilon/16)(1+\epsilon)}
\ge \frac{(1+\epsilon/8) \log n}{(1-\epsilon/16) np},
$$
and $h(x)$ is decreasing over $x \in (0,1)$, we have
$\eta \ge \epsilon/16$. 
Moreover, in view of $h(x) \ge (1-x)^2/2$ for $x \in (0,1)$, we have
$$
\eta \ge 1- \sqrt{ \frac{2(1+\epsilon/8)\log n}{(1-\epsilon/16) np} }. 
$$

\end{proof}

\begin{proof}[Proof of~\prettyref{prop:seeded}] Without loss of generality, we assume $\pi = \id$.
For any $u,v \in [n]$, let $\sfN(u,v)$ denote the number of common neighbors of $u$ and $v$ in $A\cap B$. 
If $u\neq v$, then we have
$
\sfN (u,v) \iiddistr \Binom(n-2, q^2) \,. 
$
Let $\mu \equiv  (n-2)q^2$.  
By the multiplicative Chernoff bound for Binomials~\cite[Theorem 4.4]{ProbabilityComputing05},
we get
\begin{align*}
\prob{\sfN(u,v) \ge \gamma \left(n-2 \right)q^2 }  
& \le 
 \exp\left(- \mu h(\gamma) \right) = n^{-3} \,,
\end{align*}
where the last
equality applies  $h(\gamma) = \frac{3\log n}{(n-2)q^2}$.
By a union bound, we have
\begin{align}
    \prob{\exists u \neq v,\text{ s.t. }\sfN\left(u,v\right) \ge \gamma \left(n-2 \right)q^2 } 
    & \le n^2 n^{-3} = n^{-1} \, . \label{eq:exists_UNeq_pi_v}
\end{align} 

Henceforth, we assume that
$\sfN(u,v) < \gamma(n-2)q^2$
for all $ u \neq v$. We prove by induction
that $\tilde{\pi}=\id|_J$ throughout the entire course of~\prettyref{alg:recovery}.
This certainly holds at initialization when $J=I$ and
$\tilde{\pi}=\hat{\pi}$ by our standing assumption. Now, suppose this continues to hold up to the $t$-th execution of the while-loop (line 3-5) in~\prettyref{alg:recovery}. Then
at the $(t+1)$-th execution, since
$\tilde{\pi}=\id|_J$ by the induction hypothesis, it follows that for all $j \neq i$,
$$
\sfN_{\tilde{\pi}}(i,j) \le \sfN(i,j) < \gamma(n-2)q^2.
$$
Thus, at the $(t+1)$-th iteration, either the while-loop terminates 
or some vertex $i$ is added to $J$ with $\tilde{\pi}(i)=i$.
Either way, we have $\tilde{\pi}=\id|_J$. 

Next, we show that at the end of the while-loop, it must hold that 
$J=[n]$ so that $\tilde{\pi}=\id$.
Suppose not. Then by definition, for any $i\in J^c $, $N_{\tilde{\pi}}(i,i) < \gamma (n-2) q^2$. 
Then, by the fact that $\tilde{\pi}=\id|_J$, we have that
\begin{align}
    e_{A \cap B}(J,J^c) 
    =\sum_{i \in J^c} N_{\tilde{\pi}}(i,i) 
    <  \gamma (n-2) q^2 |J^c|. \label{eq:e_G_upperbound}
\end{align}
On the contrary, since $A \cap B \sim \calG(n,qs)$ for $s\triangleq q+\rho(1-q)$,
by \prettyref{lmm:crossing_edge}, with probability at least $1-n^{-\epsilon/8}$, we have 
\begin{align}
e_{A \cap B}(J,J^c)\ge  \eta |J||J^c| qs \ge 
\eta  \left(1- \frac{\epsilon}{16} \right )n |J^c| qs  \,, \label{eq:e_G_lowerbound}
\end{align}
where $\eta$ satisfies~\prettyref{eq:eta_lower} with $qs$ in place of $p$ and the last inequality holds because $|J|\ge |I|
\ge \left(1- \frac{\epsilon}{16} \right )n $.
We claim that 
\begin{align}
\gamma \le \eta  \left(1- \frac{\epsilon}{16} \right )
 \frac{s}{q} \triangleq \bar{\gamma}, \label{eq:claim_gamma}
\end{align}
which implies that \prettyref{eq:e_G_upperbound} and~\prettyref{eq:e_G_lowerbound} contradict with each other. Therefore, we must have $J=[n]$ so that $\tilde{\pi}=\id$.

It remains to prove the claim~\prettyref{eq:claim_gamma}, which reduces to proving
$\bar{\gamma}>1$ and $h(\bar{\gamma}) \ge \frac{3\log n}{(n-2)q^2}$ in view of the definition of $\gamma$ and the monotonicity of $h(\cdot)$ over $[1, \infty)$.
We divide the analysis into two cases depending on whether $q \le q_0$, where $q_0\equiv q_0(\epsilon)$ is some constant that only depends on $\epsilon$.

{\bf Case I: $q \le q_0$.} In view of $\eta \ge \epsilon/16$ as per~\prettyref{eq:eta_lower},  $s=q+\rho(1-q)$,  and $\rho \ge \epsilon$,
we have
$$
\bar{\gamma} \ge  \frac{\epsilon s}{16q} \left(1- \frac{\epsilon}{16} \right )
\ge \frac{\epsilon}{16} \left(1- \frac{\epsilon}{16} \right ) \left( 1+ \frac{\rho(1-q)}{q}  \right) >\exp\left(1+48/\epsilon\right),
$$
where the last inequality holds by choosing $q_0(\epsilon)$ to be a sufficiently small constant. 
It follows that 
$$
h\left( \bar{\gamma} \right) 
\ge \bar{\gamma} \log(\bar{\gamma}) - \bar{\gamma} 
\ge \bar{\gamma} \left(1+48/\epsilon\right) - \bar{\gamma} 
= 48 \bar{\gamma} /\epsilon 
\ge \frac{3 s}{ q} \left(1- \frac{\epsilon}{16} \right ) 
 \ge \frac{3\log n}{(n-2)q^2},
$$
where the last inequality holds in view of $nqs \ge (1+\epsilon)\log n$.

{\bf Case II: $q > q_0$.} In view of~\prettyref{eq:eta_lower} with $qs$ in place of $p$,  we have
$$
\eta \ge 1- \sqrt{ \frac{2(1+\epsilon/8)\log n}{(1-\epsilon/16)nqs}} \ge 1- \frac{\epsilon}{16},
$$
where the last inequality holds for all sufficiently large $n$.
Moreover, in view of $q \le 1/2$ and $\rho \ge \epsilon$, we have $s/q=1+\rho(1-q)/q  \ge 1+\epsilon$. Therefore, 
$$
\bar{\gamma} \ge \left( 1- \frac{\epsilon}{16} \right)^2 (1+\epsilon) \ge \left( 1- \frac{\epsilon}{8}\right) (1+\epsilon) 
\ge 1+ \frac{\epsilon}{2}. 
$$
Using the fact that $h(x) \ge (x-1)^2/3 $ when $x \in [1,2]$
and $h(x)$ is increasing over $[1, +\infty)$, we have
$$
h\left( \bar{\gamma} \right) \ge 
h \left( 1+ \frac{\epsilon}{2} \right)
\ge \frac{\epsilon^2}{12} \ge  \frac{3\log n}{(n-2)q^2},
$$
where the last inequality holds for  all sufficiently large $n$
as $q \ge q_0$.



Finally, we bound the time complexity of~\prettyref{alg:recovery}. 
In each execution of the while-loop in~\prettyref{alg:recovery},
to update $N_{\tilde{\pi}}(i,j)$ for all $i,j$, we just need to consider the newly added seed $u \in J$ and increase $N_{\tilde{\pi}}(i,j)$  by $1$ for every  $(i,j)$ pair, where
$i$ is connected to $u$ in graph $A$ and $j$ is connected to $\tilde{\pi}(u)$ in graph $B$. 
Hence, the total complexity of updating $N_{\tilde{\pi}}(i,j)$ for all $i,j$ for a given seed $u$ is $O(d_A(u) d_B(u))$,
where $d_A(u)$ and $d_B(u)$ are the degrees of $u$ in $A$ and $B$, respectively. 
Summing over all $n$ possible seeds, the total  time complexity is $O\left( 
\sum_u d_A(u) d_B(u) \right)$.
Finally, since $nq \ge \log n$, with high probability,
the max degree in $A$ and $B$ is $O(nq)$
and hence the total complexity is $O(n^3q^2)$.
%
%
\end{proof}





\section*{Acknowledgment}
The authors are grateful for the hospitality and the support of the Simons Institute for the Theory of Computing at the University of California, Berkeley, where part of this work was carried out during the program on ``Computational Complexity of Statistical Inference'' in Fall 2021.
\appendix

\section{Auxiliary results}
\label{app:pre}
The following lemma computes the cross-moments of $\bar{A}_{uv}$ and $\bar{B}_{\pi(u)\pi(v)}$ from the centered adjacency matrices.
\begin{lemma}[{\cite[Lemma 5]{mao2021testing}}]\label{lmm:cond-exp}
Let $(A,B)\sim \calG(n,q,\rho)$.
        Assume $q \le \frac{1}{2}$.
        For any 
        $0 \le \ell, m \le 2$ with $2 \le \ell + m \le 4$, 
        \begin{align}
            \expect{\sigma^{-\ell-m} \bar{A}_{uv}^\ell\bar{B}_{\pi(u)\pi(v)}^m}= \begin{cases}
            \rho^{\indc{\ell=m=1}} & \ell+m=2\\
            \frac{\rho(1-2q)}{\sqrt{q(1-q)}} & \ell+m=3\\
            \frac{q(1-q)+\rho(1-2q)^2}{q(1-q)} & \ell+m =4 
            \end{cases} \,.  \label{eq:cross_moment_equal}
        \end{align}
    Moreover, 
        \begin{align}
              \left|\expect{\sigma^{-\ell-m} \bar{A}_{uv}^\ell\bar{B}_{\pi(u)\pi(v)}^m} \right| \le |\rho|^{\indc{\ell=m=1}}\indc{\ell+m=2}  + \sqrt{\frac{1}{q}} \indc{\ell+m=3}  +\frac{1}{q} \indc{\ell+m =4} \,. \label{eq:cross_moment_bound}
        \end{align}
\end{lemma}

\begin{lemma}
\label{lmm:enum}
 The total number of unlabeled connected graphs 
 with $v$ vertices and excess $k\geq -1$ is at most
 \[
|\calJ(v-1)| v^{2k+2},
 \]
 where $|\calJ(K)|$ denotes the number of unlabeled trees with $K$ edges.
\end{lemma}
\begin{proof}
The total number of such unlabeled graphs is at most 
\begin{align}
    |\calJ(v-1)|{\binom{\binom{v}{2}}{k+1}}\le |\calJ(v-1)| v^{2(k+1)} \,,  \label{eq:count_non_tree}
\end{align}
where we first enumerate all possible
isomorphic classes of spanning trees with $v-1$ edges and then 
add $k+1$ more edges out of all possible $\binom{v}{2}$ node pairs. 
\end{proof}

\section{A data-driven choice of the threshold $\tau$}\label{app:unknown_rho}

In this section, we describe a data-driven approach to choose threshold $\tau$ in~\prettyref{alg:GMCC} without the knowledge of $q$ and $\rho$. For each $i\in [n]$, 
let $\psi(i)$ denote one of the maximizer of $ \Phi_{ij} $ over all $j\in [n]$. 
 Let $k$ denote the corresponding node such that $\Phi_{k\psi(k)} $ is the median of $\{ \Phi_{i\psi(i)}: i \in [n]\}$. Set $\hat{\tau} = \frac{1}{2}\Phi_{k\psi(k)}$. 
We claim that 
$ \frac{1}{2}c \mu \le \hat{\tau} \le \frac{1}{2} (2-c)\mu$
for any constant $0<c<1$ with probability $1-o(1)$ when $nq=\omega(1)$
and with probability $1-3\delta$
for any constant $\delta \in (0,1)$
when $nq\ge C(\epsilon,\delta)$.
Hence by \prettyref{thm:Phi_ij_almost}, $|I|=(1-o(1))n$ with probability $1-o(1)$ in the former case
and 
$\expect{|I|} \ge (1-3\delta)(1-\delta) n \ge (1-4 \delta)$ in the latter case.


It remains to show the claim, which reduces to proving 
$c \mu \le \Phi_{k\psi(k)}\le (2-c)\mu$.
Without loss of generality, we assume $\pi = \id$. 
Let 
\[
J=\left\{  i \in [n]: i \in  \arg \max_{j} \Phi_{ij} \text{ and }  c \mu  \le \Phi_{ii} \le  (2-c)\mu \right\} \,. 
\]
Recall that $F= \{ i: | \Phi_{ii} - \mu | > (1-c)\mu  \} $ as defined in \prettyref{eq:set_F}. By \prettyref{eq:i_neq_j_union}, with probability at least $1-o(1)$, $\Phi_{ij} < c\mu$ for all $ i \neq j$ and hence  $J = [n] \backslash F$.
Moreover, we have $\expect{|F|} \le \gamma n $, where $\gamma$ is given in \prettyref{eq:gamma_F}. By Markov's inequality, 
$
\prob{ |F| \ge  n/3 } \le 3 \gamma \,.
$
Note that $\gamma=o(1)$ if $nq = \omega(1)$, and $\gamma< \delta$ for any constant $\delta \in (0,1)$ if $nq\ge C(\epsilon,\delta)$.
Hence, we have $|J| \ge 2n/3$ with probability $1-o(1)$ if $nq = \omega(1)$, and with probability $1-3\delta$ if $nq\ge C(\epsilon,\delta)$. 
Henceforth assume $|J| \ge  2n/3$. 
If $\Phi_{k \psi(k)} > (2-c)\mu$, then there are at least $n/2$ nodes $i$ with $ \Phi_{i\psi(i)} >  (2-c)\mu$, contradicting $|J| \ge  2n/3$. Analogous argument holds for the case of $\Phi_{k \psi(k)} < c\mu$. Thus, we must have $ c\mu \le \Phi_{k \psi(k)} \le (2-c)\mu$.

\bibliography{low_degree_ref}
\bibliographystyle{alpha}

\end{document}